\documentclass[aps,pra,groupedaddress,superscriptaddress,showpacs,letterpaper,nofootinbib,notitlepage,10pt]{revtex4-1}

\usepackage[
  pdfusetitle,
  colorlinks,
  linkcolor = blue,
  citecolor = blue,
  urlcolor = blue]{hyperref}

\usepackage{amsthm}
\usepackage{amsmath}
\usepackage{amssymb}
\usepackage{bm} % bold math
\usepackage{enumerate}
\usepackage{cleveref}
\usepackage{graphicx}
\usepackage{comment}
\usepackage[boxed,noline]{algorithm2e}
\newtheorem{theorem}{Theorem}
\newtheorem{lemma}[theorem]{Lemma}
\newtheorem{corollary}[theorem]{Corollary}
\newtheorem{conjecture}[theorem]{Conjecture}
\crefname{conjecture}{conjecture}{conjectures}
\Crefname{conjecture}{Conjecture}{Conjectures}
\newtheorem{definition}[theorem]{Definition}

\newtheorem{example}[theorem]{Example}

\crefname{secinapp}{appendix}{appendices}
\Crefname{secinapp}{Appendix}{Appendices}

\newcommand{\abs}[1]{\left\lvert#1\right\rvert}
\newcommand{\onenorm} [1]{\ensuremath{\lVert#1\rVert_1}}
\newcommand{\twonorm} [1]{\ensuremath{\lVert#1\rVert_2}}
\newcommand{\infnorm} [1]{\ensuremath{\lVert#1\rVert_\infty}}
\newcommand{\pnorm}   [1]{\ensuremath{\lVert#1\rVert_p}}
\newcommand{\qnorm}   [1]{\ensuremath{\lVert#1\rVert_q}}
\newcommand{\trnorm} [1]{\ensuremath{\lVert#1\rVert_{\Tr}}}
\newcommand{\bra}[1]{\ensuremath{\left\langle{#1}\right|}}
\newcommand{\ket}[1]{\ensuremath{\left|{#1}\right\rangle}}
\newcommand{\projdyad}[1]{\ket{#1}\bra{#1}}
\newcommand{\braket}[2]{\ensuremath{\left\langle{#1}\middle|{#2}\right\rangle}}
\newcommand{\braopket}[3]{\ensuremath{\left\langle{#1}\middle|{#2}\middle|{#3}\right\rangle}}
\newcommand{\braopketsmall}[3]{\ensuremath{\langle{#1}|{#2}|{#3}\rangle}}
\newcommand{\hilb}[1]{\ensuremath{\mathcal{H}_{#1}}}

\newcommand{\Tr}[0]{\ensuremath{\textnormal{Tr}}}
\newcommand{\ot}[0]{\otimes}
\newcommand{\bigomic}[0]{\textnormal{O}}

\newcommand{\expect}[1]{\mathbb{E}\left[ #1 \right]}

\newcommand{\matabs}[1]{\ensuremath{\bar{#1}}}
\newcommand{\onesvec}[0]{\ensuremath{\mathbf{1}}}

\newcommand{\zeromat}[0]{\ensuremath{\mathbf{0}}}
\newcommand{\pathidx}[0]{{i_0 \dotsc i_S}}
\newcommand{\intfc}[0]{\mathcal{I}}
\newcommand{\intfs}[0]{\mathcal{J}}
\newcommand{\ifmax}[0]{\mathcal{I}_{\textnormal{max}}}
\newcommand{\DC}[0]{\mathcal{D}}
\newcommand{\myvec}[1]{\bm{#1}}
\newcommand{\vecprod}[2]{#1^\top #2}

\newcommand{\sgn}[0]{\textnormal{sgn}}
\newcommand{\poly}[0]{\textnormal{poly}}
\newcommand{\polylog}[0]{\textnormal{polylog}}
\newcommand{\aA}[0]{\matabs{A}}

\newcommand{\aAT}[0]{\matabs{A}^\top}

\newcommand{\ihat}[0]{\hat{\imath}}

\newcommand{\jvec}[0]{\myvec{j}}
\newcommand{\kvec}[0]{\myvec{k}}
\newcommand{\uvec}[0]{\myvec{u}}
\newcommand{\vvec}[0]{\myvec{v}}
\newcommand{\xvec}[0]{\myvec{x}}
\newcommand{\yvec}[0]{\myvec{y}}
\newcommand{\haarn}[0]{G_n}
\newcommand{\abshaarn}[0]{\matabs{G}_n}
\newcommand{\fval}[0]{f}
\newcommand{\bval}[0]{b}
\newcommand{\bmax}[0]{\bval_{\textrm{max}}}
\newcommand{\bopt}[0]{\bval_{\textrm{opt}}}

\newcommand{\epsp  }[2]{\textnormal{EPS}_p(#1,#2)}
\newcommand{\epstwo}[2]{\textnormal{EPS}_{2}(#1,#2)}
\newcommand{\epsinf}[2]{\textnormal{EPS}_{\infty}(#1,#2)}
\newcommand{\ehtp  }[2]{\textnormal{EHT}_p(#1,#2)}
\newcommand{\ehttwo}[2]{\textnormal{EHT}_{2}(#1,#2)}
\newcommand{\ehtinf}[2]{\textnormal{EHT}_{\infty}(#1,#2)}

\newcommand{\Pdist  }[1]{P  (#1)}
\newcommand{\Qdist  }[1]{Q  (#1)}
\newcommand{\Rdist  }[1]{R  (#1)}
\newcommand{\Wdist  }[1]{W  (#1)}
\newcommand{\Pdistsub}[2]{P_{#1}(#2)}
\newcommand{\Qdistsub}[2]{Q_{#1}(#2)}
\newcommand{\Rdistsub}[2]{R_{#1}(#2)}
\newcommand{\Rdistopt}[1]{\Rdistsub{\textrm{opt}}{#1}}

\begin{document}

\title{Quantum interference as a resource for quantum speedup}

\author{Dan Stahlke}
\email[Electronic address:]{dan@stahlke.org}
\affiliation{Department of Physics, Carnegie Mellon University,
	Pittsburgh, Pennsylvania 15213, USA}

\date{\today}

\begin{abstract}
	Quantum states can in a sense be thought of as generalizations of classical probability
	distributions, but are more powerful than probability distributions
	when used for computation or communication.  Quantum speedup therefore requires
	some feature of quantum states that classical probability distributions lack.  One such
	feature is interference.  We quantify interference and show that there can be no quantum
	speedup due to a small number of operations incapable of generating large amounts
	of interference (although large numbers of such operations can in
	fact lead to quantum speedup).
	Low-interference operations include sparse unitaries, Grover reflections, short time/low
	energy Hamiltonian evolutions, and the Haar wavelet transform.
	Circuits built from such operations can be classically simulated via a Monte Carlo
	technique making use of a convex combination of two Markov chains.
	Applications to query complexity, communication complexity, and the Wigner representation
	are discussed.
\end{abstract}

  %Quantum algorithms, protocols, and simulations

\maketitle

\tableofcontents

\section{Introduction}

It is well known that certain quantum algorithms, such as Shor's and Grover's, provide a
speedup compared to classical algorithms.
However, the source of such quantum speedup is still somewhat of a mystery.
Insight can be gained by determining necessary resources.
Suppose that any quantum circuit not making use of some resource $X$ can be efficiently
simulated.
Being efficiently simulated, such circuits do not exhibit quantum speedup.
One can then conclude that resource $X$ is necessary for quantum speedup.
Many such resources have been identified.
For circuits on pure states there is no quantum speedup if at all times (i.e.\ before and after
every unitary) the state has small Schmidt rank~\cite{PhysRevLett.91.147902} or factors into a
product state on small subsystems~\cite{Jozsa08082003}.
For qubit circuits there is no quantum speedup if the discord across all bipartite cuts is
zero at all times~\cite{arxiv:1006.4402}.
There is no quantum speedup
for circuits that use only Clifford gates~\cite{arxiv:quant-ph/9807006},
or matchgates~\cite{Valiant:2001:QCS:380752.380785}, that have small
tree width~\cite{Markov:2008:SQC:1405087.1405105,arxiv:quant-ph/0603163},
or that use only operations having nonnegative Wigner
representation~\cite{1367-2630-14-11-113011,1367-2630-15-1-013037,PhysRevLett.109.230503}.
For a brief overview of resources identified as important for quantum speedup see section 9
of~\cite{arxiv:1206.0785}.

A tempting but naive explanation for quantum speedup is the exponentially large dimensionality
of Hilbert space ($2^n$ for $n$ qubits), combined with ``quantum parallelism''.
Shor's algorithm begins by preparing a state
$\sqrt{2^{-n}} \sum_x \ket{x}\ot\ket{f(x)}$ which can be interpreted as simultaneously
evaluating $f$ for all $2^n$ values of $x$.
However, this is not a satisfactory explanation for quantum speedup since classical probability
distributions over $n$ bits can also be considered as vectors of dimension $2^n$, and allow a
similar sort of parallelism.
We show that the quantum speedup is connected to \textit{interference}, something which
classical probability distributions lack.
Prior works have mentioned interference as being important for quantum speedup but without
offering a quantitative
definition~\cite{bennett:24,Fortnow2003597,PhysRevA.61.010301,vandennest2011}
or have quantified interference without providing a strong connection to
speedup~\cite{PhysRevA.73.022314}.

We consider quantum circuits composed of an initial state, followed by
several unitary operators, and terminated by measurement of a Hermitian observable.
The expectation value of this measurement can be written as a sum of Feynman-like paths
in the computational basis, and
this sum can be estimated via a Monte Carlo technique that considers an
ensemble of paths drawn according to a suitable probability distribution.
The required size of the ensemble is lower bounded by the square of the interference,
which we define as a sum of absolute values of the path amplitudes (\cref{def:intf_circuit}).
We are not able to reach this lower bound, however by using a convex combination of a pair of
Markov chains we are able to provide a simulation algorithm that runs in time quadratic in
the product of the
\textit{interference producing capacities} of each operator in the circuit,
defined as the largest amount of interference an operator
is capable of producing (\cref{def:ifmax}).
This ends up being equal to the largest singular value of the entrywise absolute value
of the operator in the computational basis.
Briefly, we can estimate expressions of the form $\braopket{\psi}{A \dotsm Z}{\phi}$,
of which quantum circuits
$\braopketsmall{\psi}{U^{(1)\dag} \dotsm U^{(T)\dag} M U^{(T)} \dotsm U^{(1)}}{\psi}$
are a special case, in time proportional to
$\twonorm{\matabs{A}}^2 \dotsm \twonorm{\matabs{Z}}^2$
where $\twonorm{\cdot}$ denotes maximum singular value and where a bar over an operator denotes
entrywise absolute value in the computational basis.
This work was inspired by, and extends,~\cite{vandennest2011} which provides an efficient
simulation when $A,\dotsc,Z$ are all sparse.

Operations with small interference producing capacity include the
\textit{efficiently computable sparse} operations as defined in~\cite{vandennest2011}
(e.g.\ permutation matrices and gates acting on a constant number of qubits), as well as
the Grover reflection operation, short time/low energy Hamiltonian evolutions, and the Haar
wavelet transform.  Our simulation algorithm will generally be exponentially slow
in the length of the circuit, but for the classes of gates listed in the previous sentence has
only polynomial dependence on the number of qubits.
An example of a circuit that apparently uses much ``quantum magic,'' but
which can nevertheless be simulated in time polynomial in the number of qubits, is depicted in
\cref{fig:big_circuit}.

We (of course) cannot efficiently simulate Shor's algorithm.  However, replacing the Fourier
transform by the Haar transform, which has low interference producing capacity, yields a
circuit that we can simulate (\cref{fig:shor}).
We show that there is no quantum advantage for communication
protocols that use small interference, although curiously this result does not apply to
one-round communication protocols.
To our knowledge, interference producing capacity is the first continuous-valued quantity that
has been shown necessary for quantum speedup, escaping the theorem of~\cite{arxiv:1204.3107}
which shows that a large class of continuous-valued quantities, such as entanglement and
discord, are not necessary for quantum speedup.

In sections~\ref{sec:monte} and~\ref{sec:markov} we explain our method for estimating
expectation values using a Monte Carlo technique with Markov chains.
In section~\ref{sec:eps_eht} we formalize and extend this technique and provide guarantees on runtime.
In section~\ref{sec:circuits} we characterize the types of quantum circuits that our technique can
efficiently simulate, and explore a variety of circuits that we cannot efficiently
simulate.
Section~\ref{sec:app} discusses further applications, including the Wigner representation and
communication complexity.
In section~\ref{sec:conj} we formalize our conjecture that interference, rather than interference
producing capacity, is required for quantum speedup.
Nontrivial proofs are deferred to appendices.

\section{Monte Carlo technique}
\label{sec:monte}

\subsection{Sampling of paths}

We will make use of the following circuit model.  Let $\rho$ be an initial density operator.
This state is acted upon by a sequence of unitaries $U^{(1)}, \dotsc, U^{(T)}$.
Finally, a Hermitian observable (e.g.\ a projector) $M$ is measured.
It is not assumed that the unitary operations or the final observable are local, they can be
arbitrary operations potentially involving all qubits or qudits (e.g.\ a quantum
Fourier transform).
The expectation value of this final measurement is
\begin{equation}
	\label{eq:expect_val}
	\Tr\left\{U^{(1)\dag} \dotsm U^{(T)\dag} M U^{(T)} \dotsm U^{(1)} \rho \right\}.
\end{equation}
Our goal is to estimate this expectation value to within small additive error, using a
classical computer.
We allow the unitaries to be oracle operations (as in Grover's algorithm), in which case we
grant the classical computer that runs the simulation access to an equivalent oracle (this is
further discussed in \cref{subsec:query}).

This is not the most general type of simulation.
In particular, we do not consider the case of a many-outcome measurement (e.g.\ individual
measurements on several qubits, or a measurement given by a projective decomposition
of the identity) in which the simulation is required to produce individual outcomes according
to the same probability distribution with which the quantum circuit produces those outcomes.
The ability to estimate the expectation value of a projector to within small multiplicative
error would allow simulation of such sampling, as discussed in \cite{terhal2004},
however the algorithm of the present paper only estimates to within additive error.

Although our primary goal is to estimate expressions of the form~\eqref{eq:expect_val}, we
generalize the task by considering products of the form
$\Tr\{A^{(1)} \dotsm A^{(S)} \sigma \}$
where $\sigma$ and the $A^{(s)}$ are matrices, not necessarily unitary or Hermitian, and possibly
rectangular (we label $\sigma$ separately from the $A^{(s)}$ in anticipation of the results of
the next section).
This product can be written as a sum over paths,
\begin{equation}
	\label{eq:sum_over_paths}
	\Tr\{A^{(1)} \dotsm A^{(S)} \sigma \}
	= \sum_\pathidx
	A^{(1)}_{i_0 i_1} \dotsm A^{(S)}_{i_{S-1}i_S} \sigma_{i_S i_0}.
\end{equation}
Or, by defining the tuple index $\pi = (\pathidx)$, this can be written as
\begin{align}
	& \Tr\{A^{(1)} \dotsm A^{(S)} \sigma \} = \sum_\pi V(\pi)
	\\
	& V(\pi) = A^{(1)}_{i_0 i_1} \dotsm A^{(S)}_{i_{S-1}i_S} \sigma_{i_S i_0}.
\end{align}
Our strategy is to estimate this sum by drawing a reasonably small number of paths $\pi$
according to a probability distribution, denoted $\Rdist{\pi}$.
Any probability distribution can be used, although some are more suitable than others.
Finding a good $\Rdist{\pi}$ will be a central goal of this section and the next.
Denote by $\Pi$ a random variable that takes value $\pi$ with probability $\Rdist{\pi}$.
Consider the expectation value of $V(\Pi) / \Rdist{\Pi}$.
\begin{align}
	\expect{\frac{V(\Pi)}{\Rdist{\Pi}}}
	&= \sum_\pi \frac{V(\pi)}{\Rdist{\pi}} \Rdist{\pi}
	\\ &= \sum_\pi V(\pi).
\end{align}

By the weak law of large numbers, $\sum_\pi V(\pi)$ can be approximated to arbitrary accuracy by
computing the mean of sufficiently many samples of $V(\Pi) / \Rdist{\Pi}$, however the
efficiency of this strategy hinges on two things.
First, it must be possible using a classical computer to efficiently draw random samples
according to the probability distribution $\Rdist{\pi}$ and to compute the corresponding values
$V(\pi) / \Rdist{\pi}$.
This is an important point that we will return to throughout the paper.
Second, the sample mean of $V(\Pi) / \Rdist{\Pi}$ must rapidly converge to its expectation
value.
The Chernoff-Hoeffding bound states that for a random
variable whose magnitude is bounded by $\bval$, the mean of $\bigomic(\epsilon^{-2} \bval^2)$
samples is very likely to approximate the expectation value to within additive error
$\epsilon$.
Thus there is rapid convergence when $\max_\pi \{ \abs{V(\pi)}/\Rdist{\pi} \}$ is small.
Note that this is a sufficient but not necessary condition for rapid convergence, for example
considering the variance of $V(\Pi)/\Rdist{\Pi}$ could in some cases reveal that convergence
happens more rapidly.

We now present the Chernoff-Hoeffding bound in one of its standard forms, along with a
corollary that adapts it to our application.

\begin{theorem}[Chernoff-Hoeffding bound~\cite{Hoeffding1963}]
	\label{thm:chernoff1}
	Let $X_1, \dotsc, X_K$ be independent identically distributed real-valued random variables
	with expectation value $\expect{X}$ and satisfying $\abs{X_k} \le \bval$.
	Let $\epsilon > 0$.  Then
	\begin{equation}
		% NOTE: Van den Nest has a 4 in the denominator, but I believe it should be a 2.
		\Pr\left\{\abs{ \frac{1}{K} \sum_{k=1}^K X_k - \expect{X} } > \epsilon
		\right\} \le 2 e^{-K \epsilon^2 / 2 \bval^2}.
		\label{eq:ch_bound}
	\end{equation}
\end{theorem}

\begin{corollary}
	\label{thm:chernoff3}
	Let $V(\pi)$ be a complex valued function of $\pi$
	and $\Rdist{\pi}$ be a probability distribution.
	Define
	\begin{equation}
		\bmax = \max_\pi \left\{ \frac{\abs{V(\pi)}}{\Rdist{\pi}} \right\}.
		\label{eq:bmax}
	\end{equation}
	Let $\epsilon, \delta > 0$.
	Then, with probability less than $\delta$ of exceeding the error bound,
	$\sum_\pi V(\pi)$ can be estimated to within additive error $\epsilon$
	using $\bigomic(\log(\delta^{-1}) \epsilon^{-2} \bmax^2)$ draws from the distribution
	$\Rdist{\pi}$ and the same number of evaluations of $V(\pi)/\Rdist{\pi}$.
\end{corollary}
\begin{proof}
	It can be shown\footnote{
		This is shown by applying~\cref{thm:chernoff1} separately to the real and imaginary
		parts and using the fact that the sample mean is within additive error $\epsilon$ of
		the expectation value as long as both the real and imaginary parts are within
		$\epsilon/\sqrt{2}$.
	}
	that~\cref{thm:chernoff1} can be extended to complex variables at the expense
	of replacing the right hand side of~\eqref{eq:ch_bound} by
	$4 e^{-K \epsilon^2 / 4 \bval^2}$.
	Define the independent identically distributed random variables
	$X_k = V(\Pi_k) / \Rdist{\Pi_k}$ with $k \in \{1,\dotsc,K\}$.
	Applying the complex valued version of~\cref{thm:chernoff1}, and noting that
	$\abs{X_k} \le \bmax$ and $\expect{ V(\Pi) / \Rdist{\Pi} } = \sum_\pi V(\pi)$, we get
	\begin{equation}
		\label{eq:chernoff_VP}
		\Pr\left\{\abs{ \frac{1}{K} \sum_{k=1}^K \frac{V(\Pi_k)}{\Rdist{\Pi_k}} -
				\sum_{\pi} V(\pi) } > \epsilon
		\right\} \le 4 e^{-K \epsilon^2 / 4 \bmax^2}.
	\end{equation}
	Setting $K = \ln(4/\delta) 4 \epsilon^{-2} \bmax^2 = \bigomic(\log(\delta^{-1})
	\epsilon^{-2} \bmax^2)$ makes the right hand side of~\eqref{eq:chernoff_VP} equal to
	$\delta$.
	%The estimation algorithm consists of drawing
	%$\bigomic(\log(\delta^{-1}) \epsilon^{-2} \bval^2)$ samples of $\pi$
	%from the distribution $\Rdist{\pi}$ and averaging the corresponding $V(\pi)/\Rdist{\pi}$.
\end{proof}

Since the number of samples needed depends only logarithmically on $\delta$, it is possible to
choose $\delta$ to be extremely small (say, one part in a billion) while having only minimal
impact on the number of samples needed.  With such a small $\delta$, the estimate will be very
likely to be within additive error $\epsilon$.

The number of samples needed for an accurate estimate is quadratic in $\bmax$, so finding an
$\Rdist{\pi}$ for which $\bmax$ is small is of crucial importance.
However, feasibility of the simulation also depends on the difficulty of drawing random paths
$\pi$ according to the distribution $\Rdist{\pi}$ and computing the corresponding values
$V(\pi)/\Rdist{\pi}$.
We will denote by the letter $\fval$ the time needed to carry out these operations.  Specifically,
we require that sampling from $\Rdist{\pi}$ and computing $V(\pi)/\Rdist{\pi}$ can be carried
out in average time $\bigomic(\fval)$ where $\fval$ is some function of the dimension or number
of qubits of a quantum circuit.
%\footnote{
%	More accurately, one should think of an infinite family of quantum circuits (perhaps each
%	one having more qubits than the last, although this is not the only possible scenario),
%	with $\fval$ being some function of the index into this family.
%}
Since $\sum_\pi V(\pi)$ can be estimated by averaging
$\bigomic(\log(\delta^{-1}) \epsilon^{-2} \bmax^2)$
samples of $V(\Pi) / \Rdist{\Pi}$, each of which can be computed in time
$\bigomic(\fval)$, the total runtime of the algorithm is
$\bigomic(\log(\delta^{-1}) \epsilon^{-2} \bmax^2 \fval)$.

Some probability distributions are easier to sample from than others, and this needs to be
decided on a case by case basis.
For example, consider $\Rdist{i} = \abs{\psi_i}^2$ where $\ket{\psi}$ is a quantum state.
If $\ket{\psi}$ is a computational basis state then $\Rdist{i}$ is rather
trivial and can be sampled by simply outputting the sole index $i$ for which
$\Rdist{i} \ne 0$.
If $\ket{\psi}$ is a graph state on $n$ qubits then $\Rdist{i}$ is the uniform distribution
over the $2^n$ basis states.  This can be sampled in time $\bigomic(n)$ by tossing a
fair coin $n$ times, once for each qubit, so in this case $\fval=n$.
On the other hand, if $\ket{\psi}$ is defined as being the state just before the final
measurement in Shor's algorithm, then it is probably not feasible to sample from $\Rdist{i}$
efficiently on a classical computer.

For simplicity we will assume that all operations can be carried out with perfect
computational accuracy, including the degree to which the probability distribution of the
generated samples $\pi$ agrees with an ideal distribution $\Rdist{\pi}$, and the
precision of the computed $V(\pi)/\Rdist{\pi}$ values.
Of course, computers can only compute with finite precision.
However, since we are
concerned only with approximating expectation values to within additive error $\epsilon$,
carrying out the computations to finite but high precision is sufficient as long as the total
accumulated computational error is small compared to the error tolerance $\epsilon$.  This is
discussed in more detail in appendix~A of \cite{vandennest2011}.

\subsection{Interference}
\label{subsec:intf}

An efficient simulation requires choosing a probability distribution $\Rdist{\pi}$ for
which $\bmax$ of~\eqref{eq:bmax} is not large.
A tempting choice is
\begin{equation}
	\label{eq:holistic_dist}
	\Rdistopt{\pi} := \frac{ \abs{ V(\pi) } }{ \sum_{\pi'} \abs{ V(\pi') } }.
\end{equation}
It can be shown\footnote{
	Let $\Rdist{\pi}$ be any probability distribution that differs from $\Rdistopt{\pi}$
	of~\eqref{eq:holistic_dist}.  Then there must be a $\pi'$ such that
	$\Rdist{\pi'} < \Rdistopt{\pi'}$.  It follows that
	$\max_\pi \{ \abs{V(\pi)}/\Rdist{\pi} \}
		> \abs{V(\pi')}/\Rdistopt{\pi'} = \sum_\pi \abs{V(\pi)}$.
} that this is the unique distribution yielding the minimum possible value of $\bmax$,
\begin{align}
	\label{eq:circuit_abs_sum}
	\bopt &= \sum_{\pi} \abs{ V(\pi) }.
\end{align}
Being lowest possible value of $\bmax$,~\eqref{eq:circuit_abs_sum} represents a
lower bound on the number of samples needed as guaranteed by the Chernoff-Hoeffding bound,
although a more careful analysis of variances, for instance, could show that the algorithm
actually produces a faster than expected convergence.

An efficient algorithm requires both that $\bmax$ be small and that $\Rdist{\pi}$ can
be sampled from efficiently.
We do not know of a way to efficiently sample from the probability
distribution~\eqref{eq:holistic_dist} in general, so this is not useful for
computing the expectation value.
Nevertheless, it is worthwhile to discuss for a moment the case where the one condition is met
(small $\bmax$) even if the other condition is not met (ability to efficiently draw samples).
For concreteness,
consider a simple quantum circuit with only one unitary, $\Tr\{U^\dag M U \rho \}$.
This can be written as a sum over paths
\begin{equation}
	\label{eq:circuit_sum}
	\Tr\{U^\dag M U \rho \} = \sum_\pi V(\pi)
\end{equation}
with $\pi=(i,j,k,l)$ and $V(i,j,k,l) = U^\dag_{ij} M_{jk} U_{kl} \rho_{li}$.
Plugging this into~\eqref{eq:circuit_abs_sum} gives
\begin{equation}
	\label{eq:UMUrho_cost_optimal}
	\bopt = \Tr\{\matabs{U}^\dag \matabs{M} \matabs{U} \matabs{\rho}\}
\end{equation}
where a bar over a vector or matrix denotes entrywise absolute value in the computational
basis, a notation that will be used throughout this paper.
This generalizes in the obvious way for circuits with more than one unitary.

Comparing~\eqref{eq:circuit_abs_sum} and~\eqref{eq:circuit_sum}, both are sums
over paths but the latter involves an absolute value for each path.
The sum~\eqref{eq:circuit_sum} has magnitude bounded by 1 if the observable $M$ has
eigenvalues bounded in magnitude by 1.
The sum~\eqref{eq:circuit_abs_sum} on the other hand can take a much larger value
than~\eqref{eq:circuit_sum} when the terms in the latter sum exhibit cancellations due to
destructive interference.
For example, consider the case $\ket{\psi} = N^{-1/2} \sum_i \ket{i}$, $U$ the Fourier
transform, and $M$ the identity, giving $\bopt=\sqrt{N}$.

It may be enlightening to consider a physical example.
To this end, we introduce a simple toy-model version of Young's double-slit experiment.
Let states $\ket{0}$ and $\ket{1}$ represent a particle immediately exiting the upper and
lower slits, respectively, and let $\ket{x}$ represent a particle impacting the detector
at position $x$.
The transfer operator representing passage of the particle from the slits to
the detector will be some unitary $U$ satisfying
$U(\alpha\ket{0} + \beta\ket{1}) = \int_x (\alpha\psi_x + \beta\phi_x)\ket{x} dx$.
A particle passing through the upper slit will impact the detector at position $x$ with
probability density $\abs{\psi_x}^2$; for a particle passing through the lower slit the
probability density is $\abs{\phi_x}^2$.
A particle in a superposition of passing through upper and lower slits,
in state $\ket{+}=(\ket{0}+\ket{1})/\sqrt{2}$, will impact the
screen at $x$ with probability density
\begin{align}
	\abs{\frac{1}{\sqrt{2}} \psi_x + \frac{1}{\sqrt{2}} \phi_x}^2 =
		\frac{1}{2} \abs{\psi_x}^2 + \frac{1}{2} \abs{\phi_x}^2 +
		\textrm{Re}(\psi_x^* \phi_x).
\end{align}
The first two terms on the right hand side represent the probability that would be
expected if the particle were in a classical stochastic mixture of passing through one slit
or the other.  The third is the interference term.  Integrating this term over $x$
yields zero, as it must in order for the probabilities to sum to $1$.  The total
amount of interference can be quantified by instead integrating the absolute value of this
term.
Similarly, if we were interested in only part of the detector, say $x \in [0,1]$, the
interference associated with that region could be defined by integrating only over this range.
It turns out to be more mathematically convenient to include all three terms in
the definition of interference; for one thing the resulting quantity will be
multiplicative when considering a system composed of non-interacting subsystems.
The $\abs{\psi_x}^2/2 + \abs{\phi_x}^2/2$ terms contribute at
most $1$ (exactly $1$ if integrating over the entire range).
In summary, we may define the interference associated with the $x \in [0,1]$ region of the
detector as
\begin{align}
	\intfc = \int_{x \in [0,1]} \left( \frac{1}{2} \abs{\psi_x}^2 + \frac{1}{2} \abs{\phi_x}^2 +
		\abs{\psi_x^* \phi_x} \right) dx.
\end{align}
This is essentially what is done in~\eqref{eq:UMUrho_cost_optimal}.
Specifically, setting $\rho = \ket{+}\bra{+}$ and $M=\int_{x \in [0,1]} \ket{x}\bra{x} dx$
in~\eqref{eq:UMUrho_cost_optimal} yields
\begin{align}
	\bopt &= \int_{x \in [0,1]} \left( \frac{1}{\sqrt{2}} \abs{\psi_x} +
		\frac{1}{\sqrt{2}} \abs{\phi_x} \right)^2 dx
	\\ &= \int_{x \in [0,1]}
		\left( \frac{1}{2} \abs{\psi_x}^2 + \frac{1}{2} \abs{\phi_x}^2 +
		\abs{\psi_x^* \phi_x} \right) dx.
\end{align}
Note that~\eqref{eq:UMUrho_cost_optimal} depends upon the choice of basis since the entrywise
absolute value is basis dependent.  Typically one has some canonical basis in
mind, for example when one says that the double slit experiment exhibits interference this
is relative to the position basis.
For quantum circuits there is the computational basis, although in the interest of
efficient simulation one may choose to use some other basis.

For a more complicated apparatus, such as a network of beam splitters, similar arguments
apply: we quantify interference by computing a sum over paths, summing the absolute value
of each path contribution.
This definition depends upon a choice of course graining.  For instance, a box
which simply passes a photon from input to output undisturbed could be said to contribute
no interference.  On the other hand, if one were to take a more detailed view of this
box---suppose for example that it contains a perfectly balanced Mach-Zehnder
interferometer---then one could conclude that there is in fact interference.
The same applies to simulation of quantum circuits.  Although our simulation technique has
difficulty simulating the Fourier transform, a Fourier transform followed by its inverse
presents no difficulty if one course grains the circuit by replacing $F^\dag F$ by the
identity.

The above considerations lead to the following definition.

\begin{definition}
	\label{def:intf_circuit}
	The \textit{interference} of a quantum circuit
	with initial state $\rho$, unitary operators $U^{(1)}, \dotsc, U^{(T)}$,
	and measurement $M$
	%$\Tr\{U^{(1)\dag} \dotsm U^{(T)\dag} M U^{(T)} \dotsm U^{(1)} \rho \}$
	is
	\begin{equation}
		\label{eq:intfc_circuit_def}
		\intfc\left(
			U^{(1)\dag},\dotsc,U^{(T)\dag},M,U^{(T)},\dotsc,U^{(1)},\rho
		\right) =
			\Tr\left\{
				\matabs{U}^{(1)\dag} \dotsm \matabs{U}^{(T)\dag}
				\matabs{M}
				\matabs{U}^{(T)} \dotsm \matabs{U}^{(1)}
				\matabs{\rho}
			\right\}.
	\end{equation}
	More generally, the interference of an arbitrary expression of the form
	$\Tr\{ A^{(1)} \dotsm A^{(S)} \sigma \}$ is
	\begin{equation}
		\intfc( A^{(1)}, \dotsc, A^{(S)}, \sigma ) =
			\Tr\{ \matabs{A}^{(1)} \dotsm \matabs{A}^{(S)} \matabs{\sigma} \}.
	\end{equation}
	This definition depends on the choice of basis.  Unless otherwise specified the standard
	(a.k.a.\ computational) basis is used.
\end{definition}

With this definition, we have that $\bmax \ge \intfc(U^\dag,M,U,\rho)$
in~\eqref{eq:bmax} for any choice of probability distribution, with equality
when the distribution~\eqref{eq:holistic_dist} is used.
Since the number of samples needed to estimate the expectation value using our technique is
proportional to $\bmax^2$, any quantum circuit with very large interference could never
feasibly be simulated with our technique, no matter the choice of $\Rdist{\pi}$.

While we don't know how to efficiently sample from the optimal probability
distribution~\eqref{eq:holistic_dist}, we conjecture that there is still some way to
efficiently estimate the expectation value of a quantum circuit in cases where the interference
is low.
The precise statement of this conjecture is a delicate matter taken up in section~\ref{sec:conj}.
We will however show, by the end of the next section, that it is possible to
simulate circuits in which each unitary as well as the final observable has a low
\textit{interference producing capacity} (\cref{def:ifmax}).

A connection between $\intfc$ and the decoherence functional of Gell-Mann and Hartle is
discussed in \cref{subsec:decoherence}.

\section{Markov chains}
\label{sec:markov}

\subsection{Introduction}

The problem with the probability distribution~\eqref{eq:holistic_dist}  is that there is no
obvious way to efficiently sample from it using a classical computer.
So while only $\bigomic(\log(\delta^{-1}) \epsilon^{-2} \intfc^2)$ samples are needed (with
$\intfc$ given by \cref{def:intf_circuit}), each sample may be very complicated to evaluate.
The essence of the difficulty is that this distribution treats the circuit
holistically, so drawing samples apparently requires an understanding of how
all the factors of~\eqref{eq:sum_over_paths} interact with each other.
In order to avoid this problem we instead use a probability distribution defined in terms of a
time-inhomogeneous Markov chain with a transition corresponding to each operator
in~\eqref{eq:sum_over_paths}.
More precisely, we take the convex combination of two (unrelated) Markov chains, one proceeding
left-to-right and the other proceeding right-to-left.
This way, it is only necessary to understand each individual operator, not the interactions
between operators.
The computation time of this simulation will end up being related not to the interference
$\intfc$ but rather the product of the interference producing capacities of each factor (a term
that will be defined at the end of this section).

The end result of this section will be an algorithm for estimating products of the form
$\Tr\{A^{(1)} \dotsm A^{(S)} \sigma\}$
where $\sigma$ and the $A^{(t)}$ are matrices, not necessarily unitary or Hermitian and possibly
rectangular.
This includes as a special case quantum circuits of the form
$\Tr\{U^{(1)\dag} \dotsm U^{(T)\dag} M U^{(T)} \dotsm U^{(1)} \rho\}$.
We build the algorithm step by step, considering first an example that
demonstrates why a convex combination of probability distributions is
needed, second an example that explains how the Markov chains are built, and finally
using a convex combination of Markov chains.
The exposition in this section is meant to be instructive; formal theorems will be taken up in
section~\ref{sec:eps_eht}.

\subsection{Inner product}
\label{subsec:dotprod}

Consider the task of estimating the inner product
$\braket{\psi}{\phi} = \sum_i \psi_i^* \phi_i$
where the two vectors satisfy the property $\pnorm{\psi}=\qnorm{\phi}=1$ with
$1/p+1/q=1$.\footnote{
	The $\ell^p$-norm, $\pnorm{\cdot}$, is defined as $\pnorm{\psi} = \left( \sum_i \abs{\psi_i}^p
	\right)^{1/p}$ when $1 \le p < \infty$ and $\pnorm{\psi} = \max_i \abs{\psi_i}$ when
	$p=\infty$.
	When $1/p+1/q=1$, the norms $\pnorm{\cdot}$ and $\qnorm{\cdot}$ are dual to each other.
}
In the context of quantum circuits $p=q=2$ is the natural choice; however, we
allow general $\ell^p$-norms because the case $p=1$, $q=\infty$ is also
important and because the general case may be of independent interest.
Here, as in the more general case that will follow, the key is to find a probability
distribution $\Rdist{i}$ that will be suitable for application of \cref{thm:chernoff3}.
It is needed that
\begin{equation}
	\bmax
	= \max_i \left\{ \frac{ \abs{V(i)} }{ \Rdist{i} } \right\}
	= \max_i \left\{ \frac{ \abs{\psi_i^* \phi_i} }{ \Rdist{i} } \right\}
	\label{eq:psiphiR_cost}
\end{equation}
is not large.
There are two obvious choices for the probability distribution:
$\Pdist{i} = \abs{\psi_i}^p$ and $\Qdist{i} = \abs{\phi_i}^q$.
Unfortunately, neither of these will guarantee a small $\bmax$.
However, for each $i$ at least one of the distributions $\Pdist{i}$ or $\Qdist{i}$ will work well.
The solution is to take a convex combination of these two distributions,
\begin{equation}
	\Rdist{i} = \frac{1}{p} \Pdist{i} + \frac{1}{q} \Qdist{i}.
\end{equation}
The algorithm that follows is an adaptation of one that appears in~\cite{vandennest2011}
(they used $p=q=2$ and a slightly different technique).
We present it as a formal theorem, in order to demonstrate how to carefully track
the algorithm's time complexity.

\begin{example}
	Let $1 \le p \le \infty$ and $1/p + 1/q = 1$.
	Let $\ket{\psi}$ and $\ket{\phi}$ be vectors with $\pnorm{\psi}=\qnorm{\phi}=1$.
	Suppose that it is possible to sample from the probability distributions
	$\Pdist{i} = \abs{\psi_i}^p$ and $\Qdist{i} = \abs{\phi_i}^q$,
	and to compute entries $\psi_i$ and $\phi_i$, in average time $\bigomic(\fval)$ for some
	$\fval$.
	It is possible, with probability less than $\delta>0$ of exceeding the
	error bound, to estimate $\braket{\psi}{\phi}$ to within additive error $\epsilon>0$ in
	average time $\bigomic(\log(\delta^{-1}) \epsilon^{-2} \fval)$.
\end{example}
\begin{proof}
	Let $V(i) = \psi_i^* \phi_i$ and $\Rdist{i} = \Pdist{i}/p + \Qdist{i}/q$.
	To apply \cref{thm:chernoff3} we need to bound $\bmax = \max_i\{ \abs{V(i)}/\Rdist{i} \}$.
	Making use of the (weighted) inequality of arithmetic and geometric means,\footnote{
		The weighted inequality of arithmetic and geometric means is a generalization of the
		more familiar inequality $x/2 + y/2 \ge \sqrt{xy}$.
		If $1 \le p \le \infty$ and $1/p + 1/q = 1$ then $x/p + y/q \ge x^{1/p} y^{1/q}$.
	}
	\begin{align}
		\bmax &= \max_i\{ \abs{V(i)}/\Rdist{i} \}
		\\ &=   \max_i\{ \abs{\psi_i^* \phi_i} / [\Pdist{i}/p     + \Qdist{i}/q    ] \}
		\\ &\le \max_i\{ \abs{\psi_i^* \phi_i} / [\Pdist{i}^{1/p}   \Qdist{i}^{1/q}] \}
		\\ &= 1.
	\end{align}
	By \cref{thm:chernoff3}, $\braket{\psi}{\phi} = \sum_i V(i)$ can be estimated at the cost
	of drawing $\bigomic(\log(\delta^{-1}) \epsilon^{-2})$ samples $i$ according to $\Rdist{i}$
	and computing the corresponding $V(i)/\Rdist{i}$ values.
	Sampling from $\Rdist{i}$ can be accomplished as
	follows: flip a biased coin that lands heads with probability $1/p$.  If it lands heads
	then draw $i$ from $\Pdist{i}$, otherwise draw $i$ from $\Qdist{i}$.
	By assumption this takes average time $\bigomic(\fval)$.
	Next, $V(i)/\Rdist{i}$ can be computed directly from $\psi_i$ and $\phi_i$, each of
	which can in turn be computed in average time $\bigomic(\fval)$.
	The $\bigomic(\log(\delta^{-1}) \epsilon^{-2})$ samples (as well as their mean) can
	therefore be computed in average time $\bigomic(\log(\delta^{-1}) \epsilon^{-2} \fval)$.
\end{proof}

\subsection{Nearly stochastic matrices}

We now move to a more
general case, estimation of $\braopket{\psi}{A^{(1)} \dotsm A^{(S)}}{\phi}$.
For the sake of simplicity, suppose that there are only two operators
(i.e.\ $S=2$) so that the goal is to estimate $\braopket{\psi}{AB}{\phi}$.
This can be written as a sum over paths as in~\eqref{eq:sum_over_paths},
\begin{equation}
	\braopket{\psi}{AB}{\phi} = \sum_{ijk} \psi^*_i A_{ij} B_{jk} \phi_k.
\end{equation}
To apply \cref{thm:chernoff3} to this problem,
set $\pi = (i,j,k)$ and $V(i,j,k) = \psi^*_i A_{ij} B_{jk} \phi_k$.
For efficient simulation it suffices to find a probability distribution $\Pdist{i,j,k}$ from
which we can efficiently draw samples using a classical computer, for which
$V(i,j,k) / \Pdist{i,j,k}$ can be efficiently computed, and for which
\begin{equation}
	\label{eq:cost_pathsum}
	\bmax =
	\max_{ijk} \left\{ \frac{ \abs{
		\psi^*_i A_{ij} B_{jk} \phi_k
	} }{ \Pdist{i,j,k} } \right\}
\end{equation}
is small enough that the estimation will converge reasonably fast.
As discussed in the previous section, a tempting choice for the probability distribution
is given by~\eqref{eq:holistic_dist},
however it is not clear how one would efficiently draw samples from this
since doing so apparently requires an understanding of how $\bra{\psi}$, $A$, $B$,
and $\ket{\phi}$ interact with each other.
To avoid this problem we define $\Pdist{i,j,k}$ in terms of a time-inhomogeneous Markov chain,
\begin{equation}
	\Pdist{i,j,k} = \Pdistsub{\psi}{i} \Pdistsub{A}{j|i} \Pdistsub{B}{k|j},
\end{equation}
with each transition depending on only one of the components of $\braopket{\psi}{AB}{\phi}$.
Plugging this into~\eqref{eq:cost_pathsum} gives
\begin{align}
	\bmax
	&=
	\max_{ijk} \left\{
		\frac{ \abs{
			\psi^*_i A_{ij} B_{jk} \phi_k
		} }{
			\Pdistsub{\psi}{i} \Pdistsub{A}{j|i} \Pdistsub{B}{k|j}
		}
	\right\}
	\\ &=
	\max_{ijk} \left\{
		\frac{ \abs{ \psi^*_i }}{ \Pdistsub{\psi}{i} } \cdot
		\frac{ \abs{ A_{ij} }}{ \Pdistsub{A}{j|i} } \cdot
		\frac{ \abs{ B_{jk} }}{ \Pdistsub{B}{k|j} } \cdot
		\abs{\phi_k}
	\right\}
	\\ &\le
	\label{eq:cost_markov_1inf}
	\max_i \left\{
		\frac{ \abs{ \psi^*_i }}{ \Pdistsub{\psi}{i} }
	\right\} \max_{ij} \left\{
		\frac{ \abs{ A_{ij} }}{ \Pdistsub{A}{j|i} }
	\right\} \max_{jk} \left\{
		\frac{ \abs{ B_{jk} }}{ \Pdistsub{B}{k|j} }
	\right\} \max_{k} \left\{
		\abs{\phi_k}
	\right\}.
\end{align}
The goal is then to find $\Pdistsub{\psi}{i}$, $\Pdistsub{A}{j|i}$, and
$\Pdistsub{B}{k|j}$ that minimize the terms of~\eqref{eq:cost_markov_1inf}.
Consider first the case where $\bra{\psi}$ is a probability
distribution, the matrices $A$ and $B$ are right-stochastic matrices,\footnote{
	A right-stochastic matrix is a nonnegative matrix with each row summing to 1,
	a left-stochastic matrix has each column summing to 1.
	We do not require stochastic matrices to be square.
}
and $\ket{\phi}$ has small entries (say, $\infnorm{\phi} \le 1$).
We can set $\Pdistsub{\psi}{i} = \psi_i$, $\Pdistsub{A}{j|i} = A_{ij}$, and
$\Pdistsub{B}{k|j} = B_{jk}$, with the result that each factor in~\eqref{eq:cost_markov_1inf}
is bounded by $1$.
If $\ket{\phi}$ is not a probability distribution, we can turn it into one by defining
$\Pdistsub{\psi}{i} = \abs{\psi_i} / \onenorm{\psi}$, similarly if $A$ is not a
right-stochastic matrix we can set
$\Pdistsub{A}{j|i} = \abs{A_{ij}} / \sum_{j'} \abs{A_{ij'}}$ (and likewise for $B$).
Then~\eqref{eq:cost_markov_1inf} becomes
\begin{align}
	\bmax
	&\le
	\onenorm{\psi}
	\max_i \left\{ \sum_{j'} \abs{A_{ij'}} \right\}
	\max_j \left\{ \sum_{k'} \abs{B_{jk'}} \right\}
	\infnorm{\phi}
	\\ &=
	\onenorm{\psi}
	\infnorm{\matabs{A}}
	\infnorm{\matabs{B}}
	\infnorm{\phi}.
\end{align}
Here, as in the rest of the paper, we use the induced norm for operators:
$\pnorm{A} = \max_{\uvec} \pnorm{A \uvec} / \pnorm{\uvec}$ (we do not use the entrywise or
Schatten norms).
Under this notation, $\twonorm{M}$ is the largest singular value of $M$, $\onenorm{M}$ is the
maximum absolute column sum, and $\infnorm{M}$ is the maximum absolute row sum.
By \cref{thm:chernoff3}, the value of $\braopket{\psi}{AB}{\phi}$ can be estimated by
drawing
\begin{equation}
	\bigomic\left(\log(\delta^{-1}) \epsilon^{-2} \bmax^2\right) \le
	\bigomic\left(\log(\delta^{-1}) \epsilon^{-2} \onenorm{\psi}^2
		\infnorm{\matabs{A}}^2 \infnorm{\matabs{B}}^2 \infnorm{\phi}^2\right)
\end{equation}
samples $(i,j,k)$ from $\Pdist{i,j,k}$ and averaging the corresponding $V(i,j,k)/\Pdist{i,j,k}$.

\subsection{General \texorpdfstring{$p,q$}{p,q}}
\label{subsec:sim_general_pq}

In the case of quantum circuits, it is the $\ell^2$-norm that is relevant.
Instead of
$\bmax \le \onenorm{\psi} \infnorm{\matabs{A}} \infnorm{\matabs{B}} \infnorm{\phi}$
from the previous example, we want
$\bmax \le \twonorm{\psi} \twonorm{\matabs{A}} \twonorm{\matabs{B}} \twonorm{\phi}$.
For the sake of generality, we allow arbitrary $p,q$ satisfying $1/p+1/q=1$.
The goal is to find a probability distribution that yields
$\bmax \le \pnorm{\psi} \qnorm{\matabs{A}} \qnorm{\matabs{B}} \qnorm{\phi}$.
As in \cref{subsec:dotprod}, the way to proceed is by taking a convex combination of two
probability distributions,
$\Rdist{i,j,k} = \Pdist{i,j,k}/p + \Qdist{i,j,k}/q$.
Here $\Pdist{i,j,k}$ will be a time-inhomogeneous Markov chain proceeding in the $i \to j \to k$
direction and $\Qdist{i,j,k}$ a different Markov chain proceeding in the $k \to j \to i$
direction.
Again the inequality of arithmetic and geometric means plays a crucial role, giving
\begin{align}
	\Rdist{i,j,k}
	&=
	\Pdist{i,j,k}/p + \Qdist{i,j,k}/q
	\\ &\ge
	\Pdist{i,j,k}^{1/p} \Qdist{i,j,k}^{1/q}
	\\ &=
	\left[ \Pdistsub{\psi}{i} \Pdistsub{A}{j|i} \Pdistsub{B}{k|j} \right]^{1/p}
	\left[ \Qdistsub{A}{i|j} \Qdistsub{B}{j|k} \Qdistsub{\phi}{k} \right]^{1/q}.
\end{align}
With this we have
\begin{align}
	\label{eq:cost_markov_pq_firstline}
	\bmax &=
	\max_{ijk} \left\{ \frac{ \abs{
		\psi^*_i A_{ij} B_{jk} \phi_k
	} }{ \Rdist{i,j,k} } \right\}
	\\ &\le
	\max_{ijk} \left\{ \frac{ \abs{
		\psi^*_i A_{ij} B_{jk} \phi_k
	} }{ \Pdist{i,j,k}^{1/p} \Qdist{i,j,k}^{1/q} } \right\}
	\\ &=
	\label{eq:cost_markov_pq_join}
	\max_{ijk} \left\{
		\frac{\abs{ \psi^*_i }}{ \Pdistsub{\psi}{i}^{1/p} } \cdot
		\frac{\abs{   A_{ij} }}{ \Pdistsub{A}{j|i}^{1/p} \Qdistsub{A}{i|j}^{1/q} } \cdot
		\frac{\abs{   B_{jk} }}{ \Pdistsub{B}{k|j}^{1/p} \Qdistsub{B}{j|k}^{1/q} } \cdot
		\frac{\abs{   \phi_k }}{ \Qdistsub{\phi}{k}^{1/q} }
	\right\}
	\\ &\le
	\label{eq:cost_markov_pq_split}
	\max_i    \left\{ \frac{\abs{ \psi^*_i }}{ \Pdistsub{\psi}{i}^{1/p} } \right\}
	\max_{ij} \left\{ \frac{\abs{   A_{ij} }}
		{ \Pdistsub{A}{j|i}^{1/p} \Qdistsub{A}{i|j}^{1/q} } \right\}
	\max_{jk} \left\{ \frac{\abs{   B_{jk} }}
		{ \Pdistsub{B}{k|j}^{1/p} \Qdistsub{B}{j|k}^{1/q} } \right\}
	\max_{k}  \left\{ \frac{\abs{   \phi_k }}{ \Qdistsub{\phi}{k}^{1/q} } \right\}
	\\ &= \bval_\psi \bval_A \bval_B \bval_\phi,
	\label{eq:cost_markov_pq_lastline}
\end{align}
where $\bval_\psi$, $\bval_A$, $\bval_B$, and $\bval_\phi$ label the four factors
of~\eqref{eq:cost_markov_pq_split}.
By \cref{thm:chernoff3}, the number of samples needed in order to estimate
$\braopket{\psi}{AB}{\phi}$ is
$\bigomic(\log(\delta^{-1}) \epsilon^{-2} \bval_\psi^2 \bval_A^2 \bval_B^2 \bval_\phi^2)$.
The quantities $\bval_\psi$, $\bval_A$, $\bval_B$, and $\bval_\phi$ are therefore identified as
being the simulation cost due to each of the components of $\braopket{\psi}{AB}{\phi}$.
We show in appendix~\ref{sec:uAv} (\cref{thm:best_eps_qnorm})
that for any choice of probability distribution
$\bval_A \ge \qnorm{\aA}$ and that there are optimal probability distributions achieving
$\bval_A = \qnorm{\aA}$ (and similarly for $B$, $\psi$, and $\phi$).
Using these gives
\begin{equation}
	\label{eq:cost_psi_A_B_phi}
	\bmax \le \pnorm{\psi} \qnorm{\matabs{A}} \qnorm{\matabs{B}} \qnorm{\phi}.
\end{equation}
Whether these optimal probability distributions can be efficiently sampled from is a matter
that needs to be considered on a case by case basis, however we show in
section~\ref{sec:circuits} that this is indeed the case for a wide range of matrices, both
unitary and Hermitian.
Additionally, in terms of query complexity rather than time complexity these efficient sampling
requirements can for the most part be ignored, as we will discuss further in
\cref{subsec:query}.

\subsection{Dyads and density operators}

It is possible to further generalize to expressions of the form $\Tr\{AB \sigma\}$.
The special case
$\braopket{\psi}{AB}{\phi}$ is obtained by setting $\sigma = \ket{\phi}\bra{\psi}$.
The above derivation is easily adapted by
writing $\sigma_{ki}$, $\Pdistsub{\sigma}{i}$, and $\Qdistsub{\sigma}{k}$ instead of
$\phi_k \psi^*_i$, $\Pdistsub{\psi}{i}$ and $\Qdistsub{\phi}{k}$.
With these
substitutions,~\eqref{eq:cost_markov_pq_firstline}-\eqref{eq:cost_markov_pq_lastline}
become
\begin{align}
	\label{eq:cost_sigma}
	\bmax
	&=
	\max_{ijk} \left\{ \frac{ \abs{
		A_{ij} B_{jk} \sigma_{ki}
	} }{ \Rdist{i,j,k} } \right\}
	\\ &\le
	\label{eq:cost_markov_sigma_join}
	\max_{ijk} \left\{
		\frac{\abs{    A_{ij} }}{ \Pdistsub{ A}{j|i}^{1/p} \Qdistsub{ A}{i|j}^{1/q} } \cdot
		\frac{\abs{    B_{jk} }}{ \Pdistsub{ B}{k|j}^{1/p} \Qdistsub{ B}{j|k}^{1/q} } \cdot
		\frac{\abs{ \sigma_{ki} }}{ \Pdistsub{\sigma}{i}^{1/p} \Qdistsub{\sigma}{k}^{1/q} }
	\right\}
	\\ &\le
	\label{eq:cost_markov_sigma_split}
	\max_{ij} \left\{ \frac{\abs{   A_{ij} }}
		{ \Pdistsub{A}{j|i}^{1/p} \Qdistsub{A}{i|j}^{1/q} } \right\}
	\max_{jk} \left\{ \frac{\abs{   B_{jk} }}
		{ \Pdistsub{B}{k|j}^{1/p} \Qdistsub{B}{j|k}^{1/q} } \right\}
	\max_{ki} \left\{ \frac{\abs{ \sigma_{ki} }}{
			\Pdistsub{\sigma}{i}^{1/p} \Qdistsub{\sigma}{k}^{1/q} } \right\}
	\\ &=
	\label{eq:cost_markov_sigma_bvals}
	\bval_A \bval_B \bval_\sigma.
\end{align}
The $\bval_\sigma$ factor differs from the other two in that the probability distributions are
not conditional.
This stems from the fact that $\sigma$ represents the starting point of the Markov chains.
If $\sigma = \ket{\phi}\bra{\psi}$ then taking the probability distributions
$\Pdistsub{\sigma}{i} = \abs{\psi_i}^p / \pnorm{\psi}$ and
$\Qdistsub{\sigma}{k} = \abs{\phi_k}^q / \qnorm{\phi}$ gives
$\bval_\sigma = \pnorm{\psi} \qnorm{\phi}$
as in~\eqref{eq:cost_psi_A_B_phi}.
If $p=q=2$ and if $\sigma$ is a density operator (positive semidefinite and trace 1) then
taking the probability distributions
$\Pdistsub{\sigma}{i} = \Qdistsub{\sigma}{i} = \sigma_{ii}$ gives $\bval_\sigma = 1$ due to the
inequality $\abs{\sigma_{ki}} \le \sqrt{\sigma_{kk} \sigma_{ii}}$, which is satisfied by
positive semidefinite matrices.

\subsection{Interference producing capacity}

In \cref{subsec:intf} we interpreted the lowest possible $\bmax$ value, obtained by
using the holistic probability distribution~\eqref{eq:holistic_dist}, as being the interference
of a quantum circuit.
Although this probability distribution achieves the lowest $\bmax$, there is no clear way to
draw samples efficiently and for this reason the Markov chain technique of this section was
developed.
The result was a strategy that depends only on properties of the individual operators rather
than on the expression as a whole.
The $\bmax$ value for this strategy is upper bounded by~\eqref{eq:cost_markov_sigma_bvals}.

Consider now the minimum possible value of one of the factors
in~\eqref{eq:cost_markov_sigma_bvals}, for instance $\bval_A$.
In appendix~\ref{sec:uAv} (\cref{thm:best_eps_qnorm}) we will show that the best possible
choice of $\Pdistsub{A}{j|i}$ and $\Qdistsub{A}{i|j}$ yields $\bval_A = \qnorm{\aA}$.
In the case of quantum circuits the relevant norm is $p=q=2$, so this becomes\footnote{
	We focus here on the case $p=2$ of relevance to quantum circuits,
	although the entire subsection could easily be generalized to $p \ne 2$.
}
\begin{equation}
	\bval_A = \twonorm{\aA}.
	\label{eq:bA_twonorm}
\end{equation}
This can be interpreted in terms of interference:
it is the largest possible contribution $A$ can make to the interference
$\intfc$ of \cref{def:intf_circuit}.
Specifically, since $\twonorm{\cdot}$ gives the maximum singular value of its argument, we have
\begin{equation}
	\intfc( A^{(1)}, \dotsc, A^{(S)}, \ket{\phi}\bra{\psi} ) \le
	\twonorm{\aA^{(1)}}
	\dotsm
	\twonorm{\aA^{(S)}}
	\twonorm{\phi}
	\twonorm{\psi}.
\end{equation}
Furthermore, for any operator $A$ we have
\begin{equation}
	\max_{\twonorm{\psi}=\twonorm{\phi}=1} \intfc(A, \ket{\phi}\bra{\psi})
		= \twonorm{\matabs{A}}.
\end{equation}
For this reason, we interpret $\twonorm{\aA}$ as being the interference producing
capacity of $A$.\footnote{
	Our measure of interference is different from, and seemingly unrelated to, the one defined
	in~\cite{PhysRevA.73.022314}, which in the case of unitary matrices reduces to
	$N-\sum_{ij}\abs{U_{ij}}^4$.
}

\begin{definition}
	\label{def:ifmax}
	The \textit{interference producing capacity} of a matrix $A$ is
	\begin{align}
		\ifmax(A) &= \twonorm{\aA}.
	\end{align}
\end{definition}

This definition, like \cref{def:intf_circuit}, is basis dependent.  Here the basis
dependence arises from the entrywise absolute value.  Unless otherwise specified, we will
work in the computational basis.  In the next sections we will show the product of
the $\ifmax$ values for the operations and final measurement of a circuit to be a necessary
resource for quantum speedup: if this quantity is low then a circuit can be classically
simulated.  The same claim applies also for other bases, and even for more exotic
representations (as we will show in \cref{sec:wigner}).
The situation is not so much different from, for instance, Gottesman-Knill theorem which
claims that stabilizer circuits may be efficiently simulated~\cite{arxiv:quant-ph/9807006}.
Although a circuit may at first not appear to be a stabilizer circuit it may be so after a
change of basis (i.e.\ after conjugating the initial state, all unitary operations, and
all measurements by some unitary).

The $\ifmax$ value for various operators is listed in \cref{tab:ifmax}.
As was shown informally in this section, and more formally in the next section, it is
possible to efficiently simulate quantum circuits when the product of the $\ifmax$ values
of all operators is not large.  So, one may interpret a small $\ifmax$ value to mean
that a unitary operator contributes only minimally to quantum speedup.
On the high end of the table are the Fourier and Hadamard transforms, having the maximum
possible value of $\ifmax$; these are difficult for us to simulate (at least in the
computational basis).  On the low end are the Pauli and the permutation
matrices, having $\ifmax=1$; these contribute nothing to quantum speedup (relative to our
simulation scheme).
Among unitaries, the only operators with $\ifmax=1$ are permutations with phases,
$U = \sum_j e^{i \theta_j} \ket{\sigma(j)} \bra{j}$.

\begin{table}[h]
	\centering
	\begin{tabular}{|p{7cm}|l|l|}
		\hline
		\multicolumn{1}{|c|}{Matrix} &
		\multicolumn{1}{|c|}{$\ifmax$} \\
		\hline
		Fourier or Hadamard transform on $n$ qubits & $2^{n/2}$ \\
		\hline
		Arbitrary gate on $n$ qudits & no more than $d^{n/2}$ \\
		\hline
		Haar wavelet transform on $n$ qubits & $\sqrt{1+n}$ \\
		\hline
		$k$-sparse unitary & no more than $\sqrt{k}$ \\
		\hline
		Grover reflection & $\ifmax \to 3$ as $n \to \infty$ \\
		\hline
		Permutation in computational basis & 1 \\
		\hline
		Pauli matrices & 1 \\
		\hline
		Rank one projector & 1 \\
		\hline
	\end{tabular}
	\caption{
		The $\ifmax$ value for various matrices.
		Operators with larger $\ifmax$ value are harder to simulate using our technique.
		Proofs for the nontrivial cases are presented in appendix~\ref{sec:circuits_proofs}.
	}
	\label{tab:ifmax}
\end{table}

\section{EPS and EHT operators}
\label{sec:eps_eht}

\subsection{Definitions}

We will now present two definitions codifying the requirements operators
must meet in order that products of the form
$\Tr\{A^{(1)} \dotsm A^{(S)} \sigma\}$
can be estimated using the techniques of the previous section.
In the previous section, using a pair of Markov chains yielded a simulation strategy in which
each component of $\Tr(AB \sigma)$ can be treated independently, with
$A$, $B$, and $\sigma$ contributing costs $\bval_A$, $\bval_B$, and $\bval_\sigma$ to the
total number of samples needed as per~\eqref{eq:cost_markov_sigma_bvals}.
Each sample requires drawing a random path according to the distribution $\Rdist{i,j,k}$
and then computing $V(i,j,k)/\Rdist{i,j,k}$.
Drawing the random path requires considering only one operator at a time since
$\Rdist{i,j,k}$ is defined in terms of Markov chains.
Similarly, computing $V(i,j,k)/\Rdist{i,j,k}$ can be done considering one operator at a time
since
\begin{align}
	\frac{ V(i,j,k) }{ \Rdist{i,j,k} }
	&= \frac{ A_{ij} B_{jk} \sigma_{ki} }{
		\Pdist{i,j,k}/p + \Qdist{i,j,k}/q
	}
	\\ &= \left\{
		\frac{1}{p}
		\frac{ \Pdist{i,j,k} }{ A_{ij} B_{jk} \sigma_{ki} }
		+
		\frac{1}{q}
		\frac{ \Qdist{i,j,k} }{ A_{ij} B_{jk} \sigma_{ki} }
	\right\}^{-1}
	\\ \label{eq:V_over_R_split} &= \left\{
		\frac{1}{p}
		\frac{ \Pdistsub{A}{j|i}  }{ A_{ij} }
		\frac{ \Pdistsub{B}{k|j}  }{ B_{jk} }
		\frac{ \Pdistsub{\sigma}{i} }{ \sigma_{ki} }
		+
		\frac{1}{q}
		\frac{ \Qdistsub{A}{i|j}  }{ A_{ij} }
		\frac{ \Qdistsub{B}{j|k}  }{ B_{jk} }
		\frac{ \Qdistsub{\sigma}{k} }{ \sigma_{ki} }
	\right\}^{-1}.
\end{align}
Focusing on a single component, say $A$, conditions for efficient simulation can be identified
(note that $\sigma$ requires slightly different conditions, which we deal with later).
First, the quantity $\bval_A$ of~\eqref{eq:cost_markov_sigma_bvals} should be small in order that
the number of samples required be small.
Second, it must be possible to efficiently sample from the probability distributions
$\Pdistsub{A}{j|i}$ and $\Qdistsub{A}{i|j}$ and to compute the contributions
due to $A$ in~\eqref{eq:V_over_R_split}, namely
$\Pdistsub{A}{j|i} / A_{ij}$ and $\Qdistsub{A}{i|j} / A_{ij}$.
We express these conditions as a definition.
However, it will be useful to generalize by allowing an extra index $k$ in the definition
below (not related to the $k$ that appears above).
If $k$ takes only a single value (say, $k=0$) the definition below exactly encompasses
the conditions outlined above.
The extra freedom granted by $k$ will allow, as we will show shortly,
treatment of sums, products, and exponentials of matrices
(\cref{thm:eps_sumprodexp}).
In the case $p=1$, $q=\infty$ it was the matrices resembling stochastic matrices that could
be efficiently simulated.  For this reason, for general $p,q$ we give the name
\textit{efficient pseudo-stochastic} (EPS) to matrices that we can efficiently simulate.

\begin{definition}[EPS]
	\label{def:eps}
	Let $1 \le p \le \infty$, $1/p+1/q=1$, and $\bval < \infty$.
	An $M \times N$ matrix $A$ is $\epsp{\bval}{\fval}$ if there is a finite or countable
	set $K$, values $\alpha_{mnk} \in \mathbb{C}$, and conditional probability distributions
	$\Pdist{n,k|m}$ and $\Qdist{m,k|n}$
	with $m \in \{1, \dotsc, M\}$, $n \in \{1, \dotsc, N\}$, and $k \in K$,
	satisfying the following conditions:
	\begin{enumerate}[(a)]
		\item \label{cond:eps_alphasum}
			$\sum_{k \in K} \alpha_{mnk} = A_{mn}$.\footnote{
				We show in appendix~\ref{sec:eps_eht_proofs} (\cref{thm:abs_conv}) that this series
				converges absolutely, so there is no ambiguity regarding the way that an
				infinite $K$ is enumerated.
			}
		\item \label{cond:eps_cost}
			\begin{equation}
				\label{eq:eps_cost}
				\max_{mnk} \left\{
					\frac{\abs{\alpha_{mnk}}}{
						\Pdist{n,k|m}^{1/p}
						\Qdist{m,k|n}^{1/q}
					}
				\right\} \le \bval,
			\end{equation}
			with the convention that $0/0=0$.
		\item \label{cond:eps_samplr}
			Given any $m$, it is possible in average time $\bigomic(\fval)$ on a classical
			computer to sample $n,k$ from the probability distribution $\Pdist{n,k|m}$ and then
			compute $\alpha_{mnk} / \Pdist{n,k|m}$ and $\alpha_{mnk} / \Qdist{m,k|n}$.
		\item \label{cond:eps_samprl}
			Given any $n$, it is possible in average time $\bigomic(\fval)$ on a classical
			computer to sample $m,k$ from the probability distribution $\Qdist{m,k|n}$ and then
			compute $\alpha_{mnk} / \Pdist{n,k|m}$ and $\alpha_{mnk} / \Qdist{m,k|n}$.
	\end{enumerate}
\end{definition}

This definition is related to interference producing capacity in the following way.
It is always possible to satisfy conditions~\eqref{cond:eps_alphasum}
and~\eqref{cond:eps_cost} with $\bval = \qnorm{\aA}$, and it is impossible to do better.
This is proved in appendix~\ref{sec:uAv}.
So, for the case $p=q=2$ the optimal value of $b$ is equal to the interference producing
capacity of $A$.  Since $b$ (multiplied for all operators in a circuit) determines how
many samples will be required for our simulation technique, this connects interference
producing capacity to difficulty of simulation.

Although conditions~\eqref{cond:eps_alphasum} and~\eqref{cond:eps_cost} can always be
satisfied with $\bval = \qnorm{\aA}$ for some $\alpha_{mnk}$, $\Pdist{n,k|m}$, and
$\Qdist{m,k|n}$, it could be the case that these do not satisfy~\eqref{cond:eps_samplr}
and~\eqref{cond:eps_samprl}.  In other words, it may be time consuming to sample from
these probability distributions.
An example would be a permutation matrix $A\ket{x} = \ket{g(x)}$.
Such a matrix has $\qnorm{\aA}=1$, so it has no interference producing capacity.
Nevertheless, it would be difficult to simulate if
the function $g$ were difficult to calculate.
In some sense~\eqref{cond:eps_samplr} and~\eqref{cond:eps_samprl} constitute a requirement
that the matrix $A$ be well understood from a computational perspective.
In practice,~\eqref{cond:eps_samplr} and~\eqref{cond:eps_samprl} have not presented an
obstacle for any of the operators that we have considered.
If one is concerned with query complexity rather than time complexity
then~\eqref{cond:eps_samplr} and~\eqref{cond:eps_samprl} can mostly be ignored.  This
will be explored in \cref{subsec:query}.

There is a subtlety in conditions~\eqref{cond:eps_samplr} and~\eqref{cond:eps_samprl}
that deserves discussion.  It is required that the operations be carried out in
\textit{average} time $\bigomic(\fval)$.  It is allowed that
$\alpha_{mnk} / \Pdist{n,k|m}$ and $\alpha_{mnk} / \Qdist{m,k|n}$
be difficult to compute for some $m,n,k$ triples as long as
those occur rarely when sampling from $\Pdist{n,k|m}$ or $\Qdist{m,k|n}$.
In our implementation of exponentials of operators
(\cref{thm:eps_sumprodexp}\eqref{part:eps_exp}) the time required is proportional to $k$, and
so is unbounded since $k \in \{0,1,\dotsc\}$, however $\Pdist{n,k|m}$ and
$\Qdist{m,k|n}$ decay exponentially in $k$ so the average time is small.

We now present a definition that embodies the conditions $\sigma$ must
satisfy in order to yield an efficient simulation.
Looking to~\eqref{eq:cost_markov_sigma_split} and~\eqref{eq:V_over_R_split}, the difference
between the factors relating to $\sigma$ and those relating to $A$ are that the latter involve
conditional probability distributions.
This stems from the fact that the Markov chains begin at $\sigma$ and so have no index to
condition upon.
With this difference in mind, we provide a definition analogous to \cref{def:eps} but with
non-conditional probability distributions.
Since the Markov chains begin and end at $\sigma$, we name the suitable matrices
\textit{efficient head/tail} (EHT) matrices.

\begin{definition}[EHT]
	\label{def:eht}
	Let $1 \le p \le \infty$, $1/p+1/q=1$, and $\bval < \infty$.
	An $M \times N$ matrix $\sigma$ is $\ehtp{\bval}{\fval}$  if there is a
	finite or countable
	set $K$, values $\alpha_{mnk} \in \mathbb{C}$, and probability distributions
	$\Pdist{n,k}$ and $\Qdist{m,k}$
	with $m \in \{1, \dotsc, M\}$, $n \in \{1, \dotsc, N\}$, and $k \in K$,
	satisfying the following conditions:
	\begin{enumerate}[(a)]
		\item \label{cond:eht_alphasum}
			$\sum_{k \in K} \alpha_{mnk} = \sigma_{mn}$.
		\item \label{cond:eht_cost}
			\begin{equation}
				\label{eq:eht_cost}
				\max_{mnk} \left\{
					\frac{\abs{\alpha_{mnk}}}{
						\Pdist{n,k}^{1/p}
						\Qdist{m,k}^{1/q}
					}
				\right\} \le \bval,
			\end{equation}
			with the convention that $0/0=0$.
		\item \label{cond:eht_samplr}
			It is possible in average time $\bigomic(\fval)$ on a classical computer to
			sample $n,k$ from the probability distribution $\Pdist{n,k}$ and then,
			given any $m \in \{1, \dotsc, M\}$ to
			compute $\alpha_{mnk} / \Pdist{n,k}$ and $\alpha_{mnk} / \Qdist{m,k}$.
		\item \label{cond:eht_samprl}
			It is possible in average time $\bigomic(\fval)$ on a classical computer to
			sample $m,k$ from the probability distribution $\Qdist{m,k}$ and then,
			given any $n \in \{1, \dotsc, N\}$ to
			compute $\alpha_{mnk} / \Pdist{n,k}$ and $\alpha_{mnk} / \Qdist{m,k}$.
	\end{enumerate}
\end{definition}

This definition does not relate to interference.  For the case of quantum circuits we can
assume $\sigma$ to be a density operator.  In \cref{sec:sufficient_eps_eht} we show that
for density operators it is always possible to achieve $b=1$ in the above definition as
long as one can simulate measurements in the computational basis and compute
individual matrix entries in average time $\bigomic(\fval)$.

The definition of EHT is more strict than that of EPS: any EHT operator can be seen to also be
EPS by using the probability distributions $\Pdist{n,k|m} = \Pdist{n,k}$ and
$\Qdist{m,k|n} = \Qdist{m,k}$.  Therefore, since it is not possible to have
$\bval < \qnorm{\aA}$ for EPS operators, it is also not possible to have
$\bval < \qnorm{\matabs{\sigma}}$ for EHT operators.
As mentioned above, in the case of EPS it is always possible to
satisfy conditions~\eqref{cond:eps_alphasum} and~\eqref{cond:eps_cost} with
$\bval = \qnorm{\aA}$, however since EHT is more strict there are operators $\sigma$ for which
it is not possible to have $\bval = \qnorm{\matabs{\sigma}}$.
\Cref{thm:best_eps_qnorm}\eqref{part:best_dyads} in appendix~\ref{sec:uAv} gives that
$\bval=\trnorm{\matabs{\sigma}}$ is possible when $p=q=2$ where
$\trnorm{\cdot}$ is the trace norm (and a generalization is provided for
$p \ne 2$).

In section~\ref{sec:circuits} we will consider the case $p=q=2$, which is the norm relevant to
quantum circuits, and give several examples of states that are $\ehttwo{\bval}{\fval}$ and
operators that are $\epstwo{\bval}{\fval}$ where $\bval$ is small and $\fval$ is polynomial in
the number of qubits (or polylog in the dimension of the system).
Expectation values of circuits built from such states and operators can be efficiently
simulated.
Specifically, we have the following theorem, the central theorem of this paper, whose proof
will be deferred until after \cref{thm:eht_eps_trace}.

\begin{theorem}[Efficient simulation]
	\label{thm:chain_sim}
	Let $\sigma$ be $\ehtp{\bval_\sigma}{\fval_\sigma}$
	and for $t \in \{1,\dotsc,S\}$ let $A^{(t)}$ be $\epsp{\bval_t}{\fval_t}$.
	Then, with probability less than $\delta>0$ of exceeding the error bound,
	$\Tr\{A^{(1)} \dotsm A^{(S)} \sigma\}$ can be estimated to within
	additive error $\epsilon > 0$ in average time
	$\bigomic(\log(\delta^{-1}) \epsilon^{-2} \bval^2 \fval)$
	where $\bval = \bval_\sigma \prod_t \bval_t$ and
	$\fval = \fval_\sigma + \sum_t \fval_t$.
\end{theorem}

\subsection{Operations that preserve EPS/EHT properties}

We now discuss mathematical operations that preserve the EPS and EHT properties.  These
include scaling, transpose, adjoint, multiplication, addition, and exponentiation
(\cref{thm:basic_math,thm:eps_sumprodexp}).
The first three follow immediately from the definitions, so the following theorem is presented
without proof.

\begin{theorem}
	\label{thm:basic_math}
	Let $A$ be $\epsp{\bval}{\fval}$ and $\sigma$ be $\ehtp{\bval}{\fval}$.
	Let $s \in \mathbb{C}$ be a scalar.  Then
	\begin{enumerate}[(a)]
		\item $\sigma$ is $\epsp{\bval}{\fval}$.
		\item $sA$ is $\epsp{\abs{s} \bval}{\fval}$.
		\item $s\sigma$ is $\ehtp{\abs{s} \bval}{\fval}$.
		\item $A^\top$ and $A^\dagger$ are $\epsp{\bval}{\fval}$.
		\item $\sigma^\top$ and $\sigma^\dagger$ are $\ehtp{\bval}{\fval}$.
	\end{enumerate}
\end{theorem}

The presence of the $k$ index in \cref{def:eps} allows treatment of sums and products of
operators.
Consider for instance the product $AB$.
The two factors of~\eqref{eq:cost_markov_sigma_join} relating to $A$ and $B$ can be combined to
match the conditions of \cref{def:eps} as follows.
Begin by relabeling the indices of~\eqref{eq:cost_markov_sigma_join} from $i,j,k$ to $m,k,n$
and proceed as follows,
\begin{align}
	\bmax &\le
	\max_{mnk} \left\{
		\frac{\abs{ \sigma_{nm} }}{ \Pdistsub{\sigma}{m}^{1/p} \Qdistsub{\sigma}{n}^{1/q} } \cdot
		\frac{\abs{    A_{mk} }}{ \Pdistsub{ A}{k|m}^{1/p} \Qdistsub{ A}{m|k}^{1/q} } \cdot
		\frac{\abs{    B_{kn} }}{ \Pdistsub{ B}{n|k}^{1/p} \Qdistsub{ B}{k|n}^{1/q} }
	\right\}
	\\ &\le
	\label{eq:cost_sigma_AB}
	\max_{mn} \left\{
		\frac{\abs{ \sigma_{nm} }}{ \Pdistsub{\sigma}{m}^{1/p} \Qdistsub{\sigma}{n}^{1/q} }
	\right\} \cdot
	\max_{mnk} \left\{
		\frac{\abs{    A_{mk} }}{ \Pdistsub{ A}{k|m}^{1/p} \Qdistsub{ A}{m|k}^{1/q} } \cdot
		\frac{\abs{    B_{kn} }}{ \Pdistsub{ B}{n|k}^{1/p} \Qdistsub{ B}{k|n}^{1/q} }
	\right\}
	\\ &\le
	\max_{mn} \left\{
		\frac{\abs{ \sigma_{nm} }}{ \Pdistsub{\sigma}{m}^{1/p} \Qdistsub{\sigma}{n}^{1/q} }
	\right\} \cdot
	\max_{mnk} \left\{
		\frac{\abs{ A_{mk} B_{kn} }}
		{ [\Pdistsub{A}{k|m} \Pdistsub{B}{n|k}]^{1/p}
			[\Qdistsub{A}{m|k} \Qdistsub{B}{k|n}]^{1/q} }
	\right\}
	\\ &= \bval_\sigma \bval_{AB}.
\end{align}
Defining
$\Pdistsub{AB}{n,k|m} = \Pdistsub{A}{k|m} \Pdistsub{B}{n|k}$,
$\Qdistsub{AB}{m,k|n} = \Qdistsub{B}{k|n} \Qdistsub{A}{m|k}$,
and $\alpha_{mnk} = A_{mk} B_{kn}$, the $\bval_{AB}$ factor reduces to
\begin{align}
	\bval_{AB} &=
	\max_{mnk} \left\{
		\frac{\abs{ \alpha_{mnk} }}
		{ \Pdistsub{AB}{n,k|m}^{1/p} \Qdistsub{AB}{m,k|n}^{1/q} }
	\right\}.
\end{align}
This resembles the factors involving $A$ or $B$ that appear
in~\eqref{eq:cost_markov_sigma_split} but with the addition of an extra index $k$ appearing in
both the numerator and in the probability distributions.
Allowing such an extra index enables treatment of $AB$ in the same manner as
the individual factors $A$ and $B$.  This is formalized by
\cref{thm:eps_sumprodexp}\eqref{part:eps_prod} below, which
states that the product of EPS matrices is EPS\@.
In the general case this procedure is slightly complicated by the fact that $A$
and $B$ may in turn have their own extra indices $k'$ and $k''$, which must be inherited by
the product $AB$.

Sums are handled in a similar way.  An expression such as $\Tr((A+B)\sigma)$ is estimated by
using $A$ for a fraction of the samples and $B$ for the remainder.
This works since $\Tr((A+B)\sigma)$ is twice the average of $\Tr(A \sigma)$ and
$\Tr(B \sigma)$.
The $k$ index is used to randomly choose between $A$ or $B$ for each sample.
Exponentials are treated by applying these sum and product rules to
$e^A = \sum_{j=0}^\infty A^j / j! $.

\begin{theorem}[Operations on EPS]
	\label{thm:eps_sumprodexp}
	Let $A$ be a matrix that is $\epsp{\bval_A}{\fval_A}$ and let
	$B$ be a matrix that is $\epsp{\bval_B}{\fval_B}$.
	Then, assuming in each case that $A$ and $B$ have a compatible number of rows and columns,
	the following hold.
	\begin{enumerate}[(a)]
		\item \label{part:eps_sum_simple}
			$A + B$ is $\epsp{\bval_A+\bval_B}{\max\{\fval_A, \fval_B\}}$.
		\item \label{part:eps_prod}
			$AB$ is $\epsp{\bval_A \bval_B}{\fval_A+\fval_B}$.
		\item \label{part:eps_exp}
			$e^A$ is $\epsp{e^\bval}{\bval \fval}$.
	\end{enumerate}
\end{theorem}
\begin{proof}
	The proofs are in appendix~\ref{sec:eps_eht_proofs}.
	Rule~\eqref{part:eps_sum_simple} is a special case of \cref{thm:eps_sum_fancy},
	which treats finite or infinite linear combinations.
\end{proof}

Since the value $b$ in \cref{def:eps} (with $p=q=2$) is lower bounded by interference producing
capacity $\ifmax$, \cref{thm:eps_sumprodexp} has the following interpretation.
By~\eqref{part:eps_sum_simple}, $\ifmax$ is convex.
By~\eqref{part:eps_prod}, it is sub-multiplicative.
By~\eqref{part:eps_exp}, the interference producing capacity of a Hamiltonian evolution
$e^{iHt}$ is at most exponential in $t \ifmax(H)$.

We now prove \cref{thm:chain_sim}, regarding estimation of
$\Tr\{A^{(1)} \dotsm A^{(S)} \sigma\}$.
While this can be proved directly using Markov chains, as was done in
section~\ref{sec:markov}, this would be notationally tedious.
It is much easier to first repeatedly apply the product rule,
\cref{thm:eps_sumprodexp}\eqref{part:eps_prod}, to show
that $A = A^{(1)} \dotsm A^{(S)}$ is $\epsp{\prod_t \bval_t}{\sum_t \fval_t}$.
It then suffices to show that $\Tr(A \sigma)$ can be estimated.
Although this may seem like a slightly non-constructive proof, this strategy arose due to
object-oriented techniques (C++) used during actual implementation of the algorithm.
Unrolling the proof of the product theorem, as well as the proof of the theorem that follows,
gives an argument very similar to that presented in section~\ref{sec:markov}.

\begin{lemma}
	\label{thm:eht_eps_trace}
	Let $\sigma$ be an $N \times M$ matrix that is $\ehtp{\bval_\sigma}{\fval_\sigma}$.
	Let $A$ be an $M \times N$ matrix that is $\epsp{\bval_A}{\fval_A}$.
	It is possible to estimate $\Tr(A \sigma)$
	to within additive error $\epsilon>0$, with probability less than $\delta>0$ of
	exceeding the error bound, in average time
	$\bigomic[\log(\delta^{-1}) \epsilon^{-2} \bval_\sigma^2 \bval_A^2 (\fval_\sigma+\fval_A)]$.
\end{lemma}
\begin{proof}
	The proof is in appendix~\ref{sec:eps_eht_proofs}, and follows along the lines of the
	techniques developed in section~\ref{sec:markov}.
\end{proof}

\begin{proof}[Proof of \cref{thm:chain_sim}]
	By iterated application of \cref{thm:eps_sumprodexp}\eqref{part:eps_prod},
	$A = A^{(1)} \dotsm A^{(S)}$ is $\epsp{\prod_t \bval_t}{\sum_t \fval_t}$.
	By \cref{thm:eht_eps_trace} the value of $\Tr(A \sigma)$ can be estimated
	in time
	$\bigomic(\log(\delta^{-1}) \epsilon^{-2} \bval^2 \fval)$
	where $\bval = \bval_\sigma \prod_t \bval_t$ and
	$\fval = \fval_\sigma + \sum_t \fval_t$.
\end{proof}

\subsection{Query complexity}
\label{subsec:query}

The simulation algorithm of this paper involves sampling a number of paths via Markov chains,
each path evaluation in turn requiring certain operations to be performed.
\Cref{def:eps,def:eht} each consist of two pairs of conditions,
\eqref{cond:eps_alphasum} and~\eqref{cond:eps_cost} relating to the number of paths that need
to be evaluated (quantified by $\bval$), and~\eqref{cond:eps_samplr} and~\eqref{cond:eps_samprl}
concerning tasks that need to be performed for each path (quantified by $\fval$).
In appendix~\ref{sec:uAv} we show (\cref{thm:best_eps_qnorm}) that there are always
$\alpha_{mnk}$, $\Pdist{n,k|m}$, and $\Qdist{m,k|n}$ satisfying
conditions~\eqref{cond:eps_alphasum} and~\eqref{cond:eps_cost} with
$\bval = \qnorm{\aA}$ (and in fact smaller $\bval$ is not possible).
However, these probability distributions may not satisfy~\eqref{cond:eps_samplr}
and~\eqref{cond:eps_samprl}, which require that the distributions can be sampled from
efficiently.
It is difficult to make any general statement regarding satisfaction
of~\eqref{cond:eps_samplr} and~\eqref{cond:eps_samprl}, since time complexity of
computation is in general a difficult problem; satisfaction of these two conditions needs to be
considered on a case-by-case basis.
However, when considering query complexity rather than time complexity,~\eqref{cond:eps_samplr}
and~\eqref{cond:eps_samprl} can for the most part be ignored as we shall now explain.
Note that communication complexity (discussed in \cref{subsec:comm}) offers another context in
which~\eqref{cond:eps_samplr} and~\eqref{cond:eps_samprl} can be ignored, since there too
computation time is free.

Consider the situation where an algorithm is required to answer some question about an
oracle, which is to be thought of as a black box provided to the algorithm (Grover's algorithm
is a prominent example).
For a classical (i.e.\ non-quantum) algorithm the oracle can be any function between two finite
sets, say $g : X \to Y$.
It will be convenient to consider sets of integers,
$X=\{0,1,\dotsc,\abs{X}-1\}$ and $Y=\{0,1,\dotsc,\abs{Y}-1\}$.
The algorithm can query the oracle by providing it a value $x \in X$, and the oracle responds
with $g(x)$.
This is the only allowed way to gain information about $g$.
The query complexity of the algorithm is defined to be the number of times it queries the
oracle.
In particular, the query complexity is not affected by the amount of time spent performing
computations between queries; computation, even lengthy computation, is not charged for.

Quantum circuits are provided access to an oracle in the form of a unitary operator\footnote{
	Sometimes an alternate definition
	$\mathcal{O}'_g = \sum_{x\in X, y\in Y}
	e^{2 \pi i g(x) y / \abs{Y}} \ket{x}\bra{x} \ot \ket{y}\bra{y}$
	is used.  All claims apply to this definition as well, requiring only
	a modification of~\eqref{eq:oracle_eps_first}-\eqref{eq:oracle_eps_last}.
}
\begin{equation}
	\mathcal{O}_g = \sum_{x\in X, y\in Y}
	\ket{x}\bra{x} \ot \ket{y + g(x)}\bra{y}
	\label{eq:quantum_oracle}
\end{equation}
where $\ket{x} \ot \ket{y} \in \mathbb{C}^{\abs{X}} \ot \mathbb{C}^{\abs{Y}}$ are
computational basis vectors and where the addition $y+g(x)$ is modulo $\abs{Y}$.
%Alternatively, the oracle can be defined to be the diagonal unitary
%\begin{equation}
%	\mathcal{O}'_g = \sum_{x\in X, y\in Y}
%	e^{2 \pi i g(x) y / \abs{Y}} \ket{x}\bra{x} \ot \ket{y}\bra{y}.
%	\label{eq:phase_oracle}
%\end{equation}
%This is equivalent to~\eqref{eq:quantum_oracle}
%in that one can be converted to the other through conjugation by a Fourier transform on the
%second subsystem.
The query complexity of a quantum circuit is defined to be the number of times $\mathcal{O}_g$
appears in the circuit.
For example, Grover's algorithm has query complexity $\bigomic(\sqrt{N})$.

Computational complexity classes can be analyzed by comparing how two classes perform when
given access to equivalent oracles.  For example, oracles have been constructed relative to
which quantum computers perform exponentially more efficiently than classical computers (e.g.\
Simon's problem~\cite{10.1109/SFCS.1994.365701}), whereas proving that quantum computers are
faster than classical computers in the absence of an oracle is an extremely difficult open
problem.

Considering query complexity rather than time complexity simplifies the analysis of the
present paper.
Suppose we wish to simulate a quantum circuit containing at least one instance of an oracle
$\mathcal{O}_g$ (e.g.\ Grover's algorithm) on a classical computer that also has oracle access
to $g$.
Simulation of the quantum circuit on the classical computer will require making queries to $g$
and we can ask how many queries are needed, ignoring the amount of computational time used.
We do this by modifying conditions~\eqref{cond:eps_samplr} and~\eqref{cond:eps_samprl} of
\cref{def:eps,def:eht} to require that the sampling and computation tasks be completed using
$\bigomic(\fval)$ queries to $g$, rather than requiring $\bigomic(\fval)$ time (time now being
a resource that is not charged for).
We will refer to such modified definitions by invoking the phrase ``in terms of query
complexity.''

We will now show that in terms of query complexity, $\mathcal{O}_g$ is $\epsp{1}{1}$.
Since this unitary operates on two subsystems,
$\mathbb{C}^{\abs{X}} \ot \mathbb{C}^{\abs{Y}}$,
the indices $m$ and $n$ in \cref{def:eps} are tuple valued.
We write $m=(x,y) \in X \times Y$
and $n=(x',y') \in X \times Y$.
Take $K$ to be the singleton set $\{0\}$ and define
\begin{align}
	\label{eq:oracle_eps_first}
	\alpha_{(x,y)(x',y')k}  &:= \Pdist{(x',y'),k|(x,y)}
	\\ &:= \Qdist{(x,y),k|(x',y')}
	\\ &:= \braopket{xy}{\mathcal{O}_g}{x'y'}
	\\ &= \delta(x,x') \delta(y+g(x), y')
	\label{eq:oracle_eps_last}
\end{align}
where $\delta$ is the Kronecker delta.
It is easy to see that these satisfy conditions~\eqref{cond:eps_alphasum}
and~\eqref{cond:eps_cost} of \cref{def:eps} with $\bval=1$.
Sampling from these probability distributions and computing the values of any of these
quantities can be done with a single query of $g$ (note that the conditional probability
distributions are deterministic), therefore
conditions~\eqref{cond:eps_samplr} and~\eqref{cond:eps_samprl} are satisfied with
$\fval=1$.
%Also $\mathcal{O}'_g$ is $\epsp{1}{1}$, which can be seen by
%setting
%\begin{align}
%	\alpha_{(x,y)(x',y')k}  &= \braopket{xy}{\mathcal{O}'_g}{x'y'}
%	\\
%	\Pdist{(x',y'),k|(x,y)} &= \delta(x,x') \delta(y,y')
%	\\
%	\Qdist{(x,y),k|(x',y')} &= \delta(x,x') \delta(y,y'),
%\end{align}
%where $\delta$ is the Kronecker delta.

On the other hand, for matrices that are not defined in terms of the oracle $g$, such as the
$I-2\ket{+}\bra{+}$ reflection operators in Grover's algorithm,
the operations required by conditions~\eqref{cond:eps_samplr} and~\eqref{cond:eps_samprl}
can be carried out using zero queries.
Therefore conditions~\eqref{cond:eps_samplr} and~\eqref{cond:eps_samprl} can be
completely ignored, and we can take $\fval=0$.
We are then free to focus on determining the probability distributions giving the
smallest possible value of $\bval$ in conditions~\eqref{cond:eps_alphasum}
and~\eqref{cond:eps_cost} without regard to whether these can be efficiently sampled from
(since we are charging for queries only and time is free).
It is desirable to make $\bval$ as small as possible, since this determines the number of paths
that need to be sampled.
The number of paths sampled matters, because each will require evaluating the entire Markov
chain, which involves every operator.
At least one of these operators involves the oracle, so at least one query needs to be made for
each path that is sampled.
The total number of oracle queries will be the number of paths sampled times the number of
queries per path.
In appendix~\ref{sec:uAv} we show (\cref{thm:best_eps_qnorm}) the existence of probability
distributions which satisfy conditions~\eqref{cond:eps_alphasum} and~\eqref{cond:eps_cost}
with $b=\qnorm{\aA}$.
So in terms of query complexity, any matrix $A$ not defined in terms of an oracle
is $\epsp{\qnorm{\aA}}{0}$.
In the case $p=q=2$ of relevance to quantum circuits, we have $\twonorm{\aA} = \ifmax(A)$,
the interference producing capacity of $A$.
\Cref{thm:best_eps_qnorm} also shows that any $\sigma$ not defined in terms of an oracle
is $\ehttwo{\trnorm{\sigma}}{0}$
where $\trnorm{\cdot}$ is the trace norm (a generalization is provided for $p \ne 2$).

\subsection{Sufficient conditions for EPS/EHT}
\label{sec:sufficient_eps_eht}

We now present theorems that can be used to show that specific operators are EPS or EHT\@.
As stated above, if one is only interested in query complexity then any matrix $A$ not
depending on an oracle is guaranteed to be $\epsp{\qnorm{\aA}}{0}$.
However, in terms of time complexity it is possible that the probability distributions that
achieve $\bval=\qnorm{\aA}$ cannot be sampled from efficiently (giving large $\fval$).
For this reason it is worthwhile to introduce probability distributions that are more likely to
be efficiently sampled, and which in some cases still achieve a small $\bval$.
In the theorem below each row and column of $A$ is treated as a probability distribution,
correcting for phases and normalization.
This works well when the absolute row and column sums of $A$ are small.

\begin{theorem}
	Let $1 \le p \le \infty$ and $1/p+1/q=1$.
	Let $A$ be an $M \times N$ matrix.  Define the probability distributions
	\begin{equation}
		\Pdist{n|m} = \frac{ \abs{A_{mn}} }{ \sum_{n'} \abs{A_{mn'}}},\quad
		\Qdist{m|n} = \frac{ \abs{A_{mn}} }{ \sum_{m'} \abs{A_{m'n}}}.
		\label{eq:eps1_dist}
	\end{equation}
	Suppose that it is possible in average time $\bigomic(\fval)$ on a classical computer to
	perform the following operations.
	\begin{enumerate}[(a)]
		\item \label{cond:eps1_samp_P}
			Given $m$, sample $n$ from the probability distribution $\Pdist{n|m}$.
		\item \label{cond:eps1_samp_Q}
			Given $n$, sample $m$ from the probability distribution $\Qdist{m|n}$.
		\item \label{cond:eps1_compute}
			Given $m,n$, compute
			$A_{mn}$, $\sum_{n'} \abs{A_{mn'}}$, and $\sum_{m'} \abs{A_{m'n}}$.
	\end{enumerate}
	Then $A$ is $\epsp{\bval}{\fval}$ with
	$\bval = \infnorm{A}^{1/p} \onenorm{A}^{1/q}$.
	Note that $\bval$ is the weighted geometric mean of the maximum row and column sums of $A$.
	\label{thm:eps1}
\end{theorem}
\begin{proof}
	This follows directly from plugging the probability distributions~\eqref{eq:eps1_dist}
	into \cref{def:eps}, with $K=\{0\}$ (i.e.\ not making use of the index
	$k$).  Note that $\infnorm{A}$ is the maximum absolute row sum and $\onenorm{A}$ is the
	maximum absolute column sum of $A$.
\end{proof}

Finally, we present theorems that cover the two most important examples of EHT operators:
dyads and density operators.

\begin{theorem}[Dyads are EHT]
	\label{thm:dyads}
	Let $\ket{\phi}$ and $\bra{\psi}$ be vectors such that the probability distributions
	$\Pdist{n} = \abs{\psi_n}^p/\pnorm{\psi}^p$ and
	$\Qdist{m} = \abs{\phi_m}^q/\qnorm{\phi}^q$ can be sampled from,
	and the corresponding $\psi_n$ and $\phi_m$ can be computed, in average time
	$\bigomic(\fval)$.
	Then the dyad $\ket{\phi}\bra{\psi}$ is $\ehtp{\pnorm{\psi} \qnorm{\phi}}{\fval}$.
\end{theorem}
\begin{proof}
	This can be seen immediately by plugging the given probability distributions
	into \cref{def:eht}, with $K=\{0\}$ (i.e.\ without making use of index $k$).
	This is the best possible value of $\bval$, which can be seen by applying
	\cref{thm:best_eps_qnorm}\eqref{part:best_bval_ge_qnorm} and using
	$\qnorm{(\ket{\phi}\bra{\psi})} = \pnorm{\psi} \qnorm{\phi}$.
\end{proof}

\begin{corollary}[Estimate matrix entries]
	Let $A$ be $\epsp{\bval}{\fval}$.  Then, given any indices $i,j$, the value of the matrix entry
	$A_{ij}$ can be estimated
	to within additive error $\epsilon>0$, with probability less than $\delta>0$ of
	exceeding the error bound, in average time
	$\bigomic(\log(\delta^{-1}) \epsilon^{-2} \bval^2 \fval)$.
\end{corollary}
\begin{proof}
	By \cref{thm:dyads} the dyad of computational basis vectors $\ket{j}\bra{i}$ is
	$\ehtp{1}{\log(N)}$.
	Note: $\fval \ge \log(N)$ in all cases (unless one is dealing with query complexity) since it
	takes $\bigomic(\log(N))$ time to even write down the indices $i$ and $j$, which are
	$\log(N)$ bits long.
	By \cref{thm:eht_eps_trace}, $A_{ij} = \Tr(A \ket{j}\bra{i})$ can be estimated
	in time
	$\bigomic(\log(\delta^{-1}) \epsilon^{-2} \bval^2 [\fval+\log(N)]) =
		\bigomic(\log(\delta^{-1}) \epsilon^{-2} \bval^2 \fval)$.
\end{proof}

\begin{theorem}[Density operators are EHT]
	Let $\sigma$ be a density operator.
	Suppose that it is possible to sample from the probability distribution
	$\Pdist{n} = \sigma_{nn}$ in average time $\bigomic(\fval)$ and, given $i,j$, to compute
	$\sigma_{ij}$ in average time $\bigomic(\fval)$.
	Then $\sigma$ is $\ehttwo{1}{\fval}$.
\end{theorem}
\begin{proof}
	This follows from plugging the probability distributions $\Pdist{n} = \sigma_{nn}$ and
	$\Qdist{m} = \sigma_{mm}$ into \cref{def:eht} and using the inequality
	$\abs{\sigma_{mn}} \le \sqrt{\sigma_{mm} \sigma_{nn}}$, which is satisfied by positive
	semidefinite matrices.
\end{proof}

\section{Simulation of quantum circuits}
\label{sec:circuits}

\subsection{Efficiently simulated states and operators}

\begin{figure}
	\begin{center}
		\includegraphics{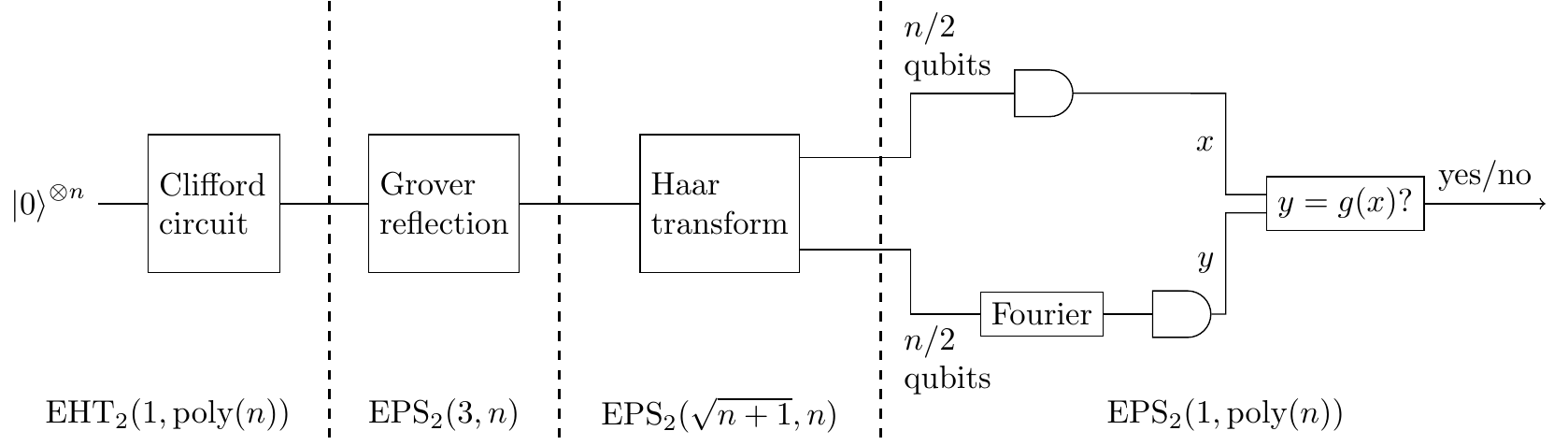}
	\end{center}
	\caption{
		An example of the type of circuit that can be simulated in $\poly(n)$
		time using the techniques of this paper.  The circuit is divided into four
		sections: the first
		section is considered to be the initial state, the middle two sections are unitary
		matrices, and the last section is a projector.
		The block labeled $y=g(x)$ represents a classical computation step that outputs
		``yes'' if the first and second measurement operations result in values that are
		related by an arbitrary (but $\poly(n)$ time computable) function $g$.
	}
	\label{fig:big_circuit}
\end{figure}

In this section we take up the case $p=q=2$, which is relevant to quantum circuits, and list
several examples of $\ehttwo{\bval}{\fval}$ states and $\epstwo{\bval}{\fval}$ operators where
$\bval$ is small and $\fval \le \polylog(N)$ where $N$ is the dimension of the system
(i.e.\ $N=2^n$ where $n$ is the number of qubits).
By \cref{thm:chain_sim}, circuits made of such states and operators can be efficiently
simulated.
For example, the circuit depicted in \cref{fig:big_circuit} can be simulated in
$\polylog(N)$ time.
After providing several examples of such states and operators, we discuss a few circuits that
cannot be efficiently simulated using our technique.

The initial states we are able to efficiently simulate include the \textit{computationally
	tractable} (CT) states of~\cite{vandennest2011}.  We reproduce the definition
here.\footnote{
	Their definition referred to qubits.  We generalize slightly to the abstract case where the
	decomposition into subsystems is not defined, only the total dimension of the space matters.
}

\begin{definition}
	\label{def:ct}
	A normalized state $\ket{\psi}$ of dimension $N$ is called \textit{computationally
	tractable} (CT) if the following conditions hold:
	\begin{enumerate}[(a)]
		\item It is possible to sample in $\polylog(N)$ time with classical means from the
			probability distribution $\Pdist{i} = \abs{\psi_i}^2$.
		\item Upon input of any $i \in \{0, \dotsc, N-1\}$, the coefficient $\psi_i$ can be
			computed in $\polylog(N)$ time on a classical computer.
	\end{enumerate}
\end{definition}

It follows immediately from \cref{thm:dyads} that if $\ket{\psi}$ is a CT state then
$\rho=\ket{\psi}\bra{\psi}$ is $\ehttwo{1}{\polylog(N)}$.
For convenience we present here a brief list of examples of such states
from~\cite{vandennest2011} and refer the reader to their paper for details:
\begin{itemize}
	\item Product states of qubits (we allow also qudits).
	\item Stabilizer states.
	\item States of the form
		$\ket{\psi} = \frac{1}{\sqrt{N}} \sum_{x=0}^{N-1} e^{i\theta(x)} \ket{x}$
		where $e^{i \theta(x)}$ for a given $x$ can be computed in $\polylog(N)$ time.
	\item Matrix product states of polynomial bond dimension.
	\item States obtained by applying a polynomial sized nearest-neighbor matchgate circuit
		to a computational basis state.
	\item States obtained by applying the quantum Fourier transform to a product state.
	\item The output of quantum circuits with logarithmically scaling tree-width acting on
		product input states.
\end{itemize}
We present a list of examples of $\epstwo{\bval}{\fval}$ operators with $\bval$ small and
$\fval \le \polylog(N)$.
All proofs are in appendix~\ref{sec:circuits_proofs}.
\begin{itemize}
	\item If $A$ is $\epsp{\bval}{\fval}$ then
		$I \ot \cdots \ot I \ot A \ot I \ot \cdots \ot I$ is
		$\epsp{\bval}{\max\{\fval, \log^2(N)\}}$
		(\cref{thm:subsystem_eps}).
		In other words, EPS operations on subsystems are EPS\@.
		The $\log^2(N)$ is due to the amount of time needed to
		convert indices of $I \ot \cdots \ot I \ot A \ot I \ot \cdots \ot I$ to
		indices of $A$.
	\item Any operator $A$ on a constant number of qubits or qudits is
		$\epstwo{\ifmax(A)}{1}$ where $\ifmax(A)=\twonorm{\aA}$ is the interference producing
		capacity of $A$.
		In other words, the simulation cost due to such an operator is equal to the fourth
		power of its interference producing capacity (because of the $b_t^4$ term
		in~\eqref{eq:circuitsim_cost}).
	\item If $A$ is an $M \times M$ matrix with maximum singular value bounded by 1 (e.g.\ a
		unitary, projector, or POVM element) then $\ifmax(A) \le \sqrt{M}$.  This inequality is
		saturated when $A$ is a unitary with rows forming a basis mutually unbiased to the
		computational basis (e.g.\ a Hadamard or Fourier transform).
	\item In terms of query complexity rather than time complexity, any operator $A$
		not depending on an oracle is $\epstwo{\ifmax(A)}{0}$ by
		\cref{thm:best_eps_qnorm}.
		Oracles themselves are $\epstwo{1}{1}$.
	\item Efficiently computable sparse matrices as defined in \cite{vandennest2011}
		are $\epsp{\polylog(N)}{\polylog(N)}$ (\cref{thm:ecs_is_eps}).  These include:
		\begin{itemize}
			\item Permutation matrices are $\epsp{1}{\fval}$ as long as the
				permutation and its inverse can be computed in time $\bigomic(\fval)$.
			\item Diagonal unitary matrices are $\epsp{1}{\fval}$ as long as the phases can be
				computed in time $\bigomic(\fval)$.
			\item Pauli matrices are $\epsp{1}{1}$.
		\end{itemize}
	\item Grover reflections $I-2(\ket{+}\bra{+})^{\ot n}$ are $\epstwo{3}{n}$
		(\cref{thm:grover_eps}).
	\item The Haar wavelet transform on $n$ qubits (\cref{def:haar}) is $\epstwo{\sqrt{n+1}}{n}$
		(\cref{thm:haar_eps}).
	\item One dimensional projectors onto CT states are $\epstwo{1}{\polylog(N)}$
		since CT dyads are $\ehttwo{1}{\polylog(N)}$ and EHT operators are EPS
		(\cref{thm:basic_math}).
	\item Rank $r$ projectors onto spaces defined by CT states are $\epstwo{r}{\polylog(N)}$
		(by applying the sum rule \cref{thm:eps_sumprodexp}\eqref{part:eps_sum_simple} to the
		previous item).
	\item Block diagonal matrices where each block is $\epsp{\bval}{\fval}$ are
		$\epsp{\bval}{\fval}$, as long as matrix indices can be converted to/from block indices in
		time $\bigomic(\fval)$ (\cref{thm:block_diag_eps}).
	\item As a special case of block diagonal matrices, projectors of the form $\sum_x
		\ket{x}\bra{x} \ot \ket{\phi_x}\bra{\phi_x}$, where the $\ket{x}$ are computational
		basis states and each $\ket{\phi_x}$ is a CT state, are $\epstwo{1}{\polylog(N)}$.
		Example: given an even number of qubits, measure half of the qubits in the
		computational basis to get $x$, measure
		the other half in the Fourier basis to get $y$, return true if $y=g(x)$ for some
		function $g$ computable in $\polylog(N)$ time (\cref{thm:low_rank_fourier_eps}).  In
		this example, $\ket{\phi_x} = F\ket{g(x)}$.
		The measurement depicted in \cref{fig:big_circuit} is of this form.
\end{itemize}

\subsection{Simulation techniques}

As a matter of convenience, we present a theorem that is essentially a direct corollary of
\cref{thm:chain_sim}, but written in the language of quantum circuits.

\begin{theorem}
	\label{thm:circuitsim}
	Consider a quantum circuit using states of dimension $N$ (i.e.\ $\log_2(N)$ qubits
	or $\log_d(N)$ qudits).
	Let $\ket{\psi}$ be a computationally tractable (CT) state.
	For $t \in \{1,\dotsc,T\}$ let $U^{(t)}$ be an $\epstwo{\bval_t}{\polylog(N)}$ unitary
	and let $M$ be an $\epstwo{\bval_M}{\polylog(N)}$ Hermitian observable.
	It is possible, with probability less than $\delta>0$ of exceeding the error bound, to
	estimate
	\begin{equation}
		\label{eq:circuitsim_braket}
		\braopket{\psi}{U^{(1)\dag} \dotsm U^{(T)\dag} M U^{(T)} \dotsm U^{(1)}}{\psi}
	\end{equation}
	to within additive error $\epsilon>0$ in average time
	\begin{equation}
		\label{eq:circuitsim_cost}
		\bigomic\left( T \log(\delta^{-1}) \epsilon^{-2}
			\polylog(N) \bval_M^2 \prod_{t=1}^T \bval_t^4 \right).
	\end{equation}
	In particular, if $\bval_M$, $\prod_t \bval_t$, and $T$ are $\polylog(N)$, and if
	$\delta$ and $\epsilon$ are constant, then the simulation time is $\polylog(N)$ on average.
\end{theorem}

Note that in~\eqref{eq:circuitsim_cost} each unitary $U^{(t)}$ incurs a cost of $\bval_t^4$
rather than $\bval_t^2$ since it appears twice in~\eqref{eq:circuitsim_braket}.
If $M$ is a rank one projector onto a CT state, $M = \ket{\phi}\bra{\phi}$, then it is much
more efficient to compute~\eqref{eq:circuitsim_braket} as the absolute square of
\begin{equation}
	\Tr\{ \ket{\psi}\bra{\phi} U^{(T)} \dotsm U^{(1)} \}.
	\label{eq:circuitsim_half}
\end{equation}
Since $\ket{\psi}\bra{\phi}$ is $\ehttwo{1}{\polylog(N)}$, and since each unitary only occurs
once, \cref{thm:chain_sim} gives that this expression can be estimated in average time
\begin{equation}
	\bigomic\left( T \log(\delta^{-1}) \epsilon^{-2} \polylog(N) \prod_{t=1}^T \bval_t^2 \right),
\end{equation}
which is much better than~\eqref{eq:circuitsim_cost}.
If $M$ is a low rank projector, the same trick can be used by decomposing $M$ as the
sum of rank one projectors and computing each resulting term individually.
The complexity of such a technique will scale proportional to the rank of $M$.

\Cref{thm:circuitsim} is just an application of \cref{thm:chain_sim} with $p=q=2$.
One may wonder whether other values of $p,q$ would lead to a lower simulation cost.
Ignore for the moment the efficient sampling
conditions~\eqref{cond:eps_samplr} and~\eqref{cond:eps_samprl}
of \cref{def:eps} and \cref{def:eht}.
When estimating~\eqref{eq:circuitsim_braket},
the optimal probability distributions give (by \cref{thm:best_eps_qnorm})
\begin{align}
	\bval :&= \bval_\psi \bval_{U^{(1)}} \dotsm \bval_{U^{(T)}} \bval_M
		\bval_{U^{(T)}} \dotsm \bval_{U^{(1)}} \bval_\psi
	\\ &=
		\pnorm{\psi}
		\qnorm{\matabs{U}^{(1)\dag}}
		\dotsm
		\qnorm{\matabs{U}^{(T)\dag}}
		\qnorm{\matabs{M}}
		\qnorm{\matabs{U}^{(T)}}
		\dotsm
		\qnorm{\matabs{U}^{(1)}}
		\qnorm{\psi}.
\end{align}
This achieves its minimum value at $p=q=2$, since
\begin{align}
	\bval &=
		\pnorm{\psi}
		\qnorm{\psi}
		\pnorm{\matabs{U}^{(1)}}
		\qnorm{\matabs{U}^{(1)}}
		\dotsm
		\pnorm{\matabs{U}^{(T)}}
		\qnorm{\matabs{U}^{(T)}}
		\left(
			\pnorm{\matabs{M}}
			\qnorm{\matabs{M}}
		\right)^{1/2}
		&& \mbox{(using $\qnorm{A^\dag} = \pnorm{A}$)}
	\\ &\ge
		\braket{\psi}{\psi}
		\pnorm{\matabs{U}^{(1)}}
		\qnorm{\matabs{U}^{(1)}}
		\dotsm
		\pnorm{\matabs{U}^{(T)}}
		\qnorm{\matabs{U}^{(T)}}
		\left(
			\pnorm{\matabs{M}}
			\qnorm{\matabs{M}}
		\right)^{1/2}
		&& \mbox{(H{\"o}lder's inequality)}
	\\ &\ge
		\braket{\psi}{\psi}
		\twonorm{\matabs{U}^{(1)}}^2
		\dotsm
		\twonorm{\matabs{U}^{(T)}}^2
		\twonorm{\matabs{M}}
		&& \mbox{(Riesz-Thorin theorem)}
	\\ &=
		\twonorm{\psi}
		\twonorm{\matabs{U}^{(1)\dag}}
		\dotsm
		\twonorm{\matabs{U}^{(T)\dag}}
		\twonorm{\matabs{M}}
		\twonorm{\matabs{U}^{(T)}}
		\dotsm
		\twonorm{\matabs{U}^{(1)}}
		\twonorm{\psi}.
\end{align}
On the other hand, when estimating an expression of the form~\eqref{eq:circuitsim_half}, each
unitary is no longer repeated twice and Riesz-Thorin cannot be applied.  In this case the
minimum value of $\bval$ does not necessarily occur at $p=2$.

Certain algorithms, such as Shor's algorithm, consist of a quantum circuit terminating in a
many-outcome measurement (e.g.\ measurement in the computational basis of several different
qubits) which is then post-processed by a classical computer to produce a final result.
This does not immediately fit into our scheme of estimating expectation values.
However, in the case where the final result is a two-outcome yes/no answer (e.g.\ ``does $N$
have a prime factor in the range $[a, b]$''), the final measurement and classical
post-processing can be combined into a single collective projector or POVM element as follows.
Suppose the final state is measured using a POVM $\{F_i\}$.  A classical post-processing
step then inspects the measurement outcome $i$ and returns ``yes'' or ``no''.  Denote by $R$
the set of measurement outcomes that will result in ``yes''.  The classical post-processing can
be absorbed into the measurement, resulting in the POVM element $F' = \sum_{i \in R} F_i$.
The expectation value of $F'$ gives the probability that a measurement of $\{F_i\}$ would
yield ``yes'' after post-processing.

In some cases $F'$ may be efficiently simulated, a (somewhat contrived) example being the final
stage of the circuit of \cref{fig:big_circuit}.
Note that this example involves a Fourier transform, which by itself cannot be efficiently
simulated by our technique since it has large interference producing capacity.
However, when the Fourier transform is followed by the particular classical
post-processing depicted in \cref{fig:big_circuit}, the resulting composite operator
\textit{can} be efficiently simulated (\cref{thm:low_rank_fourier_eps}).
Shor's algorithm also has a Fourier transform followed by classical post-processing,
however in that case the composite operator (Fourier transform followed by
post-processing) has large interference producing capacity and so \textit{cannot} be
efficiently simulated (by our algorithm).

\subsection{Circuits that our technique can't efficiently simulate}

Many examples of efficiently simulatable circuits can be constructed, but it is probably more
enlightening to instead discuss examples of circuits that cannot be efficiently simulated
using our technique.
Since the efficiency of our technique depends upon choice of basis and on
choice of representation (see \cref{sec:wigner}), a circuit which our technique cannot
simulate efficiently in one basis may be efficiently simulatable in another basis.
In this section we choose to focus only on the computational basis.
That being said, most of the examples in this section have been proved (relative to an
oracle) to have no efficient classical solution.

We cannot efficiently simulate Shor's algorithm.  The reason for this
is that the Fourier transform has high interference producing capacity: the Fourier transform
$F$ on $n$ qubits has $\ifmax(F) = 2^{n/2}$.
Replacing the Fourier transform by the Haar wavelet transform (\cref{fig:shor}) yields a
circuit that can be efficiently simulated, since the Haar transform has low interference
producing capacity, $\ifmax(G_n)=\sqrt{n+1}$.
Note that this circuit no longer factors numbers (and probably does nothing at all useful).
The Fourier and Haar transforms play similar roles in classical signal processing,
with the latter providing spatially localized rather than global information for the high
frequency components.
The fact that replacing the Fourier transform enables efficient classical simulation points to
the Fourier transform as being the source of the quantum speedup in Shor's algorithm (for a
contrasting point of view, see~\cite{arxiv:quant-ph/0611156,PhysRevA.76.060302}).

\begin{figure}
	\begin{center}
		\includegraphics{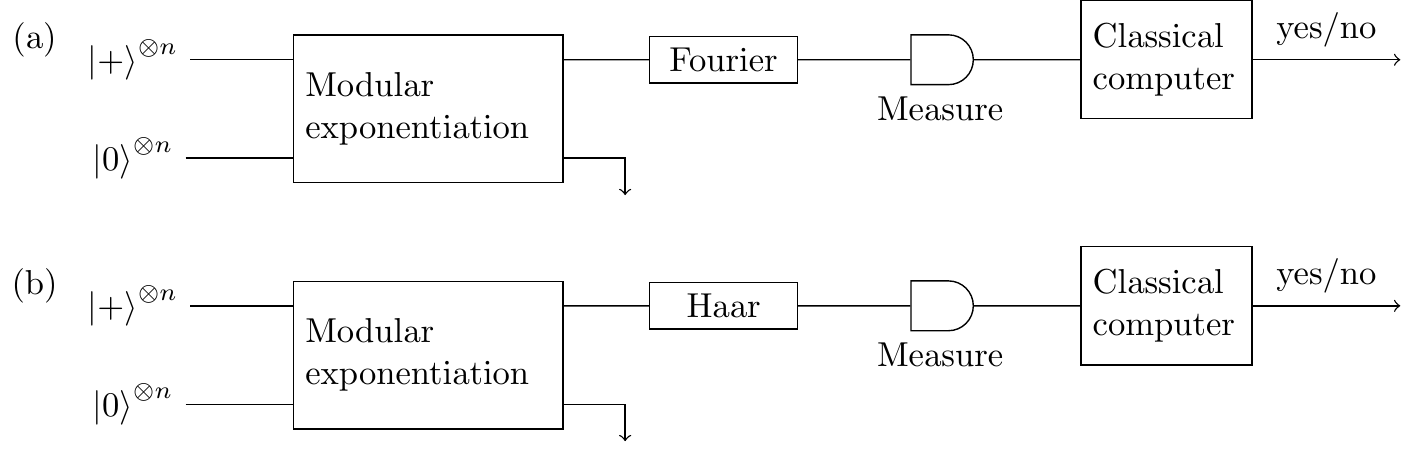}
	\end{center}
	\caption{
		(a) A depiction of the decisional version of Shor's algorithm, which outputs ``yes'' if
		there is a prime factor within some given range.
		(b) The Haar wavelet transform (\cref{def:haar}) plays a similar role as the
		Fourier transform in classical signal processing.  However, substituting the Haar
		transform for the Fourier transform in Shor's algorithm yields a circuit that can be
		efficiently simulated on a classical computer.  Note that the resulting circuit won't
		factor numbers, and in fact probably has no practical use.
	}
	\label{fig:shor}
\end{figure}

%%%%%%%%%%%%%%%%%%%%%%%%%%%%%%%%%%%%%%%%
% Problem: the classical lower bound comes from the fact that the answer is
% $\Omega(\log(n))$ bits long.  But I only want to consider problems with binary output.
% Kimmel cites Zhan as giving an $\Omega(\polylog(n))$ lower bound for
% \textsc{1-Fault Nand Tree}, but that has a weaker promise.
% Probably the parity of the \textsc{Haar Problem} can be easily shown to take
% $\Omega(\log(n))$ queries classically, but I don't feel like showing it.
%%%%%%%%%%%%%%%%%%%%%%%%%%%%%%%%%%%%%%%%
%Since the Haar transform is $\epstwo{\sqrt{n+1}}{n}$ it cannot, in the absence of other
%operations with interference producing capacity, lead to a superpolynomial (in the number of
%qubits) speedup.
%It can, however provide a subpolynomial speedup.  For
%instance,~\cite{springerlink:10.1007/978-3-642-31594-7_47} provides a quantum circuit that
%solves the \textsc{Haar Problem} with $\bigomic(1)$ queries whereas any classical
%algorithm would take $\Omega(\log(n))$ queries.

Deutsch-Jozsa provides an oracle relative to which deterministic quantum computation is more
powerful than deterministic classical computation.
Our algorithm can efficiently simulate the Deutsch-Jozsa algorithm, but not
deterministically.\footnote{
	This was discussed in~\cite{vandennest2011}, which our paper extends.  However, we mention
	it here for completeness.
}
The Deutsch-Jozsa algorithm
consists of an initial CT state $\ket{+}^{\ot n} \ot \ket{-}$,
acted upon by an oracle $\sum_{xy} \ket{x}\bra{x} \ot \ket{y + g(x)}\bra{y}$, followed by
a rank-one projective measurement onto the state $\ket{+}^{\ot n} \ot \ket{-}$.
The initial state is $\ehttwo{1}{n}$ and the operators are
$\epstwo{1}{n}$, so we can efficiently simulate this algorithm.
However, the simulation will always have a small chance of error due to the $\delta$ in
\cref{thm:circuitsim}.

Our simulation algorithm performs very poorly when applied to Grover's algorithm.  Each
iteration of Grover's algorithm consists of an oracle query followed by a Grover reflection.
These operations have low interference producing capacity: 1 for the oracle and just under 3
for the Grover reflection.  However, our algorithm is
exponentially slow in the circuit length, due to the $\prod_t \bval_t^4$ factor
in~\eqref{eq:circuitsim_cost}.  Since the Grover reflection is used $\Theta(\sqrt{N})$
times, the simulation would run in time $\exp(\Theta(\sqrt{N}))$.  Even though each iteration
of Grover's algorithm produces small interference, the total interference of the whole circuit,
by \cref{def:intf_circuit}, is $\exp(\Theta(\sqrt{N}))$.

In \cite{Childs2003} a quantum random walk is presented that provides
an exponential speedup over any possible classical algorithm for the graph traversal problem.
The walk is carried out by evolving the initial state with a Hamiltonian that is defined in
terms of an oracle.
We cannot efficiently simulate this algorithm for the same reason that we cannot
efficiently simulate Grover: the runtime of the quantum algorithm increases with the problem
size, and our simulation must pay an exponentially large penalty for this due to the $\prod_t
\bval_t^4$ factor in~\eqref{eq:circuitsim_cost}.
On the other hand, short time/low energy Hamiltonian evolutions can be efficiently simulated
by our technique.
In particular, \cref{thm:eps_sumprodexp}\eqref{part:eps_exp} gives that if $H$ is
$\epsp{\bval}{\fval}$ then $e^{iHt}$ is $\epsp{e^{\bval t}}{\bval t \fval}$.
In terms of query complexity the Hamiltonian in the algorithm of \cite{Childs2003} is
$\epstwo{\bigomic(1)}{1}$, so we could feasibly simulate $e^{iHt}$ for small $t$.
However, their algorithm has $t=\Theta(n^4)$, so our simulation would have query
complexity $e^{\Theta(n^4)}$, making it unfeasibly slow.

\section{Applications and discussion}
\label{sec:app}

\subsection{Wigner representation}
\label{sec:wigner}

An $N \times N$ matrix can also be viewed as an $N^2$ dimensional vector, so we can write
for instance $\braket{\myvec{M}}{\myvec{\rho}}$ in place of $\Tr\{M \rho\}$.
Superoperators become $N^2 \times N^2$ matrices in this representation, and we can write
$\braopket{\myvec{M}}{\myvec{V}\myvec{U}}{\myvec{\rho}} = \Tr\{ M VU \rho U^\dag V^\dag \}$.
Simulating a quantum circuit using this representation offers an alternative to the customary
representation that was the focus of section~\ref{sec:circuits}.

Any basis can be used (even ones that are not orthonormal), although some choices of basis may
yield more efficient simulation.  One notable choice is given by the discrete Wigner
representation, which is only defined for qudits of odd dimension.
We will not describe the details here but refer the reader
to~\cite{1367-2630-14-11-113011,PhysRevLett.109.230503}
in which it is shown that in the discrete Wigner representation stabilizer
states become probability distributions and Clifford operations become permutation
matrices.

It was shown independently in~\cite{1367-2630-15-1-013037,PhysRevLett.109.230503}
that when operations in the Wigner representation are given by nonnegative
matrices, such matrices are stochastic and therefore can be efficiently simulated.
Our algorithm, taking $p=\infty$ and $q=1$, extends this result by also allowing states and
operations in which the Wigner representation contains a small quantity of negative
values, although ours is weaker in that it only computes expectation values rather than
allowing sampling of a many-outcome measurement.
With $q=1$ rather than $q=2$, the difficulty of simulating an operation is given not by
$\ifmax(A) = \twonorm{\aA}$ but rather by $\onenorm{\aA}=\onenorm{A}$, the maximum
absolute column sum.
In cases where the matrix in the Wigner representation is nonnegative, the matrix will be
left-stochastic and $\onenorm{A}=1$, such matrices will not increase the number of samples
needed.  If there are some negative values then $\onenorm{A}$ will be larger.

After the present work was completed, the
quantity $\log \onenorm{\myvec{\rho}}$ was investigated in~\cite{1367-2630-16-1-013009}.
This quantity was termed ``mana'' and was shown to be monotone under Clifford operations,
and to be monotone on average under stabilizer measurements, thus providing bounds on
magic state distillation by Clifford circuits.
Given the results of the present paper, it should perhaps make sense to extend the
concept of mana also to quantum operations, defining their mana to be
$\log \onenorm{A}$.
Then Clifford operations have zero mana and in general the following monotonicity relation
is satisfied:
\begin{align}
	\log \onenorm{A \myvec{\rho}} \le \log( \onenorm{A} \onenorm{\myvec{\rho}} )
	= \log \onenorm{A} + \log \onenorm{\myvec{\rho}}.
\end{align}
So $\log \onenorm{A}$, which is the Wigner representation analogue of the log of
interference producing capacity, bounds the amount by which the operator $A$ may increase
the mana of a state.
For each $A$ there will be some $\myvec{\rho}$ that saturates this inequality
(by the definition of operator norm), but it is not clear whether this would correspond to
a physical state.

Stated in this language, \cref{thm:chain_sim}, applied in the Wigner representation, gives
that quantum circuits may be efficiently simulated classically in time polynomial in
$\infnorm{\myvec{M}}$ (where $\myvec{M}$ is the final measurement) and exponential in the sum of
the mana of the initial state and the mana of each operation.
Specifically, write $\braopket{\myvec{M}}{\myvec{V}\myvec{U}}{\myvec{\rho}}
= \Tr\{ \ket{\myvec{\rho}}\bra{\myvec{M}} \myvec{V}\myvec{U}\}$.
Then, ignoring for the moment conditions~\eqref{cond:eps_samplr}-\eqref{cond:eps_samprl}
of \cref{def:eps} and~\eqref{cond:eht_samplr}-\eqref{cond:eht_samprl}
of \cref{def:eht}, we have (by \cref{thm:best_eps_qnorm}) that
$\ket{\myvec{\rho}}\bra{\myvec{M}}$ is
$\ehtinf{\onenorm{\myvec{\rho}} \infnorm{\myvec{M}}}{f}$
and $\myvec{U}$ is $\epsinf{\onenorm{\myvec{U}}}{f}$ (similarly for $\myvec{V}$).
So by \cref{thm:chain_sim} this can be simulated in time
\begin{align}
	\bigomic(\log(\delta^{-1}) \epsilon^{-2} \infnorm{\myvec{M}}
	\onenorm{\myvec{U}} \onenorm{\myvec{V}}
	\onenorm{\myvec{\rho}} \fval).
\end{align}
This complements the result of~\cite{1367-2630-16-1-013009} which showed mana to be a
necessary resource for magic state distillation but did not show that circuits of low
total mana have no quantum speedup (although the zero mana case was treated
in~\cite{1367-2630-15-1-013037,PhysRevLett.109.230503}).

\subsection{Communication complexity}
\label{subsec:comm}

Consider a scenario in which two parties, Alice and Bob, are to cooperatively evaluate a
boolean function.  Specifically, suppose that Alice receives input $x$, Bob receives input
$y$, and they are to evaluate $g(x,y)$ where the function $g : X \times Y \to \{0,1\}$ is known
to the two parties ahead of time.
They must provide the correct answer with probability at least $2/3$.
For non-trivial functions this will require communication, which can be either quantum or
classical.
The communication complexity of $g$ is the number of bits of communication required by the
optimal protocol, with no regard for the amount of time Alice and Bob spend on local
computations.
For some problems quantum communication is exponentially more efficient than classical
communication~\cite{Regev:2011:QOC:1993636.1993642}.

\begin{figure}
	\begin{center}
		\includegraphics{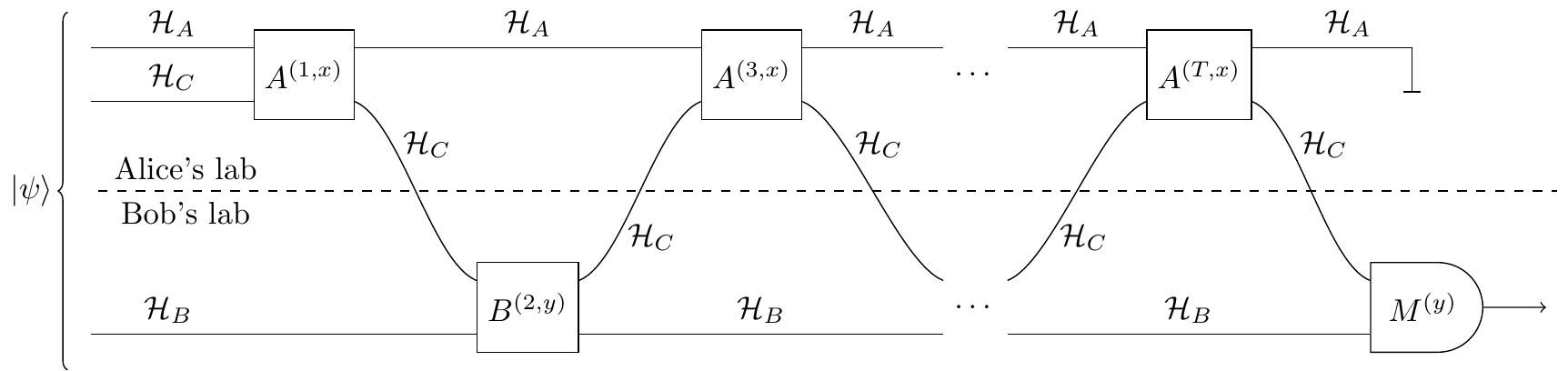}
	\end{center}
	\caption{
		A quantum communication protocol.  The expectation value of the final measurement
		is given by~\eqref{eq:comm_braket}.
	}
	\label{fig:comm}
\end{figure}

Consider a quantum communication protocol as depicted by \cref{fig:comm}.
The initial state, denoted $\ket{\psi}$, is a pure (but possibly entangled) state on three
subsystems $\hilb{A} \ot \hilb{B} \ot \hilb{C}$.
Subsystems $\hilb{A}$ and $\hilb{B}$ are owned by Alice and Bob respectively, and subsystem
$\hilb{C}$ is passed between Alice and Bob through a noiseless quantum channel for each round
of communication.
Alice begins by performing a unitary operation $A^{(1,x)}$, which can depend on her input $x$,
on subsystems $\hilb{A} \ot \hilb{C}$.  She then sends the
$\hilb{C}$ subsystem to Bob, who performs a unitary operation $B^{(2,y)}$, which can depend
on his input $y$, on subsystems
$\hilb{B} \ot \hilb{C}$.
Bob sends $\hilb{C}$ back to Alice who then performs $A^{(3,x)}$ and so on.
Finally, the last party (say, Bob) performs a two outcome projective (or POVM) measurement
$\{M^{(y)}, I-M^{(y)}\}$, which can depend on $y$, on subsystems $\hilb{B} \ot \hilb{C}$ and
reports the outcome.
The expectation value of the final measurement is given by
\begin{equation}
	\braopket{\psi}{A^{(1,x)\dag} B^{(2,y)\dag} A^{(3,x)\dag} \dotsm A^{(T,x)\dag} M^{(y)}
		A^{(T,x)} \dotsm A^{(3,x)} B^{(2,y)} A^{(1,x)}}{\psi}
	\label{eq:comm_braket}
\end{equation}
and must be $\le 1/3$ if $g(x,y)=0$ and $\ge 2/3$ if $g(x,y)=1$.
The communication complexity of the protocol is the number of qubits transmitted,
$T \log(\dim(\hilb{C}))$ where $T$ is the number of rounds of communication.
The dimensionality of the subsystems $\hilb{A}$ and $\hilb{B}$ is not taken into consideration.

The algorithm of this paper can be adapted to provide classical communication simulations of
quantum communication protocols, in the case where the quantum protocols are built using
operators having low interference producing capacity, and making a certain assumption regarding
the initial state $\ket{\psi}$.
Since the expectation value of the final measurement in the quantum protocol will be either
$\le 1/3$ or $\ge 2/3$, a classical simulation of the quantum protocol can with probability
$\ge 2/3$ determine $g(x,y)$ if it can, with chance of error $\delta \le 1/3$, estimate the
expectation value of the quantum protocol to within additive error $\epsilon < 1/6$.
This is exactly the type of estimation provided by the algorithm of this paper, we need only
adapt it to the communication scenario.

The algorithm presented in \cref{subsec:sim_general_pq} involves computing
$\bigomic(\bmax^2)$ path samples,\footnote{
	Specifically, $\bigomic(\log(\delta^{-1}) \epsilon^{-2} \bmax^2)$ samples are needed.
	However, in order to achieve the goal of guessing $g(x,y)$ with probability $\ge 2/3$
	it suffices to set constant $\delta < 1/3$ and $\epsilon < 1/6$.
}
each of which require evaluation of a left-to-right or a right-to-left Markov chain.
Crucially, each transition operator in these chains is defined solely in terms of a single
operator of~\eqref{eq:comm_braket}.  Therefore, each transition can be computed by Alice alone
(for the $A^{(t,x)}$ operators) or by Bob alone (for the $B^{(t,y)}$ and $M^{(y)}$ operators).
The state space of the Markov chains consists of indices corresponding to computational basis
states of $\hilb{A} \ot \hilb{B} \ot \hilb{C}$, so the
indices can be thought of as triples $(i_A, i_B, i_C)$ of indices over
$\hilb{A}$, $\hilb{B}$, and $\hilb{C}$.
Since Alice's operators $A^{(t,x)}$ act only on subsystems $\hilb{A} \ot \hilb{C}$, the
corresponding transition operators in the Markov chain involve only indices $i_A$ and $i_C$.
Similarly, Bob's transition operators involve only $i_B$ and $i_C$.
Therefore Alice and Bob need to communicate only the index $i_C$ for each transition of the Markov
chain.

Also needed is selection of the initial index according to the probability distribution
$\Pdist{i_A, i_B, i_C} = \abs{\braket{i_A,i_B,i_C}{\psi}}^2$ (with Alice getting $(i_A,i_C)$
and Bob getting $i_B$), as well as evaluation of
$\braket{i_A,i_B,i_C}{\psi}$ for a given $(i_A, i_B, i_C)$ triple (where Alice knows
$(i_A,i_C)$ and Bob knows $i_B$).
If the initial state is a product state, $\ket{\psi} = \ket{\psi_{AC}} \ot \ket{\psi_B}$, these
tasks are easily accomplished using no communication.
In fact, even if $\ket{\psi}$ is entangled between Alice and Bob these two tasks can in some
cases be accomplished using only a small amount of communication.
Alice and Bob both know $\ket{\psi}$ (since it does not depend on $x$ or $y$), so they can
individually sample from $\Pdist{i_A, i_B, i_C}$.  If Alice and Bob are granted access to
shared randomness (a.k.a.\ public coins), they can sample from $\Pdist{i_A, i_B, i_C}$ in a
synchronous way (i.e.\ they both get the same outcome).
Computation of $\braket{i_A,i_B,i_C}{\psi}$ for a given $(i_A, i_B, i_C)$ triple, with
$(i_A,i_C)$ known to Alice and $i_B$ known to Bob, is trickier and how much communication is
needed depends on $\ket{\psi}$.
For example, let $\hilb{A} = \hilb{A'} \ot \hilb{A''}$ and
$\hilb{B} = \hilb{B'} \ot \hilb{B''}$ and consider an initial state of the form
\begin{equation}
	\label{eq:initial_entangled}
	\ket{\psi} = \ket{\psi_{A'}} \ot \ket{\psi_{B'}} \ot \ket{\psi_C} \ot
		\sum_i \alpha_i \ket{i}_{A''} \ot \ket{i}_{B''}
\end{equation}
with $\ket{i}_{A''}$ and $\ket{i}_{B''}$ denoting computational basis vectors.
This is the most common type of initial state for quantum protocols that make use of shared
entanglement.
Then
\begin{equation}
	\braket{i_A,i_B,i_C}{\psi} =
		\braket{i_{A'}}{\psi_{A'}}
		\braket{i_{B'}}{\psi_{B'}}
		\braket{i_{C }}{\psi_{C }}
		\alpha_{i_{A''}} \delta(i_{A''}, i_{B''})
\end{equation}
where $\delta$ is the Kronecker delta.
This can be computed using shared randomness and $\bigomic(1)$ communication by making use of a
bounded error protocol for testing equality
of $i_{A''}$ and $i_{B''}$ (example 3.13 of~\cite{KushilevitzNisan200611}).

Since each unitary appears twice in~\eqref{eq:comm_braket},
evaluation of the entire Markov chain is accomplished with twice as much
communication as the classical protocol, or $2 T \log(\dim(\hilb{C}))$ bits.
The algorithm also requires computing the amplitude associated with the path, as well as the
probability of the path.
However, this requires only transmission of $\bigomic(T)$ scalar quantities from Alice to Bob,
using $\bigomic(T)$ bits of communication.\footnote{
	Actually a careful look shows that only $\bigomic(1)$ communication is needed.  Alice can
	locally multiply her transition probabilities and the amplitudes for her operators for the
	given path and report these $O(1)$ values to Bob who is then able to complete the
	computation.
}
The total classical communication complexity of this simulation protocol is therefore
$\bigomic(\bmax^2 T \log[\dim(\hilb{C})])$, a factor $O(\bmax^2)$ greater than that of
the quantum protocol.
Using the optimal probability distributions defined in appendix~\ref{sec:uAv}, $\bmax$ is upper
bounded by the product of the interference producing capacities of the operators
in~\eqref{eq:comm_braket}.  The communication complexity of the classical simulation is then
\begin{equation}
	\bigomic\left(T \log[\dim(\hilb{C})] \max_{x,y} \left\{
		\twonorm{\matabs{A}^{(1,x)}}^4
		\cdot
		\twonorm{\matabs{B}^{(2,y)}}^4
		\cdot
		\twonorm{\matabs{A}^{(3,x)}}^4
		\dotsm
		\twonorm{\matabs{A}^{(T,x)}}^4
		\cdot
		\twonorm{\matabs{M}^{(y)}}^2
	\right\}\right).
\end{equation}

The consequence of this construction is that any quantum communication protocol exhibiting
superpolynomial advantage in communication complexity over any classical protocol must have
a superpolynomial value of $\bmax$ (i.e.\ the product of the interference producing capacities
of the quantum operators must be high) or must make use of an initial state not of the
form~\eqref{eq:initial_entangled}.
There is, however, an interesting caveat to this claim.
Due to the fact that each unitary, as well as the initial state, appears twice
in~\eqref{eq:comm_braket}, our classical simulation will require twice as many communication
rounds as the quantum protocol.\footnote{
	Note that independent evaluations of the Markov chain can be run in parallel, otherwise the
	number of rounds would scale as $\bigomic(\bmax^2)$.
}
Our technique therefore does not apply if one limits the number of rounds.
For example, the quantum protocol for the \textsc{Perm-Invariance} problem described
in~\cite{montanaro2011} has $\bmax = 1$ yet is exponentially more efficient than any one-round
classical protocol.

There is a way to avoid the doubling of the number of rounds of communication, but at a price.
Consider a one-round quantum protocol in which Alice sends a state $\ket{\psi}$ and Bob measures
a projector (or POVM element) $M$.
The expectation value is $\braopket{\psi}{M}{\psi} = \Tr\{ \ket{\psi}\bra{\psi} M \}$.
As described in the previous subsection, the state $\ket{\psi}\bra{\psi}$ and operator $M$ can
be vectorized to give
$\braket{\myvec{\rho}}{\myvec{M}} = \Tr\{ \ket{\psi}\bra{\psi} M \}$.
By taking $p=1$ and $q=\infty$ instead of $p=q=2$ our algorithm can estimate
$\braket{\myvec{\rho}}{\myvec{M}}$ using only a left-to-right Markov chain, thus requiring only
a single round of communication, from Alice to Bob.
However, since $p=1$ and $q=\infty$, the number of bits communicated is
$\bigomic(\onenorm{\myvec{\rho}}^2 \infnorm{\myvec{M}}^2 n)$ with $n$ being the number of
qubits in $\ket{\psi}$.
The reason we can't efficiently simulate the quantum protocol of~\cite{montanaro2011}
using this technique is that $\onenorm{\myvec{\rho}}$ is exponentially large.
Interestingly,~\cite{Kremer:1995:ROC:225058.225277} provides a one-round protocol that can
estimate $\braket{\myvec{\rho}}{\myvec{M}}$ using
$\bigomic(\twonorm{\myvec{\rho}}^2 \twonorm{\myvec{M}}^2)$ bits of classical communication.
However, this again fails to provide an efficient simulation since $\twonorm{\myvec{M}}$ is
exponentially large.

\subsection{Continuity of \texorpdfstring{$\intfc$ and $\ifmax$}{I and Imax}}

Our measures $\ifmax$ of \cref{def:ifmax} (which we have related to quantum speedup) and
$\intfc$ of \cref{def:intf_circuit} (which we have conjectured to be related to quantum
speedup)
are continuous as a function of the states and operators of a circuit.
To our knowledge, this is the first continuous quantity that has been identified as being
a necessary resource for quantum speedup, other resources such as Schmidt
rank~\cite{PhysRevLett.91.147902} or
tree width~\cite{Markov:2008:SQC:1405087.1405105,arxiv:quant-ph/0603163} being discrete valued.

An argument was put forth in~\cite{arxiv:1204.3107} as to why most continuous quantities could
not be considered as a necessary resource for quantum speedup.
Although their argument focuses on functions of the state vector, such as entanglement entropy,
rather than of the operators, it is still worthwhile to examine whether it is
applicable to the present work.
We paraphrase their argument here, modifying it slightly to fit the circuit paradigm that we
have been using in this paper.
Consider a quantum circuit with initial state $\ket{0}^{\ot n}$, followed by several unitaries,
terminated by a final measurement having expectation value $v$.
Add a control to all of the operators in the circuit:
$I \ot \ket{0}\bra{0} + U \ot \ket{1}\bra{1}$ in place of $U$ for each unitary and similarly
for the final measurement.
All operators are controlled by an ancillary qubit initially in the state
$\sqrt{1-\epsilon}\ket{0} + \sqrt{\epsilon}\ket{1}$.
By repeating execution of the circuit $\bigomic(\epsilon^{-2})$
times, the value of $v$ can be recovered to high accuracy.
However, by setting $\epsilon$ to a sufficiently low value, the state at all times during the
computation will be arbitrarily close to $\ket{0}^{\ot n+1}$, and thus will have arbitrarily
low entanglement.
The most commonly used entanglement measures take values that depend polynomially on
$\epsilon$, so entanglement can be made quite low without $\bigomic(\epsilon^{-2})$ growing to
an unfeasible magnitude.
As a consequence, it is not possible to claim without qualification that entanglement is
necessary for quantum speedup.

This construction has no effect on the interference producing capacity of the operators
of the circuit since $\ifmax(I \ot \ket{0}\bra{0} + U \ot \ket{1}\bra{1}) = \ifmax(U)$.
For this reason, our main result regarding $\ifmax$ as a necessary resource for quantum speedup
is immune to the above argument.
On the other hand, the interference measure $\intfc$ of \cref{def:intf_circuit}, which is the
subject of the conjectures of section~\ref{sec:conj}, is immune to this argument for a
different reason.
The value of $\intfc$ can be exponentially high in the number of qubits or number of unitaries
of a circuit.  In order to make $\intfc$ small, $\epsilon$ would have to be exponentially
small, in turn requiring an exponentially large number of repetitions of the circuit.
So the construction of~\cite{arxiv:1204.3107} is not able to significantly lower the
interference of a circuit without also losing the quantum speedup.

\subsection{Connection to decoherence functional}
\label{subsec:decoherence}

There is a close connection between the interference $\intfc$ of \cref{def:intf_circuit} and
the decoherence functional introduced by Gell-Mann and Hartle.%
\footnote{See \cite{PhysRevD.47.3345}.  Here we use the notation of Chs.~7, 8 and 10
  of \cite{Grff02c}, which is more convenient for our purposes because it employs
  the Schr\"odinger rather than the Heisenberg representation.} %
The latter represents an extension of the Born rule so as to be able to define
probabilities for a sequence of events in a closed quantum system.
Consider a \emph{family of histories} corresponding to projection onto the computational basis
at each step
(i.e.\ after the initial state and after each unitary)
of a quantum circuit
$\Tr\{U^{(1)\dag} \dotsm U^{(T)\dag} M U^{(T)} \dotsm U^{(1)} \rho \}$.
In this case the \emph{decoherence functional} is defined as
\begin{equation}
	\DC(\jvec;\kvec) = \Tr[ M W(\jvec) \rho W^\dag(\kvec) ],
\end{equation}
where $\rho$ is the initial state, $M$ is a projector, and
\begin{equation}
	W(\jvec) = \projdyad{j_T} U^{T} \dotsm \projdyad{j_2} U^{(2)} \projdyad{j_1} U^{(1)}
		\projdyad{j_0}.
\end{equation}
It is convenient to think of $\DC(\jvec;\kvec)$ as a matrix with rows labeled
by $\jvec$ and columns by $\kvec$, and then it is not difficult to show that
\begin{equation}
	\sum_{\jvec}\sum_{\kvec} \DC(\jvec;\kvec) =
		\Tr\{U^{(1)\dag} \dotsm U^{(T)\dag} M U^{(T)} \dotsm U^{(1)} \rho \}.
	\label{eq:deco_pathsum}
\end{equation}
If the \emph{consistency condition}
\begin{equation}
	\DC(\jvec;\kvec) = 0 \text{ whenever } \jvec\neq \kvec
	\label{eq:deco_consis}
\end{equation}
is satisfied, then each diagonal element $\DC(\jvec;\jvec)$ can be interpreted (up to
normalization) as the probability of the history corresponding to $\jvec$ occurring.
The sum of these diagonal elements is then equal to the expectation value of the final
observable, the right side of~\eqref{eq:deco_pathsum}, since the off diagonal terms vanish.

It is straightforward to show that $\intfc$ of \cref{def:intf_circuit} is equal to
\begin{equation}
	\label{eq:deco_intf}
	\intfc\left(
		U^{(1)\dag},\dotsc,U^{(T)\dag},M,U^{(T)},\dotsc,U^{(1)},\rho
	\right) =
	\sum_{\jvec}\sum_{\kvec} \abs{\DC(\jvec;\kvec)}.
\end{equation}
When the consistency condition~\eqref{eq:deco_consis} is satisfied, this will be equal to
$\sum_{\jvec}\DC(\jvec;\jvec)$ (since the diagonal entries are always positive), which in turn
is equal to the right hand side of~\eqref{eq:deco_pathsum}.
In general,~\eqref{eq:deco_intf} gives a measure of how badly the consistency condition is
violated.

\section{Conjectures}
\label{sec:conj}

We have shown that quantum speedup requires circuit elements with a large
interference producing capacity.
In this section we formally state our conjecture that low interference (rather than low
interference producing capacity) is sufficient to ensure efficient simulation of a quantum
circuit.
In general we are interested in circuits of arbitrary length, but for concreteness consider
the task of estimating sums of the form
\begin{align}
	\braopket{\psi}{U^\dag M U}{\psi} &= \sum_{ijkl} V(i,j,k,l),
	\label{eq:conj_pathsum}
	\\
	V(i,j,k,l) &= \psi^*_{i} U^\dag_{ij} M_{jk} U_{kl} \psi_{l}.
\end{align}
As discussed in section~\ref{sec:monte}, this sum can be estimated by considering a number of
randomly chosen paths $\pi=(i,j,k,l)$.  If these paths are chosen according to the optimal
probability distribution $\Rdistopt{\pi}$ of~\eqref{eq:holistic_dist} then the number of
samples required to estimate~\eqref{eq:conj_pathsum} to within error $\epsilon$ (with probability
$\delta$ of exceeding this error bound) is
$\bigomic(\log(\delta^{-1}) \epsilon^{-2} \intfc^2)$ where
$\intfc=\braopket{\matabs{\psi}}{\matabs{U^\dag} \matabs{M} \matabs{U}}{\matabs{\psi}}$
is the interference of the circuit as given by \cref{def:intf_circuit}.
The difficulty with this strategy is that we do not know how to efficiently sample paths
according to the distribution $\Rdistopt{\pi}$, or anything sufficiently close to it.
In other words, we do not have a strategy for finding the most relevant paths.
However, we conjecture that there is a way.

Loosely speaking, we conjecture that a quantum circuit can be simulated in time
$\poly(\log(\delta^{-1}) \epsilon^{-1} \intfc)$ as long as the initial state and operators meet
some computational tractability conditions, analogous to conditions~\eqref{cond:eps_samplr}
and~\eqref{cond:eps_samprl} of \cref{def:eps,def:eht}.
Exactly what tractability conditions should be required is difficult to know ahead of time
for the following reason.
In sections~\ref{sec:monte} and~\ref{sec:markov} a simulation algorithm was developed,
which required certain tasks to be performed involving the initial state and the operators
of the circuit being simulated.
The need to efficiently perform these tasks led directly to the definition of
conditions~\eqref{cond:eps_samplr} and~\eqref{cond:eps_samprl}.
Now we conjecture a better algorithm, whose specific structure is not known ahead of time.
Not knowing the specifics of this conjectured algorithm, it is not clear what should be
required in place of conditions~\eqref{cond:eps_samplr} and~\eqref{cond:eps_samprl}.
The intuition is that we assume any necessary task involving any individual operator in
the circuit can be efficiently performed, but we make no assumption regarding the interactions
between several operators.

This can be made more precise.
\Cref{subsec:query} (on query complexity) and \cref{subsec:comm} (on communication complexity)
each provided a framework in which the computational tractability
conditions~\eqref{cond:eps_samplr} and~\eqref{cond:eps_samprl} were not relevant.
We could use either of these to form a conjecture that avoids the need to state similar
conditions.
Of these two, communication complexity is representative of a certain algorithmic structure.
Consider algorithms that involve dealing with the elements of a circuit one at a time.
For instance, when estimating~\eqref{eq:conj_pathsum} one could imagine
carrying out some calculations involving $\ket{\psi}$, making notes of the result, carrying out
further calculations involving $U$, and so on.
The time complexity of such an algorithm is lower bounded by the amount of notes taken and the
number of times attention is shifted from one circuit element to another.
This can be quantified by imagining that each of $\ket{\psi}$, $U$, and $M$ are stored in
separate rooms, and considering how many notes need to be carried back and forth between the
rooms by somebody who seeks to estimate~\eqref{eq:conj_pathsum}.
Equivalently, stated in terms of communication complexity, imagine that Alice has $\ket{\psi}$,
Bob has $U$, and Charlie has $M$.  How much communication is needed in order to
estimate~\eqref{eq:conj_pathsum}?  We conjecture that the amount of communication needed is
polynomial in the interference of the circuit:

\begin{conjecture}
	\label{conj:intf_braket}
	Suppose that Alice has a classical description of a vector $\ket{\psi}$ of dimension $N$,
	Bob has a description of an $N \times N$ POVM element $M$, and $T$ other
	parties have descriptions of $N \times N$ unitary matrices $U^{(1)}, \dotsc, U^{(T)}$.
	Then, with probability less than $\delta$ of exceeding the error bound, the value
	of
	\begin{equation}
		\braopket{\psi}{U^{(1)\dag} \dotsm U^{(T)\dag} M U^{(T)} \dotsm U^{(1)}}{\psi}
		\label{eq:conj_braopket}
	\end{equation}
	can be estimated to within additive error $\epsilon$ using
	$\poly( \log(\delta^{-1}) \epsilon^{-1} \max\{1, \intfc\} \log(N))$
	bits of classical communication where $\intfc$ is the interference
	of~\eqref{eq:conj_braopket} as given by \cref{def:intf_circuit}.
\end{conjecture}

The reader may worry that this communication scenario has little bearing on the problem of
simulating quantum circuits, however it is expected that
any proof in the positive of this conjecture will be adaptable into an
algorithm that can be used in the computation context.
Indeed, the Markov chain technique of section~\ref{sec:markov} was first developed as
as solution to a problem resembling \cref{conj:intf_braket}.

We have been unable to prove this conjecture even for the simple case where there are
no unitary operations and the goal is to estimate the expectation value
$\braopket{\psi}{M}{\psi}$.
We present this simplified case formally, as it deserves some discussion.

\begin{conjecture}
	\label{conj:intf_simple}
	\Cref{conj:intf_braket} holds in the case $T=0$.
	In other words,
	suppose that Alice has a classical description of a vector $\ket{\psi}$ of dimension $N$
	and Bob has a classical description of an $N \times N$ POVM element $M$.
	Then, with probability less than $\delta$ of exceeding the error bound, the value
	$\braopket{\psi}{M}{\psi}$
	can be estimated to within additive error $\epsilon$ using
	$\poly\left( \log(\delta^{-1}) \epsilon^{-1} \max\{1, \intfc\} \log(N) \right)$
	bits of classical communication
	where $\intfc = \braopket{\matabs{\psi}}{\matabs{M}}{\matabs{\psi}}$ is the interference
	of $\braopket{\psi}{M}{\psi}$ as given by \cref{def:intf_circuit}.
\end{conjecture}

\Cref{conj:intf_simple}, being weaker than \cref{conj:intf_braket}, should be easier to prove
true.
However, it would probably be very difficult to prove false since a proof that estimating
$\braopket{\psi}{M}{\psi}$ requires a large amount of classical communication in the general
case (not assuming low interference) remained open for 11
years~\cite{Regev:2011:QOC:1993636.1993642}.

\Cref{conj:intf_simple} would be false if only one round of communication was allowed, from
Alice to Bob.
In~\cite{montanaro2011} the \textsc{Perm-Invariance} problem was defined and shown to be solved
efficiently by a one-round quantum protocol, however no efficient one-round classical protocol
exists.
The quantum protocol has Bob measuring a POVM element $M$ on a state $\ket{\psi}$ sent by
Alice and this protocol is low interference,
$\intfc = \braopket{\matabs{\psi}}{\matabs{M}}{\matabs{\psi}} \le 1$.
However, there can be no efficient one-round classical protocol for estimating
$\braopket{\psi}{M}{\psi}$, since such a protocol would efficiently solve
\textsc{Perm-Invariance}.
This does not provide a counterexample to \cref{conj:intf_simple} since we allow multiple
rounds of communication, and there is indeed an efficient classical two round protocol, which
can be constructed using the technique of~\cref{subsec:comm}.

A potential problem with \cref{conj:intf_braket} is that the unitary portion of the
circuit could create very large interference which could be masked by the final measurement.
For example, consider the initial state $\ket{\psi} = \ket{0}^{\ot n}$, acted upon
by an arbitrary circuit involving all but the first qubit, followed by measurement of the
observable $M = \ket{1}\bra{1} \ot I^{\ot n-1}$.
For this circuit $\intfc = 0$ so \cref{conj:intf_braket} says the expectation value
can be computed in $\poly(\log(\delta^{-1}) \epsilon^{-1} n)$ time, as indeed
it can in this case.
However, it seems there may be similar situations in which $\intfc$ is small because of
the final measurement, but the circuit is nevertheless difficult to simulate.
For this reason we provide an alternate definition that quantifies the interference just
before the final measurement, computed by substituting $M=I$ in \cref{def:intf_circuit}.
This will be used to form a weaker conjecture.

\begin{definition}
	\label{def:intf_state}
	The \textit{interference} of a quantum circuit without a measurement,
	$U^{(T)} \dotsm U^{(1)} \rho U^{(1)\dag},\dotsc,U^{(T)\dag}$, is
	\begin{equation}
		\intfs(U^{(T)},\dotsc,U^{(1)},\rho) =
			\Tr\left\{
				\matabs{U}^{(T)} \dotsm \matabs{U}^{(1)}
				\matabs{\rho}
				\matabs{U}^{(1)\dag} \dotsm \matabs{U}^{(T)\dag}
			\right\}.
	\end{equation}
\end{definition}

In other words, $\intfs$ is the amount by which normalization is
spoiled when destructive interference is turned into constructive interference by means of the
absolute value applied to each path.
This is nondecreasing in time,
\begin{equation}
	\intfs(U^{(T)},\dotsc,U^{(1)},\rho) \ge \intfs(U^{(T-1)},\dotsc,U^{(1)},\rho)
\end{equation}
and $\intfs=1$ if all of the unitaries are permutation matrices as in a classical
computation.
We conjecture that a circuit can be efficiently simulated when $\intfs$ is small.
Since $\intfs$ doesn't see the final measurement $M$, we need an extra constraint.
We require $M$ to be a projector diagonal in the computational basis.

\begin{conjecture}
	Suppose that Alice has a classical description of a vector $\ket{\psi}$ of dimension $N$,
	Bob has a description of an $N \times N$ projector $M$ that is diagonal in the
	computational basis, and $T$ other
	parties have descriptions of $N \times N$ unitary matrices $U^{(1)}, \dotsc, U^{(T)}$.
	Then, with probability less than $\delta$ of exceeding the error bound, the value of
	\begin{equation}
		\braopket{\psi}{U^{(1)\dag}\dotsm U^{(T)\dag} M U^{(T)} \dotsm U^{(1)}}{\psi}
		\label{eq:conj_J_braopket}
	\end{equation}
	can be estimated to within additive error $\epsilon$ using
	$\poly( \log(\delta^{-1}) \epsilon^{-1} \intfs \log(N))$
	bits of classical communication where
	$\intfs = \intfs(U^{(T)},\dotsc,U^{(1)},\ket{\psi}\bra{\psi})$ is the interference
	of~\eqref{eq:conj_J_braopket}
	just before the final measurement, as given by \cref{def:intf_state}.
\end{conjecture}

\section{Summary and open problems}

We have provided an algorithm for efficiently simulating quantum circuits in which each
operator has low interference producing capacity.
Therefore, interference producing capacity is identified as a resource necessary for quantum
speedup.
The runtime of the simulation is quadratic in the interference producing capacities of each
operator, so it is typically exponentially slow in the length of the circuit.
However, for constant length circuits making use of operators with low interference producing
capacity (many such operators are listed in section~\ref{sec:circuits}), the simulation runs in time
polynomial in the number of qubits.

In general, our technique is able to estimate
expressions of the form $\braopket{\psi}{A \dotsm Z}{\phi}$,
of which quantum circuits
$\braopketsmall{\psi}{U^{(1)\dag} \dotsm U^{(T)\dag} M U^{(T)} \dotsm U^{(1)}}{\psi}$
are a special case, in time proportional to
$\pnorm{\psi}^2 \qnorm{\matabs{A}}^2 \dotsm \qnorm{\matabs{Z}}^2 \qnorm{\phi}^2$
for any $1/p+1/q=1$
where a bar over a vector or operator denotes entrywise absolute value in the computational
basis, and where $\pnorm{\cdot}$ denotes the $\ell^p$-norm for vectors and the induced norm for
operators.
The choice $p=q=2$ is most relevant for quantum mechanics, and $\twonorm{\aA}$ gives the
interference producing capacity of $A$.
The technique was also generalized to expressions of the form
$\Tr\{ A \dotsm Z \sigma \}$.

We formalized the conditions necessary for efficient simulation by
introducing two definitions: EHT for the initial state $\sigma$ and EPS for the operators
$A,\dotsc,Z$.
These definitions consist of requirements having to do with the number of
samples needed as well as requirements having to do with efficient computability.
The latter requirements can for the most part be ignored if one is concerned with query
complexity or communication complexity rather than time complexity.
A wide range of initial states and operators are EHT or EPS; many examples were listed in
section~\ref{sec:circuits}.
In addition to discussing circuits which can be efficiently simulated, we gave several examples
of circuits which we cannot efficiently simulate, and explained why.

The choice $p=q=2$ makes the most sense for simulating expressions of the form
$\braopketsmall{\psi}{U^\dag V^\dag M V U}{\psi}$.
However, using the Wigner representation this expression can also be written as
$\braopket{\myvec{M}}{\myvec{V}\myvec{U}}{\myvec{\rho}}$,
and here the choice $p=\infty$ and $q=1$ works well, allowing efficient simulation of circuits
that consist mainly of Clifford operations.
We showed how our simulation technique can be applied to communication problems, with the
conclusion that there can be no superpolynomial advantage of quantum communication over
classical communication unless the quantum protocol uses operations with high interference
producing capacity.
Curiously, this result does not apply to one-round communication, since our simulation requires
doubling the number of rounds.  And indeed, there is an example of a one-round quantum protocol
with low interference producing capacity which is exponentially more efficient than any
one-round classical protocol.

Finally, we would like to suggest three open questions:
\begin{enumerate}[1)]
	\item Can it be shown that interference, rather than interference producing capacity, is
		necessary for quantum speedup?
		In section~\ref{sec:conj} we formalized a series of conjectures on this topic, using the
		framework of communication complexity.
	\item While we have shown interference producing capacity to be a necessary resource for
		quantum speedup, it is also fruitful to investigate sufficient resources for quantum
		speedup.
		For example~\cite{arxiv:1010.3654}, building on the work
		of~\cite{Aaronson:2010:BPH:1806689.1806711}, showed that any
		operator $U$ having the property that $\max_{ij} \abs{U_{ij}}$ is sufficiently small
		can be used to exhibit exponential quantum speedup.
		Can the gap between necessary (e.g.\ our result) and the sufficient
		(e.g.~\cite{arxiv:1010.3654}) conditions for quantum speedup be narrowed?
	\item Can our technique be combined with existing Monte Carlo or other
		techniques to provide an improved simulation algorithm for systems of physical
		interest?  Our algorithm in its present form is not likely to be more efficient than
		existing techniques for such problems.
\end{enumerate}

\section{Acknowledgments}

The author thanks Robert Griffiths and Scott Cohen for many helpful comments and suggestions.
This research received financial support from the National
Science Foundation through Grant PHY-1068331.

\appendix

\section{Generalized singular vectors}
\label[secinapp]{sec:uAv}

The goal of this appendix is to determine the minimum value of $\bval$ such that a given operator
$A$ is $\epsp{\bval}{\fval}$ and bounds on $\bval$ such that an operator $\sigma$ is
$\ehtp{\bval}{\fval}$.
We will show that
conditions~\eqref{cond:eps_alphasum}-\eqref{cond:eps_cost} of \cref{def:eps}
require $\bval \ge \qnorm{\aA}$ and will construct probability distributions that satisfy this
with equality.
Whether these also satisfy
conditions~\eqref{cond:eps_samplr}-\eqref{cond:eps_samprl} of \cref{def:eps}
needs to be determined on a case by case basis.
Note that when $p=q=2$ we have $\twonorm{\aA} = \ifmax(A)$, the interference producing
capacity of $A$.
The end result of this appendix is the following theorem.\footnote{
	In the case $p=q=2$, claims~\eqref{part:best_bval_ge_qnorm}
	and~\eqref{part:best_bval_eq_qnorm} of \cref{thm:best_eps_qnorm} are similar to results
	of~\cite{Mathias1990269}, although the techniques are different.
}

\begin{theorem}
	\label{thm:best_eps_qnorm}
	Let $A$ and $\sigma$ be matrices, $p,q \in [1,\infty]$, and $1/p+1/q=1$.  Then
	\begin{enumerate}[(a)]
		\item \label{part:best_bval_ge_qnorm}
			It is not possible to satisfy
			conditions~\eqref{cond:eps_alphasum}-\eqref{cond:eps_cost} of
			\cref{def:eps} unless $\bval \ge \qnorm{\aA}$.
			The same goes for
			\eqref{cond:eht_alphasum} and~\eqref{cond:eht_cost} of
			\cref{def:eht} since they are stricter (i.e.\ $\bval \ge \qnorm{\matabs{\sigma}}$).
		\item \label{part:best_bval_eq_qnorm}
			It is possible to satisfy
			conditions~\eqref{cond:eps_alphasum}-\eqref{cond:eps_cost} of
			\cref{def:eps} with $\bval = \qnorm{\aA}$.
			The $k$ index is not needed (i.e.\ $k \in K = \{0\}$ and
			$\alpha_{mnk} = A_{mn}$).
		\item \label{part:best_query}
			If one is concerned with
			query complexity rather than time complexity, and if $A$ is not defined in terms of
			an oracle, then conditions~\eqref{cond:eps_samplr}-\eqref{cond:eps_samprl} of
			\cref{def:eps} can be ignored, as explained in \cref{subsec:query}.
			Therefore, $A$ is $\epsp{\qnorm{\aA}}{0}$.
		\item \label{part:best_dyads}
			Let $w$ be the smallest value such that $\sigma/w$ is a convex combination of
			normalized dyads.
			That is to say, let
			\begin{equation}
				\label{eq:gen_tr_norm}
				w = \min\left\{ \sum_i \abs{s_i}
					\middle|
					s_i \in \mathbb{C},
					\sigma = \sum_i s_i \vvec^{(i)} \uvec^{(i)\top},
					\pnorm{\uvec^{(i)}} = \qnorm{\vvec^{(i)}} = 1
				\right\}.
			\end{equation}
			It is possible to satisfy
			conditions~\eqref{cond:eht_alphasum}-\eqref{cond:eht_cost} of
			\cref{def:eht} with $\bval = w$
			(although this is not necessarily the smallest possible value of $\bval$).
			The $k$ index is not needed (i.e.\ $k \in K = \{0\}$ and
			$\alpha_{mnk} = \sigma_{mn}$).
			Note that when $p=q=2$, $w$ is the trace norm of $\sigma$.
		\item \label{part:best_dyads_query}
			If one is concerned with query complexity rather than time complexity, and if
			$\sigma$ is not defined in terms of an oracle, then
			conditions~\eqref{cond:eps_samplr}-\eqref{cond:eps_samprl} of \cref{def:eht} can be
			ignored.
			Therefore, $\sigma$ is $\ehtp{w}{0}$ (although this is not necessarily the smallest
			possible value of $\bval$).
	\end{enumerate}
\end{theorem}

We present immediately the proof of parts~\eqref{part:best_bval_ge_qnorm},
\eqref{part:best_dyads}, and~\eqref{part:best_dyads_query}.
Parts~\eqref{part:best_bval_eq_qnorm} and~\eqref{part:best_query}
will require more preliminary discussion.

\begin{proof}[Proof of \cref{thm:best_eps_qnorm}\eqref{part:best_bval_ge_qnorm}]
	Let $A$ be an $M \times N$ matrix.
	Suppose
	conditions~\eqref{cond:eps_alphasum}-\eqref{cond:eps_cost} of
	\cref{def:eps} are satisfied by some
	$\bval$, $K$, $\alpha_{mnk}$, $\Pdist{n,k|m}$, and $\Qdist{m,k|n}$.
	Then, for all $m \in \{1,\dotsc,M\}$, $n \in \{1,\dotsc,N\}$, and $k \in K$, we have
	$A_{mn} = \sum_{k' \in K} \alpha_{mnk'}$ and
	\begin{align}
		\frac{\abs{\alpha_{mnk}}}{
			\Pdist{n,k|m}^{1/p}
			\Qdist{m,k|n}^{1/q}
		} \le \bval.
	\end{align}
	Rearranging this expression yields
	\begin{equation}
		\label{eq:eps_cost_nofrac}
		\abs{\alpha_{mnk}} \le \bval \cdot
			\Pdist{n,k|m}^{1/p}
			\Qdist{m,k|n}^{1/q}.
	\end{equation}
	Let $\uvec$ and $\vvec$ be nonnegative vectors satisfying $\pnorm{\uvec}=\qnorm{\vvec}=1$ and
	$\vecprod{\uvec}{\aA \vvec} = \qnorm{\aA}$ (that such vectors exist is well known,
	but is also a consequence of \cref{thm:uAv}).
	Multiply both sides of~\eqref{eq:eps_cost_nofrac} by $u_m v_n$ and sum over $m,n,k$ to get
	\begin{align}
		\sum_{mnk} u_m \abs{\alpha_{mnk}} v_n
		&\le \bval \sum_{mnk}
			u_m
			\Pdist{n,k|m}^{1/p}
			\Qdist{m,k|n}^{1/q}
			v_n
		\\ &= \bval \sum_{mnk}
			\left[ \Pdist{n,k|m} u_m^p \right]^{1/p}
			\left[ \Qdist{m,k|n} v_n^q \right]^{1/q}
		\\ \label{eq:b_ge_qnorm_amgm}
		&\le \bval \sum_{mnk} \left[
			\frac{1}{p} \Pdist{n,k|m} u_m^p +
			\frac{1}{q} \Qdist{m,k|n} v_n^q
		\right]
		\\ &= \bval
			\sum_m \frac{1}{p} u_m^p +
			\sum_n \frac{1}{q} v_n^q
		\\ &= \bval( 1/p + 1/q)
		\\ &= \bval
	\end{align}
	where~\eqref{eq:b_ge_qnorm_amgm} follows from the inequality of arithmetic and geometric
	means.
	We now place a lower bound on the left hand side.  By the triangle inequality,
	$\sum_k \abs{\alpha_{mnk}} \ge \abs{\sum_k \alpha_{mnk}} = \abs{A_{mn}}$
	for all $m,n$.  Since $\uvec$ and $\vvec$ are
	nonnegative,
	\begin{align}
		\bval &\ge \sum_{mnk} u_m \abs{\alpha_{mnk}} v_n
		\\ &\ge \sum_{mn} u_m \abs{A_{mn}} v_n
		\\ &= \qnorm{\aA}.
	\end{align}
\end{proof}

\begin{proof}[Proof of
	\cref{thm:best_eps_qnorm}\eqref{part:best_dyads}-\eqref{part:best_dyads_query}]
	Let $\sigma$ be an $M \times N$ matrix.
	Let $s_i$, $\uvec^{(i)}$, and $\vvec^{(i)}$ take values achieving the minimum
	in~\eqref{eq:gen_tr_norm}.
	By absorbing phase into $\uvec^{(i)}$ we can assume that the $s_i$ are positive.
	We then have $w = \sum_i s_i$, $\pnorm{\uvec^{(i)}} = \qnorm{\vvec^{(i)}} = 1$,
	and $\sigma = \sum_i s_i \vvec^{(i)} \uvec^{(i)\top}$.
	Define
	\begin{align}
		P(n) &= \sum_i \frac{s_i}{w} \abs{u^{(i)}_n}^p,
		\\
		Q(m) &= \sum_i \frac{s_i}{w} \abs{v^{(i)}_m}^q.
	\end{align}
	Since $\uvec^{(i)}$ and $\vvec^{(i)}$ are normalized for all $i$, and since
	$\sum_i s_i/w = 1$, these $P(n)$ and $Q(m)$ are convex combinations of probability
	distributions and hence are probability distributions themselves.

	For any $m \in \{1,\dotsc,M\}, n \in \{1,\dotsc,N\}$, H{\"o}lder's inequality gives
	\begin{align}
		&\hspace{0.9cm}
			\sum_i
				\frac{s_i^{1/p}}{w^{1/p}} \abs{u^{(i)}_n} \cdot
				\frac{s_i^{1/q}}{w^{1/q}} \abs{v^{(i)}_m}
			\le
				\left[ \sum_i \left( \frac{s_i^{1/p}}{w^{1/p}}
						\abs{u^{(i)}_n} \right)^p \right]^{1/p}
				\left[ \sum_i \left( \frac{s_i^{1/q}}{w^{1/q}}
						\abs{v^{(i)}_m} \right)^q \right]^{1/q}
		\\ &\implies
			\sum_i
				\frac{s_i}{w}
				\abs{u^{(i)}_n v^{(i)}_m}
			\le
				\left[ \sum_i \frac{s_i}{w} \abs{u^{(i)}_n}^p \right]^{1/p}
				\left[ \sum_i \frac{s_i}{w} \abs{v^{(i)}_m}^q \right]^{1/q}
		\\ &\implies
			\abs{\sum_i
				\frac{s_i}{w}
				u^{(i)}_n v^{(i)}_m}
			\le P(n)^{1/p} Q(m)^{1/q}
		\\ &\implies
			\frac{\abs{\sigma_{mn}}}{w} \le P(n)^{1/p} Q(m)^{1/q}
		\\ &\implies
			\frac{\abs{\sigma_{mn}}}{ P(n)^{1/p} Q(m)^{1/q} } \le w
	\end{align}
	Therefore conditions~\eqref{cond:eht_alphasum}-\eqref{cond:eht_cost} of \cref{def:eht}
	are satisfied with $\alpha_{mn0}=\sigma_{mn}$ and $\bval=w$.

	If one is concerned with query complexity rather than time complexity, and if $\sigma$ is not
	defined in terms of an oracle, then
	conditions~\eqref{cond:eht_samplr}-\eqref{cond:eht_samprl} of \cref{def:eht}
	are satisfied trivially with $\fval=0$ since no oracle queries are needed in order to carry out
	the required operations.  So $\sigma$ is $\ehtp{w}{0}$.
\end{proof}

We now begin construction of the probability distributions satisfying
conditions~\eqref{cond:eps_alphasum}-\eqref{cond:eps_cost} of
\cref{def:eps} with $\bval = \qnorm{\aA}$.
The bulk of the discussion concerns the $p \in (1,\infty)$ case; the reader interested only in
$p=1$ or $p=\infty$ may skip directly to the second half of the proof of
\cref{thm:best_eps_qnorm}\eqref{part:best_bval_eq_qnorm}-\eqref{part:best_query}
at the end of this section.

It suffices to let $k$ take only a single value, say $k=0$, and to set $\alpha_{mn0} = A_{mn}$.
Making this simplification, and plugging in the desired bound
$\bval=\qnorm{\aA}$, conditions~\eqref{cond:eps_alphasum}-\eqref{cond:eps_cost} of
\cref{def:eps} become
\begin{equation}
	\label{eq:cost_simple}
	\max_{mn} \left\{
	\frac{\abs{A_{mn}}}{
		\Pdist{n|m}^{1/p}
		\Qdist{m|n}^{1/q}
	} \right\} \le \qnorm{\aA}.
\end{equation}
It will be convenient to derive the probability distributions from a pair of vectors.
With $A$ being an $M \times N$ matrix, let $\uvec$ be a positive vector of dimension $M$ and
let $\vvec$ be a positive vector of dimension $N$.  Taking the probability distributions
\begin{align}
	\Pdist{n|m} &= \abs{A_{mn}} v_n / [\aA  \vvec]_m, \\
	\Qdist{m|n} &= \abs{A_{mn}} u_m / [\aAT \uvec]_n
\end{align}
brings~\eqref{eq:cost_simple} to the form
\begin{equation}
	\label{eq:cost_uv}
	\max_{mn} \left\{
		\left( \frac{ [\aA \vvec]_m }{ v_n } \right)^{1/p}
		\left( \frac{ [\aA^\top \uvec]_n }{ u_m } \right)^{1/q}
	\right\} \le \qnorm{\aA}.
\end{equation}

Consider for a moment the case $p=q=2$.
If $\aA$ is not block diagonal (even under permutations of rows and columns) then the left and
right singular vectors of $\aA$ will be positive.  Taking these for $\uvec$ and $\vvec$ it is
easy to see that~\eqref{eq:cost_uv} holds.
If $p \ne 2$ we can use a sort of generalization of singular vectors: we will show the
existence of positive vectors satisfying
\begin{align}
	(\matabs{A}^\top \uvec)_n &\le v_n^{q/p} \qnorm{\matabs{A}},
	\label{eq:generalized_sv_u}
	\\
	(\matabs{A}      \vvec)_m &\le u_m^{p/q} \qnorm{\matabs{A}}.
	\label{eq:generalized_sv_v}
\end{align}
These vectors are easily seen to satisfy~\eqref{eq:cost_uv}.
If $\aA$ is not block diagonal then $\uvec$ and $\vvec$ can be computed using
the power method~\cite{Boyd197495,Bhaskara:2011:AMP:2133036.2133076} since $\aA$ is nonnegative.
In this case the inequalities~\eqref{eq:generalized_sv_u}-\eqref{eq:generalized_sv_v} become
equalities.
On the other hand, if $\aA$ is block diagonal then $\uvec$ and $\vvec$ can be built from the
generalized left and right singular vectors of each block.
The rest of this section is devoted to proving the existence of such vectors.

First we will need some basic facts about $\ell^p$-norms.
If $\vvec$ is a real vector normalized under the $\ell^2$-norm then $\uvec=\vvec$ is the unique
$\ell^2$-normalized vector with the property that $\vecprod{\uvec}{\vvec}=1$.
This generalizes to arbitrary $\ell^p$-norms, with some adaptation.

\begin{definition}
	Let $p,q \in [1,\infty]$ and $1/p + 1/q = 1$.
	Let $\vvec \in \ell^q$.
	Any $\uvec \in \ell^p$ satisfying the conditions
	$\vecprod{\uvec}{\vvec} = \qnorm{\vvec}$ and
	$\pnorm{\uvec} = 1$ is called a \textit{support functional} of $\vvec$.
\end{definition}

\begin{lemma}
	\label{thm:support_fcn}
	Let $p,q \in (1,\infty)$ and $1/p + 1/q = 1$.
	For any nonzero $\vvec \in \ell^q$, the vector $\uvec \in \ell^p$ defined by
	\begin{equation}
		\label{eq:support_qp}
		u_i = \qnorm{\vvec}^{-q/p} \abs{v_i}^{q/p} \sgn(v_i)
	\end{equation}
	is the unique support functional of $\vvec$.
	Similarly, for any nonzero $\uvec \in \ell^p$, the vector $\vvec \in \ell^q$ defined by
	\begin{equation}
		\label{eq:support_pq}
		v_i = \pnorm{\uvec}^{-p/q} \abs{u_i}^{p/q} \sgn(u_i)
	\end{equation}
	is the unique support functional of $\uvec$.
\end{lemma}
\begin{proof}
	Uniqueness of the support functional when $1<p<\infty$ follows from strict convexity of the
	norm (chapter 11 of~\cite{Carothers200412}).
	That the specific vectors~\eqref{eq:support_qp} and~\eqref{eq:support_pq} are support
	functionals is easily verified through direct computation~\cite{ArmstrongHolderDual}.
	%\begin{align}
	%	\vecprod{\uvec}{\vvec} &= \sum_i \qnorm{\vvec}^{-q/p} \abs{v_i}^{q/p} \sgn(v_i) v_i \\
	%	&= \qnorm{\vvec}^{-q/p} \sum_i \abs{v_i}^{q/p+1} \\
	%	&= \qnorm{\vvec}^{1-q} \sum_i \abs{v_i}^q \\
	%	&= \qnorm{\vvec}, \\
	%	\pnorm{\uvec} &= \qnorm{\vvec}^{-q/p} \left[
	%		\sum_i \abs{ \abs{v_i}^{q/p} \sgn(v_i) }^p \right]^{1/p} \\
	%	&= \qnorm{\vvec}^{-q/p} \left[
	%		\sum_i\abs{v_i}^q \right]^{1/p} \\
	%	&= \qnorm{\vvec}^{-q/p} \qnorm{\vvec}^{q/p} \\
	%	&= 1.
	%\end{align}
\end{proof}

We now describe generalized singular vectors.
Ordinary ($p=2$) left and right singular vectors $\uvec$ and $\vvec$ satisfy
$\twonorm{A \vvec} = \twonorm{A^\top \uvec} = \twonorm{A}$, furthermore $\uvec$ is
the support functional of $A\vvec$ (since $p=2$ this just means that $\uvec
\propto A\vvec$), and $\vvec$ is the support functional of $A^\top \uvec$.
These properties generalize to arbitrary $\ell^p$-norms, as we now show.

\begin{theorem}
	\label{thm:uAv}
	Let $p,q \in [1,\infty]$ and $1/p + 1/q = 1$.
	Let $A$ be a matrix.
	Then there are vectors $\uvec \in \ell^p$ and $\vvec \in \ell^q$ such that
	\begin{enumerate}[(a)]
		\item $\pnorm{\uvec} = \qnorm{\vvec} = 1$
		\item $\vecprod{\uvec}{A \vvec} = \pnorm{A^\top \uvec} = \qnorm{A\vvec} = \qnorm{A} =
			\pnorm{A^\top}$
		\item $\uvec$ is a support functional of $A\vvec$
		\item $\vvec$ is a support functional of $A^\top \uvec$.
		\item If $A$ is nonnegative then $\uvec$ and $\vvec$ are nonnegative.
	\end{enumerate}
\end{theorem}
\begin{proof}
	Let $\vvec$ be a vector satisfying $\qnorm{\vvec}=1$ and $\qnorm{A\vvec} = \qnorm{A}$.
	Such a vector is guaranteed to exist (see definition~5.6.1 of~\cite{HornJohnson199002}).
	Let $\uvec$ be a support functional of $A\vvec$.
	By the definition of a support functional, $\pnorm{\uvec}=1$ so claims (a) and (c) have
	been proved.
	With these two vectors defined, we have
	\begin{align}
		\qnorm{A} &= \qnorm{A\vvec}
		\\ &= \vecprod{\uvec}{A\vvec}
			&& \mbox{($\uvec$ is the support functional of $A\vvec$)}
		\\ &= \vecprod{\vvec}{(A^\top \uvec)}
		\\ &\le \qnorm{\vvec} \pnorm{A^\top \uvec}
			&& \mbox{(H{\"o}lder's inequality)}
		\\ &= \pnorm{A^\top \uvec}
		\\ &\le \pnorm{A^\top} \pnorm{\uvec}
		\\ &= \pnorm{A^\top}.
	\end{align}
	By symmetry we also have $\pnorm{A^\top} \le \qnorm{A}$, therefore the inequalities become
	equalities.  Claim (b) is proved.
	Since $\qnorm{\vvec}=1$ and $\vecprod{\vvec}{(A^\top \uvec)} = \pnorm{A^\top \uvec}$,
	claim (d) is proved as well.

	To prove claim (e), assume that $A$ is nonnegative.
	Then $\pnorm{\matabs{\uvec}}=\qnorm{\matabs{\vvec}}=1$ and
	$\qnorm{A\matabs{\vvec}} \ge \vecprod{\matabs{\uvec}}{A\matabs{\vvec}} \ge
	\vecprod{\uvec}{A\vvec} = \qnorm{A}$.
	It follows that $\qnorm{A\matabs{\vvec}}=\qnorm{A}$, thus $\matabs{\uvec}$ is a
	support functional of $A\matabs{\vvec}$.
	Therefore $\matabs{\uvec}$ and $\matabs{\vvec}$ could have been taken instead of $\uvec$
	and $\vvec$ in the first steps of
	this proof, justifying the claim that $\uvec$ and $\vvec$ can be chosen to be nonnegative.
\end{proof}

The Perron-Frobenius theorem states that an irreducible nonnegative matrix has a first
eigenvector that has positive components.  A similar statement holds for the first singular
vector: if $\aA$ is a nonnegative matrix that is not block diagonal then the left and right
singular vectors associated with the largest singular value of $\aA$ have positive entries.
This is true also for our generalized singular vectors, as we now show.

\begin{definition}
	A matrix $A$ is \emph{block diagonal} if there are permutation matrices $\sigma$ and $\tau$
	such that $A$ can be decomposed as $\matabs{A} = \sigma^\top ( A^{(1)} \oplus \dotsb
	\oplus A^{(L)} \oplus \zeromat^{M \times N} ) \tau$ where the $A^{(l)}$ are nonzero and have
	nonvanishing dimension, and at least one of the inequalities $L>1$, $M>0$, or $N>0$
	holds.\footnote{
		If $M>0,N=0$ then $\oplus \zeromat^{M \times N}$ adds $M$ rows of zeros.
		Similarly, if $M=0,N>0$ then $\oplus \zeromat^{M \times N}$ adds $N$ columns of zeros.
	}
	A matrix is \emph{not block diagonal} if no such decomposition is possible.  In
	particular, a matrix that is not block diagonal has no totally zero rows or columns.
\end{definition}

\begin{lemma}
	\label{thm:pos_vec}
	Let $q \in (1,\infty)$.
	Let $\aA$ be a nonnegative matrix that is not block diagonal.
	Let $\vvec$ be a nonzero, nonnegative vector that maximizes
	$\qnorm{\aA \vvec}/\qnorm{\vvec}$.
	Then $\vvec$ is in fact a positive vector (has no zero entries).
	%\footnote{
	%	If $q=2$ then the Perron-Frobenius theorem leads directly to the desired result.
	%	Let $\vvec$ be the eigenvector corresponding to the largest eigenvalue of $\aAT \aA$.
	%	Since $\aA$ is not block diagonal, $\aAT \aA$ is irreducible and the Perron-Frobenius
	%	theorem guarantees that $\vvec$ is positive.
	%}
\end{lemma}
\begin{proof}
	Let $Z = \{ i : v_i = 0 \}$.
	This will be a proof by contradiction;
	suppose that $\vvec$ has at least one zero entry, so that $Z$ is nonempty.
	Since $\vvec \ne 0$, the complement $Z^C$ is nonempty, therefore
	$Z$ and $Z^C$ partition the entries of $\vvec$ into two nonempty sets.
	Also, $Z$ and $Z^C$ can be considered as a partition of the columns of $\aA$.
	Since $\aA$ is not block diagonal, there must be indices $i \in Z$, $j \notin Z$, and $k$
	such that $\aA_{ki} > 0$ and $\aA_{kj} > 0$.
	We will show that $\vvec$ cannot maximize $\qnorm{\aA \vvec}/\qnorm{\vvec}$ by showing that
	$\vvec$ is not a critical point of $\qnorm{\aA \vvec}/\qnorm{\vvec}$, or equivalently of
	$\qnorm{\aA \vvec}^q/\qnorm{\vvec}^q$.
	Without loss of generality take $\qnorm{\vvec}=1$.
	Let $\ihat$ be the unit vector corresponding to $i$.
	We have
	\begin{align}
		\left.
		\frac{\partial}{\partial \alpha}
			\frac{\qnorm{\aA(\vvec+\alpha \ihat)}^q}{\qnorm{\vvec + \alpha \ihat}^q}
			\right|_{\alpha=0}
		&= \left. \frac{
				\left(
				\frac{\partial}{\partial \alpha}
				\qnorm{\aA(\vvec+\alpha \ihat)}^q
				\right)
				\qnorm{\vvec}^q -
				\qnorm{\aA\vvec}^q
				\left(
				\frac{\partial}{\partial \alpha}
				\qnorm{\vvec + \alpha \ihat}^q
				\right)
			}{\qnorm{\vvec}^{2q}}
			\right|_{\alpha=0} \\
		\label{eq:qpos_eqone}
		&= \left. \frac{\partial}{\partial \alpha} \qnorm{\aA(\vvec+\alpha \ihat)}^q
			\right|_{\alpha=0} \\
		&= \left. \frac{\partial}{\partial \alpha}
			\sum_l ([\aA \vvec]_l + \alpha \aA_{li})^q
			\right|_{\alpha=0} \\
		&= \sum_l q \aA_{li} [\aA \vvec]_l^{q-1} \\
		\label{eq:qpos_ineq_one}
		&\ge q \aA_{ki} [\aA \vvec]_k^{q-1} \\
		\label{eq:qpos_ineq_two}
		&\ge q \aA_{ki} (\aA_{kj} v_j)^{q-1} \\
		&> 0.
	\end{align}
	Equality~\eqref{eq:qpos_eqone} follows from $\qnorm{\vvec}=1$ as well as
	$(v_i=0 \implies \partial \qnorm{\vvec + \alpha \ihat}^q/\partial \alpha=0)$.
	Inequality~\eqref{eq:qpos_ineq_one} follows from each term of the previous summation being
	nonnegative.
	Inequality~\eqref{eq:qpos_ineq_two} follows from each term of the sum
	$[\aA \vvec]_k = \sum_n \aA_{kn} v_n$ being nonnegative.
\end{proof}

\begin{theorem}
	\label{thm:uAv_pos_noblock}
	Let $p,q \in (1,\infty)$ and $1/p + 1/q = 1$.
	Let $\aA$ be a nonnegative matrix that is not block diagonal.
	Then there are positive vectors $\uvec$ and $\vvec$ satisfying
	\begin{align}
		\label{eq:pos_Au_v_nb}
		(\matabs{A}^\top \uvec)_n &= v_n^{q/p} \qnorm{\matabs{A}}, \\
		\label{eq:pos_Av_u_nb}
		( \matabs{A}\vvec )_m &= u_m^{p/q} \qnorm{\matabs{A}}.
	\end{align}
	Note: if $p=q=2$ then $\uvec$ and $\vvec$ will be the left and right singular vectors
	associated with the largest singular value of $\aA$.
\end{theorem}
\begin{proof}
	\Cref{thm:uAv} guarantees the existence of nonnegative vectors $\uvec$ and $\vvec$ that
	satisfy $\pnorm{\uvec}=\qnorm{\vvec}=1$ and
	$\vecprod{\uvec}{\matabs{A} \vvec} = \qnorm{\matabs{A}} = \pnorm{\matabs{A}^\top}$
	with $\uvec$ being the support functional of $A\vvec$ and $\vvec$ being the support
	functional of $A^\top \uvec$.
	\Cref{thm:support_fcn} give the exact form of these support
	functionals:
	\begin{align}
		u_m &= \qnorm{\matabs{A} \vvec}^{-q/p} (\matabs{A} \vvec)_m^{q/p}
			\sgn(\matabs{A} \vvec) \\
		v_n &= \pnorm{\matabs{A}^\top \uvec}^{-p/q} (\matabs{A}^\top \uvec)_n^{p/q}
			\sgn(\matabs{A}^\top \uvec).
	\end{align}
	Since $\matabs{A}$, $\uvec$, and $\vvec$ are nonnegative, the $\sgn$ functions disappear.
	\Cref{thm:uAv} gives $\qnorm{\matabs{A} \vvec} = \pnorm{\matabs{A}^\top \uvec} =
	\qnorm{\matabs{A}}$.
	With these simplifications, we get~\eqref{eq:pos_Au_v_nb}-\eqref{eq:pos_Av_u_nb}.
	That $\uvec$ and $\vvec$ have nonzero entries follows from \Cref{thm:pos_vec}.
\end{proof}

We now generalize \cref{thm:uAv_pos_noblock} to
matrices that are not block diagonal.  This is done by applying
\cref{thm:uAv_pos_noblock} to each individual block of the matrix.
Each block of $\aA$ may have a different operator norm, but each of these is upper bounded by
$\qnorm{\aA}$.
For this reason, we end up with an inequality rather than an equality when generalizing
\eqref{eq:pos_Au_v_nb}-\eqref{eq:pos_Av_u_nb}.

\begin{theorem}
	\label{thm:uAv_pos}
	Let $p,q \in (1,\infty)$ and $1/p + 1/q = 1$.
	Let $\aA$ be a nonnegative matrix that can possibly be block diagonal and that may have
	some totally zero rows or columns.
	Then there are positive vectors $\uvec$ and $\vvec$ satisfying
	\begin{align}
		\label{eq:pos_Au_v}
		(\matabs{A}^\top \uvec)_n &\le v_n^{q/p} \qnorm{\matabs{A}}, \\
		\label{eq:pos_Av_u}
		( \matabs{A}\vvec )_m &\le u_m^{p/q} \qnorm{\matabs{A}}.
	\end{align}
\end{theorem}
\begin{proof}
	Let $\sigma$ and $\tau$ be permutations matrices that bring out the block structure of
	$\matabs{A}$, and let $A^{(1)}, \dotsc, A^{(L)}$ be the blocks.
	Specifically, suppose
	$\sigma^\top ( A^{(1)} \oplus \dotsb \oplus A^{(L)} \oplus \zeromat^{M \times N} ) \tau =
	\matabs{A}$
	where the $A^{(1)} \dotsm A^{(L)}$ matrices are not block diagonal and
	$\zeromat^{M \times N}$ is an $M$-by-$N$ matrix of zeros (if there is no zero block then
	just take $M=N=0$).
	It is easy to see that $\qnorm{A^{(l)}} \le \qnorm{\aA}$ for all $l \in \{1, \dotsc, L\}$.

	By \cref{thm:uAv_pos_noblock},
	there are positive vectors $\uvec^{(1)}, \dotsc, \uvec^{(L)}$ and
	$\vvec^{(1)}, \dotsc, \vvec^{(L)}$ such that
	\begin{align}
		\label{eq:Au_block}
		(A^{(l) \top} \uvec^{(l)})_n &= v_n^{(l) q/p} \qnorm{A^{(l)}} \\
		&\le v_n^{(l) q/p} \qnorm{\matabs{A}}, \\
		(A^{(l)} \vvec^{(l)} )_m &= u_m^{(l) p/q} \qnorm{A^{(l)}} \\
		\label{eq:Av_block}
		&\le u_m^{(l) p/q} \qnorm{\matabs{A}}
	\end{align}
	for all $l \in \{1, \dotsc, L\}$.
	Define $\uvec = \sigma^\top (\uvec^{(1)} \oplus \dotsb \oplus \uvec^{(L)} \oplus \onesvec^M)$
	and $\vvec = \tau^\top (\vvec^{(1)} \oplus \dotsb \oplus \vvec^{(L)} \oplus \onesvec^N)$
	where $\onesvec^M$ and $\onesvec^N$ are the all-ones vectors of lengths $M$ and $N$,
	respectively.
	Then~\eqref{eq:Au_block}-\eqref{eq:Av_block} imply~\eqref{eq:pos_Au_v}-\eqref{eq:pos_Av_u}.
	Since the $\uvec^{(1)}, \dotsc, \uvec^{(L)}$ and $\vvec^{(1)}, \dotsc, \vvec^{(L)}$
	are positive, $\uvec$ and $\vvec$ are positive.
\end{proof}

We are now ready to complete the proof of \cref{thm:best_eps_qnorm}.

\begin{proof}[Proof of
	\cref{thm:best_eps_qnorm}\eqref{part:best_bval_eq_qnorm}-\eqref{part:best_query}
	]
	Let $A$ be a matrix.
	Set $K = \{0\}$ and $\alpha_{mn0} = A_{mn}$.  Clearly
	condition~\eqref{cond:eps_alphasum} of \cref{def:eps} is satisfied.

	Consider the case $p \in (1,\infty)$.
	Let $\uvec$ and $\vvec$ be positive vectors
	satisfying~\eqref{eq:pos_Au_v}-\eqref{eq:pos_Av_u}.  The
	existence of such vectors is guaranteed by \cref{thm:uAv_pos}.
	Define the probability distributions
	\begin{align}
		\Pdist{n|m} &= \abs{A_{mn}} v_n / [\aA  \vvec]_m, \\
		\Qdist{m|n} &= \abs{A_{mn}} u_m / [\aAT \uvec]_n.
	\end{align}
	These satisfy condition~\eqref{cond:eps_cost} of \cref{def:eps} with
	$\bval = \qnorm{\aA}$ since
	\begin{align}
		\max_{mnk} \left\{ \frac{\abs{\alpha_{mnk}}}{
			\Pdist{n|m}^{1/p}
			\Qdist{m|n}^{1/q}
		} \right\}
		&=
		\max_{mn} \left\{ \frac{\abs{A_{mn}}}{
			\Pdist{n|m}^{1/p}
			\Qdist{m|n}^{1/q}
		} \right\}
		\\ &=
		\max_{mn} \left\{
			\left( \frac{ [\aA \vvec]_m }{ v_n } \right)^{1/p}
			\left( \frac{ [\aA^\top \uvec]_n }{ u_m } \right)^{1/q}
		\right\}
		\\ &\le
		\max_{mn} \left\{
			\left( \frac{ u_m^{p/q} \qnorm{\aA} }{ v_n } \right)^{1/p}
			\left( \frac{ v_n^{q/p} \qnorm{\aA} }{ u_m } \right)^{1/q}
		\right\}
		\\ &= \qnorm{\aA}.
	\end{align}

	Now consider the case $p=1$, $q=\infty$ (the case $p=\infty$, $q=1$ follows by a
	symmetrical argument).
	Define $\Pdist{n|m} = \abs{A_{mn}} / \sum_{n'} \abs{A_{mn'}}$ and define
	$\Qdist{m|n}$ arbitrarily.  Condition~\eqref{cond:eps_cost} of \cref{def:eps} is satisfied
	with $\bval = \infnorm{\aA}$ since
	\begin{align}
		\max_{mnk} \left\{ \frac{\abs{\alpha_{mnk}}}{
			\Pdist{n|m}^{1/p}
			\Qdist{m|n}^{1/q}
		} \right\}
		&=
		\max_{mn} \left\{ \frac{\abs{A_{mn}}}{
			\Pdist{n|m}^1
			\Qdist{m|n}^0
		} \right\}
		\\ &=
		\max_{mn} \left\{ \frac{\abs{A_{mn}}}{
			\abs{A_{mn}} / \sum_{n'} \abs{A_{mn'}}
		} \right\}
		\\ &\le \infnorm{\aA}.
	\end{align}

	If one is concerned with query complexity rather than time complexity, and if $A$ is not
	defined in terms of an oracle, then
	conditions~\eqref{cond:eps_samplr}-\eqref{cond:eps_samprl} of \cref{def:eps}
	are satisfied trivially with $\fval=0$ since no oracle queries are needed in order to carry out
	the required operations.  So $A$ is $\epsp{\qnorm{\aA}}{0}$.
\end{proof}

\section{Proofs for section~\ref{sec:eps_eht}}
\label[secinapp]{sec:eps_eht_proofs}

In this section we prove \cref{thm:eps_sumprodexp,thm:eht_eps_trace}.
The proofs are conceptually rather simple, however they are notationally tedious.
Since we will at times be manipulating infinite series, we begin by showing that these series
converge absolutely.  This will be useful, since absolutely convergent series allow permutation
of terms and reordering of double summations.

\begin{lemma}
	Let $\bval$ and $\alpha_{mnk}$ satisfy
	condition~\eqref{cond:eps_cost} of \cref{def:eps}.
	Then series $\sum_{k \in K} \alpha_{mnk}$ is absolutely convergent for all $m,n$,
	and $\sum_{k \in K} \abs{\alpha_{mnk}} \le \bval$.
	\label{thm:abs_conv}
\end{lemma}
\begin{proof}
	Rearranging~\eqref{eq:eps_cost} of condition~\eqref{cond:eps_cost}
	gives, for all $m,n,k$,
	\begin{align}
		\abs{\alpha_{mnk}} &\le \bval \cdot \Pdist{n,k|m}^{1/p} \Qdist{m,k|n}^{1/q}
		\\ &\le \bval \cdot \left[ \Pdist{n,k|m}/p + \Qdist{m,k|n}/q \right].
	\end{align}
	Therefore,
	\begin{align}
		\sum_{k \in K} \abs{\alpha_{mnk}}
		&\le \bval \sum_{k \in K} \left[ \Pdist{n,k|m}/p + \Qdist{m,k|n}/q \right]
		\\ &=   \bval \cdot \left[ \Pdist{n|m}/p + \Qdist{m|n}/q \right]
		\\ &\le \bval
		\\ &< \infty.
	\end{align}
\end{proof}

We now prove that linear combinations of EPS operators are EPS.
\Cref{thm:eps_sumprodexp}\eqref{part:eps_sum_simple}, regarding sums of EPS operators,
follows as a corollary.
This will also be used to prove \cref{thm:eps_sumprodexp}\eqref{part:eps_exp},
regarding exponentials of EPS operators.

\begin{theorem}[Linear combination of EPS]
	\label{thm:eps_sum_fancy}
	Let $L$ be a finite or countable set.
	For $l \in L$ let $s_l$ be a complex number
	and let $A^{(l)}$ be an $M \times N$ matrix that is $\epsp{\bval_l}{\fval_l}$
	for some $\fval_l$ and $\bval_l$.
	Let $\Wdist{l}$ be a probability distribution\footnote{
		The lowest $\bval$ is obtained when $\Wdist{l}$ is proportional to $\abs{s_l} \bval_l$.
	} on $l$ such that $\Wdist{l}$ can be sampled from, and $s_l/\Wdist{l}$ computed,
	in average time $\bigomic(\fval_0)$.
	Let
	$\bval := \max_l\{ \abs{s_l} \bval_l / \Wdist{l} \} < \infty$
	and
	$\fval := \fval_0 + \sum_l \Wdist{l} \fval_l$.
	Then $\sum_l s_l A^{(l)}$ is $\epsp{\bval}{\fval}$.
\end{theorem}
\begin{proof}
	For each $l \in L$, $A^{(l)}$ is $\epsp{\bval_l}{\fval_l}$ so there are
	$K_l$, $\alpha^{(l)}_{mnk}$,
	$\Pdistsub{l}{n,k|m}$, and $\Qdistsub{l}{m,k|n}$ satisfying \cref{def:eps}.
	Let $K = L \times \cup_{l \in L} K_l$.  For $(l,k) \in K$ define
	\begin{equation}
		\label{eq:opsum_alpha}
		\alpha_{mn(l,k)} = \left\{ \begin{array}{ll}
				s_l \alpha^{(l)}_{mnk} & \mbox{ if $k \in K_l$} \\
				0 & \mbox{ otherwise.}
			\end{array} \right.
	\end{equation}
	We first show that $\sum_{(l,k) \in K} \alpha_{mn(l,k)}$ is absolutely convergent, so that
	it can be expressed as a double sum.
	By \cref{thm:abs_conv}, $\sum_{k \in K_l} \abs{\alpha^{(l)}_{mnk}} \le \bval_l$
	for all $l \in L$, therefore
	\begin{align}
		\sum_{(l,k) \in K} \abs{\alpha_{mn(l,k)}}
		&= \sum_{l \in L} \abs{s_l} \sum_{k \in K_l} \abs{\alpha^{(l)}_{mnk}}
		\\ &\le \sum_{l \in L} \abs{s_l \bval_l}
		\\ &\le \bval.
	\end{align}
	Since $\bval < \infty$ by assumption, the series $\sum_{(l,k) \in K} \alpha_{mn(l,k)}$
	is absolutely convergent.  We can then decompose it as a double series,
	\begin{align}
		\sum_{(l,k) \in K} \alpha_{mn(l,k)}
		&= \sum_{l \in L} s_l \sum_{k \in K_l} \alpha^{(l)}_{mnk}
		\\ &= \sum_{l \in L} s_l A^{(l)},
	\end{align}
	showing that condition~\eqref{cond:eps_alphasum} of \cref{def:eps} is satisfied.

	Define the probability distributions
	\begin{align}
		\label{eq:opsum_P}
		\Pdist{n,(l,k)|m} &= \left\{ \begin{array}{ll}
				\Wdist{l} \Pdistsub{l}{n,k|m} & \mbox{ if $k \in K_l$} \\
				0 & \mbox{ otherwise}
			\end{array} \right.
		\\
		\label{eq:opsum_Q}
		\Qdist{m,(l,k)|n} &= \left\{ \begin{array}{ll}
				\Wdist{l} \Qdistsub{l}{m,k|n} & \mbox{ if $k \in K_l$} \\
				0 & \mbox{ otherwise.}
			\end{array} \right.
	\end{align}
	We now show that condition~\eqref{cond:eps_cost} holds.
	Let $m \in \{1, \dotsc, M\}$, $n \in \{1, \dotsc, N\}$, and $(l,k) \in K$.
	We need only consider $k \in K_l$ since otherwise $\alpha_{mn(l,k)}$ vanishes.
	\begin{align}
		\frac{\abs{\alpha_{mn(l,k)}}}{
			\left[ \Pdist{n,(l,k)|m} \right]^{1/p}
			\left[ \Qdist{m,(l,k)|n} \right]^{1/q}
		}
		&=
		\frac{\abs{ s_l \alpha^{(l)}_{mnk} }}{
			\left[ \Wdist{l} \Pdistsub{l}{n,k|m} \right]^{1/p}
			\left[ \Wdist{l} \Qdistsub{l}{m,k|n} \right]^{1/q}
		}
		\\ &=
		\frac{\abs{s_l}}{\Wdist{l}} \cdot
		\frac{\abs{ \alpha^{(l)}_{mnk} }}{
			\left[ \Pdistsub{l}{n,k|m} \right]^{1/p}
			\left[ \Qdistsub{l}{m,k|n} \right]^{1/q}
		}
		\\ &\le
		\abs{s_l} \bval_l / \Wdist{l}
		\\ &\le \bval.
	\end{align}
	Condition~\eqref{cond:eps_samplr} requires that
	the distribution $\Pdist{n,(l,k)|m}$ can be sampled from,
	and $\alpha_{mn(l,k)} / \Pdist{n,(l,k)|m}$ and
	$\alpha_{mn(l,k)} / \Qdist{m,(l,k)|n}$ can be computed,
	in average time
	$\bigomic(\fval) = \bigomic(\fval_0 + \sum_l \Wdist{l} \fval_l)$.
	This can be accomplished as follows:
	\begin{enumerate}[(i)]
		\item
			Draw $l$ according to the distribution $\Wdist{l}$ and compute $s_l/\Wdist{l}$.
			This can be done in average time $\bigomic(\fval_0)$.
		\item
			Draw $n,k$ according to the distribution $\Pdistsub{l}{n,k|m}$ and compute
			$\alpha^{(l)}_{mnk} / \Pdistsub{l}{n,k|m}$ and
			$\alpha^{(l)}_{mnk} / \Qdistsub{l}{m,k|n}$.
			This can be done in average time $\bigomic(\fval_l)$.
		\item
			The quantities $\alpha_{mn(l,k)} / \Pdist{n,(l,k)|m}$ and
			$\alpha_{mn(l,k)} / \Qdist{m,(l,k)|n}$ can be directly computed
			from~\eqref{eq:opsum_alpha}, \eqref{eq:opsum_P}, and~\eqref{eq:opsum_Q}
			in time $\bigomic(1)$ given the quantities that have been computed in the previous
			two steps.
	\end{enumerate}
	The average time needed for a given $l$ is $\bigomic(\fval_0 + \fval_l)$, therefore the
	average time needed given that $l$ is drawn according to $\Wdist{l}$ is
	$\bigomic(\fval) = \bigomic(\fval_0 + \sum_l \Wdist{l} \fval_l)$.
	Condition~\eqref{cond:eps_samplr} is satisfied.
	Condition~\eqref{cond:eps_samprl} follows from a symmetric argument.
\end{proof}

\begin{proof}[Proof of \cref{thm:eps_sumprodexp}\eqref{part:eps_sum_simple}]
	This follows directly from \cref{thm:eps_sum_fancy}.
	Specifically, apply \cref{thm:eps_sum_fancy} with $L=\{A,B\}$, $s_A=s_B=1$,
	$\Wdist{A} = \bval_A/(\bval_A+\bval_B)$, and $\Wdist{B} = \bval_B/(\bval_A+\bval_B)$.
	Then
	$\bval = \max_l\{ \abs{s_l} \bval_l / \Wdist{l} \} = \bval_A + \bval_B$
	and
	$\fval = \bigomic(1) + \sum_l \Wdist{l} \fval_l = \bigomic(\max\{\bval_A, \bval_B\})$.
\end{proof}

\begin{proof}[Proof of \cref{thm:eps_sumprodexp}\eqref{part:eps_prod}]
	Since $A$ is $\epsp{\bval_A}{\fval_A}$, there are $K_A$, $\alpha^{(A)}_{lmk}$,
	$\Pdistsub{A}{m,k|l}$, and $\Qdistsub{A}{l,k|m}$ satisfying \cref{def:eps}
	with $l \in \{1, \dotsc, L\}$, $m \in \{1, \dotsc, M\}$, and $k \in K_A$.
	Likewise, since $B$ is $\epsp{\bval_B}{\fval_B}$, there are $K_B$, $\alpha^{(B)}_{mnk}$,
	$\Pdistsub{B}{n,k|m}$, and $\Qdistsub{B}{m,k|n}$ satisfying \cref{def:eps}
	with $m \in \{1, \dotsc, M\}$, $n \in \{1, \dotsc, N\}$, and $k \in K_B$.

	Let $K = K_A \times K_B \times \{1, \dotsc, M\}$ and
	\begin{equation}
		\alpha_{ln(k',k'',m)} = \alpha^{(A)}_{lmk'} \alpha^{(B)}_{mnk''}.
	\end{equation}
	We first show that $\sum_{(k',k'',m) \in K} \alpha_{ln(k',k'',m)}$ is absolutely
	convergent, so that it can be expressed as a double series.
	By \cref{thm:abs_conv}, $\sum_{k' \in K_A}  \abs{\alpha^{(A)}_{lmk'}} \le \bval_A$
	and $\sum_{k'' \in K_B} \abs{\alpha^{(B)}_{mnk''}} \le \bval_B$, therefore
	\begin{align}
		\sum_{(k',k'',m) \in K} \abs{\alpha_{ln(k',k'',m)}}
		&=
			\sum_{m \in \{1, \dotsc, M\}}
			\sum_{k' \in K_A}  \abs{\alpha^{(A)}_{lmk'}}
			\sum_{k'' \in K_B} \abs{\alpha^{(B)}_{mnk''}}
		\\ &\le M \bval_A \bval_B
		\\ &\le \infty.
	\end{align}
	Being absolutely convergent, $\sum_{(k',k'',m) \in K} \alpha_{ln(k',k'',m)}$ can be
	expressed as a double series, giving
	\begin{align}
		\sum_{(k',k'',m) \in K} \alpha_{ln(k',k'',m)}
		&=
			\sum_{m \in \{1, \dotsc, M\}}
			\sum_{k' \in K_A} \alpha^{(A)}_{lmk'}
			\sum_{k'' \in K_B} \alpha^{(B)}_{mnk''} \\
		&= \sum_m A_{lm} B_{mn} \\
		&= (AB)_{ln}
	\end{align}
	so condition~\eqref{cond:eps_alphasum} of \cref{def:eps} is satisfied.

	Define the probability distributions
	\begin{align}
		\label{eq:prod_dist_P}
		\Pdist{n,(k',k'',m)|l} &= \Pdistsub{A}{m,k'|l} \Pdistsub{B}{n,k''|m}, \\
		\label{eq:prod_dist_Q}
		\Qdist{l,(k',k'',m)|n} &= \Qdistsub{A}{l,k'|m} \Qdistsub{B}{m,k''|n}.
	\end{align}
	These satisfy condition~\eqref{cond:eps_cost} of \cref{def:eps} since
	for all $l,m,n,k',k''$,
	\begin{align}
		\bval_A \bval_B
		&\ge
		\frac{\abs{\alpha^{(A)}_{lmk'}}}{
			\Pdistsub{A}{m,k'|l}^{1/p}
			\Qdistsub{A}{l,k'|m}^{1/q}
		}
		\frac{\abs{\alpha^{(B)}_{mnk''}}}{
			\Pdistsub{B}{n,k''|m}^{1/p}
			\Qdistsub{B}{m,k''|n}^{1/q}
		} \\
		&=
		\frac{\abs{\alpha_{ln(k',k'',m)}}}{
			\Pdist{n,(k',k'',m)|l}^{1/p}
			\Qdist{l,(k',k'',m)|n}^{1/q}
		}.
	\end{align}

	Condition~\eqref{cond:eps_samplr} requires that it be possible in average time
	$\bigomic(\fval_A+\fval_B)$ to sample from the
	probability distribution $\Pdist{n,(k',k'',m)|l}$
	and to compute
	$\frac{ \alpha_{ln(k',k'',m)} }{ \Pdist{n,(k',k'',m)|l} }$ and
	$\frac{ \alpha_{ln(k',k'',m)} }{ \Qdist{l,(k',k'',m)|n} }$.
	This can be accomplished as follows:
	\begin{enumerate}[(i)]
		\item
			Draw $m,k'$ from $\Pdistsub{A}{m,k'|l}$ and compute
			$\frac{ \alpha^{(A)}_{lmk'} }{ \Pdistsub{A}{m,k'|l} }$ and
			$\frac{ \alpha^{(A)}_{lmk'} }{ \Qdistsub{A}{l,k'|m} }$.
			This can be done in average time $\bigomic(\fval_A)$.
		\item
			Draw $n,k''$ from $\Pdistsub{B}{n,k''|m}$ and compute
			$\frac{ \alpha^{(B)}_{mnk''} }{ \Pdistsub{B}{n,k''|m} }$ and
			$\frac{ \alpha^{(B)}_{mnk''} }{ \Qdistsub{B}{m,k''|n} }$.
			This can be done in average time $\bigomic(\fval_B)$.
		\item Compute
			\begin{align}
				\frac{ \alpha_{ln(k',k'',m)} }{ \Pdist{n,(k',k'',m)|l} }
				&=
				\frac{ \alpha^{(A)}_{lmk' } }{ \Pdistsub{A} {m,k' |l} }
				\cdot
				\frac{ \alpha^{(B)}_{mnk''} }{ \Pdistsub{B}{n,k''|m} }
				\\
				\frac{ \alpha_{ln(k',k'',m)} }{ \Qdist{l,(k',k'',m)|n} }
				&=
				\frac{ \alpha^{(A)}_{lmk' } }{ \Qdistsub{A} {l,k' |m} }
				\cdot
				\frac{ \alpha^{(B)}_{mnk''} }{ \Qdistsub{B}{m,k''|n} }.
			\end{align}
			This can be done in time $\bigomic(1)$ since the factors on the right
			hand sides of these expressions have already been computed in the previous two
			steps.
	\end{enumerate}
	So condition~\eqref{cond:eps_samplr} is satisfied.
	Condition~\eqref{cond:eps_samprl} follows from a symmetric argument.
\end{proof}

\begin{proof}[Proof of \cref{thm:eps_sumprodexp}\eqref{part:eps_exp}]
	Let $A$ be a square matrix that is $\epsp{\bval}{\fval}$.  We will show that $e^A$ is
	$\epsp{e^\bval}{\bval \fval}$.

	This follows from applying \cref{thm:eps_sum_fancy} and
	\cref{thm:eps_sumprodexp}\eqref{part:eps_prod} to
	$e^A = \sum_{j=0}^\infty A^j/j! $.
	Specifically, let $L = \{0,1,\dotsc\}$, $A^{(l)}=A^l$, $s_l=1/l! $, and
	$\Wdist{l}=\bval^l/(l! e^{\bval})$.
	By repeated application of \cref{thm:eps_sumprodexp}\eqref{part:eps_prod},
	$A^{(l)}$ is $\epsp{\bval^l}{lf}$.
	Assume for now that $\Wdist{l}$ can be sampled in average time $\bigomic(\bval)$.
	Then by \cref{thm:eps_sum_fancy}, $e^A = \sum_{j=0}^\infty A^j/j! $
	is $\epsp{\bval'}{\fval'}$ with
	$\bval' = \max_l\{ \abs{s_l} \bval_l / \Wdist{l} \} = e^\bval$
	and
	\begin{align}
		\fval' &= \bval + \sum_{l=0}^\infty \Wdist{l} \fval_l
		\\ &= \bval + \sum_{l=0}^\infty \frac{lf \bval^l}{l! e^{\bval}}
		\\ &= \bval + \frac{\bval \fval}{e^b} \sum_{l=1}^\infty \frac{\bval^{l-1}}{(l-1)!}
		\\ &= \bval + \bval \fval
		\\ &= \bigomic(\bval \fval)
	\end{align}

	It remains only to show that $\Wdist{l}$ can be sampled in time $\bigomic(\bval)$.
	The procedure is as follows.  Flip a weighted coin that lands heads with probability
	$W(0)$, and if it lands heads take $l=0$.  This can be done in time $\bigomic(1)$.
	If the coin landed tails then flip another coin that lands heads with probability
	$W(1)/(1-W(0))$, and if it lands heads take $l=1$.
	Continue, each iteration flipping a coin that lands heads with probability
	$\Wdist{l}/(1-\sum_{j=0}^{l-1} W(j))$.
	Each iteration requires computing $\Wdist{l}/(1-\sum_{j=0}^{l-1} W(j))$, which in turn
	requires computing $\Wdist{l}$ and updating the partial sum with the previous $W(l-1)$.
	This can be done in $\bigomic(1)$ time.
	The expected number of iterations is $\sum_l l \Wdist{l}=\bval$.
	Therefore, this sampling algorithm takes average time $\bval$.
\end{proof}

\begin{proof}[Proof of \cref{thm:eht_eps_trace}]
	Since $\sigma$ is $\ehtp{\bval_\sigma}{\fval_\sigma}$, there are $\alpha^{(\sigma)}_{nmk}$,
	$\Pdistsub{\sigma}{m,k}$, and $\Qdistsub{\sigma}{n,k}$ with $k \in K_\sigma$ satisfying
	\cref{def:eht} (note that $m$ and $n$ have been swapped since $\sigma$ is an $N \times M$
	operator).
	Similarly, since $A$ is $\epsp{\bval_A}{\fval_A}$ there are $\alpha^{(A)}_{mnk'}$,
	$\Pdistsub{A}{n,k'|m}$, and $\Qdistsub{A}{m,k'|n}$ with $k' \in K_A$ satisfying
	\cref{def:eps}.

	We have
	\begin{align}
		\Tr(A \sigma) &= \sum_{mn} A_{mn} \sigma_{nm}
		\\ &= \sum_{mnkk'} \alpha^{(A)}_{mnk'} \alpha^{(\sigma)}_{nmk}.
		\label{eq:trace_sigma_A_sum}
	\end{align}

	Define the probability distribution
	\begin{equation}
		\Rdist{m,n,k,k'} = \frac{1}{p} \Pdistsub{\sigma}{m,k} \Pdistsub{A}{n,k'|m} +
			\frac{1}{q} \Qdistsub{\sigma}{n,k} \Qdistsub{A}{m,k'|n}.
	\end{equation}
	By the inequality of arithmetic and geometric means,
	\begin{equation}
		\Rdist{m,n,k,k'} \ge [\Pdistsub{\sigma}{m,k} \Pdistsub{A}{n,k'|m}]^{1/p}
			[\Qdistsub{\sigma}{n,k} \Qdistsub{A}{m,k'|n}]^{1/q}.
	\end{equation}
	Setting $V(m,n,k,k') = \alpha^{(A)}_{mnk'} \alpha^{(\sigma)}_{nmk}$ we get the bound
	\begin{align}
		\bmax :&=
			\max_{mnkk'} \left\{
			\frac{\abs{V(m,n,k,k')}}{\Rdist{m,n,k,k'}}
			\right\}
		\\ &\le
			\max_{mnkk'} \left\{
			\frac{\abs{\alpha^{(A)}_{mnk'} \alpha^{(\sigma)}_{nmk}}}
			{[\Pdistsub{\sigma}{m,k} \Pdistsub{A}{n,k'|m}]^{1/p}
			[\Qdistsub{\sigma}{n,k} \Qdistsub{A}{m,k'|n}]^{1/q}}
			\right\}
		\\ &\le
			\max_{mnk'} \left\{
			\frac{\abs{\alpha^{(A)}_{mnk'}}}
			{\Pdistsub{A}{n,k'|m}^{1/p} \Qdistsub{A}{m,k'|n}]^{1/q}}
			\right\}
		\cdot
			\max_{mnk} \left\{
			\frac{\abs{\alpha^{(\sigma)}_{nmk}}}
			{\Pdistsub{\sigma}{m,k}^{1/p} \Qdistsub{\sigma}{n,k}^{1/q}}
			\right\}
		\\ &\le \bval_A \bval_\sigma.
	\end{align}
	By \cref{thm:chernoff3}, the sum~\eqref{eq:trace_sigma_A_sum} can be estimated at the cost of
	drawing $\bigomic(\log(\delta^{-1}) \epsilon^{-2} \bval_\sigma^2 \bval_A^2)$
	samples from $\Rdist{m,n,k,k'}$ and
	evaluating the corresponding $V(m,n,k,k')/\Rdist{m,n,k,k'}$.
	Each of these samples can be computed in average time $\bigomic(\fval_\sigma+\fval_A)$
	as follows.
	\begin{enumerate}[(i)]
		\item Flip a weighted coin that lands heads with probability $1/p$.
		\item If it lands heads, sample $m,k$ according to $\Pdistsub{\sigma}{m,k}$
			and then sample $n,k'$ according to $\Pdistsub{A}{n,k'|m}$.
		\item If it lands tails, sample $n,k$ according to $\Qdistsub{\sigma}{n,k}$
			and then sample $m,k'$ according to $\Qdistsub{A}{m,k'|n}$.
		\item The previous steps produce a sample according to $\Rdist{m,n,k,k'}$ and can be
			accomplished in time $\bigomic(\fval_\sigma+\fval_A)$ by conditions
			\eqref{cond:eps_samplr} and~\eqref{cond:eps_samprl} of
			\cref{def:eps} and
			\eqref{cond:eht_samplr} and~\eqref{cond:eht_samprl} of
			\cref{def:eht}, with the side effect of producing values
			$\alpha^{(\sigma)}_{nmk} / \Pdistsub{\sigma}{m,k}$,
			$\alpha^{(A)}_{mnk'}     / \Pdistsub{A}{n,k'|m}  $,
			$\alpha^{(\sigma)}_{nmk} / \Qdistsub{\sigma}{n,k}$, and
			$\alpha^{(A)}_{mnk'}     / \Qdistsub{A}{m,k'|n}  $.
		\item These values can be used to compute $V(m,n,k,k')/\Rdist{m,n,k,k'}$ since
		\begin{align}
			\frac{V(m,n,k,k')}{\Rdist{m,n,k,k'}}
			&=
			\frac{\alpha^{(A)}_{mnk'} \alpha^{(\sigma)}_{nmk}}{\Rdist{m,n,k,k'}}
			\\ &= \left[
				\frac{1}{p}
				\frac{\Pdistsub{A}{n,k'|m}}{\alpha^{(A)}_{mnk'}}
				\cdot
				\frac{\Pdistsub{\sigma}{m,k}}{\alpha^{(\sigma)}_{nmk}}
				+
				\frac{1}{q}
				\frac{\Qdistsub{A}{m,k'|n}}{\alpha^{(A)}_{mnk'}}
				\cdot
				\frac{\Qdistsub{\sigma}{n,k}}{\alpha^{(\sigma)}_{nmk}}
			\right]^{-1}
		\end{align}
	\end{enumerate}
	Therefore, the sum~\eqref{eq:trace_sigma_A_sum} can be estimated in average time
	$\bigomic[\log(\delta^{-1}) \epsilon^{-2} \bval_\sigma^2 \bval_A^2 (\fval_\sigma+\fval_A)]$.
\end{proof}

\section{Proofs for section~\ref{sec:circuits}}
\label[secinapp]{sec:circuits_proofs}

In section~\ref{sec:circuits} several matrices and classes of matrices were claimed to be
$\epstwo{\bval}{\fval}$ or $\epsp{\bval}{\fval}$ for small values of $\bval$ and $\fval$.
In this appendix we provide proofs for these claims.

We first prove that the efficiently computable sparse (ECS) matrices from
\cite{vandennest2011} (definition reproduced below) are $\epsp{\polylog(N)}{\polylog(N)}$.
This covers a rather large class of matrices including permutation matrices, Pauli matrices,
controlled phase matrices, and arbitrary unitaries on a constant number of qudits.
The original definition from \cite{vandennest2011} was in terms of qubits, but we adapt it to
systems of arbitrary dimension.
\begin{definition}[ECS]
	A matrix $A$ is efficiently computable sparse (ECS) if
	\begin{enumerate}[(a)]
		\item Each row and column of $A$ has at most $\polylog(N)$ nonzero entries.
		\item For any given row index $m$, it is possible in $\polylog(N)$ time to list the
			indices of the nonzero entries in that row, $\{n : A_{mn} \ne 0\}$, and to compute
			their values $A_{mn}$.
		\item For any given column index $n$, it is possible in $\polylog(N)$ time to list the
			indices of the nonzero entries in that column, $\{m : A_{mn} \ne 0\}$, and to
			compute their values $A_{mn}$.
	\end{enumerate}
\end{definition}

\begin{theorem}[ECS is EPS]
	\label{thm:ecs_is_eps}
	Let $A$ be an ECS matrix satisfying $\max_{mn}
	\{\abs{A_{mn}}\} = \polylog(N)$.
	Unitaries and Hermitian matrices whose eigenvalues are in the $[-1,1]$ range satisfy this
	bound.
	Then $A$ is $\epsp{\polylog(N)}{\polylog(N)}$ for any $p \in [1, \infty]$.
\end{theorem}
\begin{proof}
	\Cref{thm:eps1} is applicable here with $\fval=\polylog(N)$.
	Let $\Pdist{n|m}$ and $\Qdist{m|n}$ be the probability distributions defined
	in~\eqref{eq:eps1_dist}.
	Given any $m$ and $n$, the value $A_{mn}$ can be computed in $\polylog(N)$ time.
	Since each row and column contains $\polylog(N)$ nonzero entries, which can be enumerated
	and computed in $\polylog(N)$ time, the sums
	$\sum_{n'} \abs{A_{mn'}}$ and $\sum_{m'} \abs{A_{m'n}}$
	can be computed in $\polylog(N)$ time.
	Thus condition~\eqref{cond:eps1_compute} of \cref{thm:eps1} is satisfied.

	For any given $m$, the distribution $\Pdist{n|m}$ has support of size $\polylog(N)$,
	the indices of which can be enumerated in $\polylog(N)$ time,
	and each individual probability can be computed in time $\polylog(N)$.
	Therefore, this distribution can be sampled from in time $\polylog(N)$.
	Similarly for $\Qdist{m|n}$,
	so conditions~\eqref{cond:eps1_samp_P} and~\eqref{cond:eps1_samp_Q}
	of \cref{thm:eps1} are satisfied and $A$ is
	$\epsp{\infnorm{A}^{1/p} \onenorm{A}^{1/q}}{\polylog(N)}$.
	Each row and column of $A$ has at most $\polylog(N)$ nonzero
	entries, each bounded by $\max_{mn} \{\abs{A_{mn}}\} = \polylog(N)$.
	It follows that $\infnorm{A} = \polylog(N)$ and $\onenorm{A} = \polylog(N)$,
	giving $\infnorm{A}^{1/p} \onenorm{A}^{1/q} = \polylog(N)$.
\end{proof}

A block diagonal matrix is $\epsp{\bval}{\fval}$ if each of its blocks is $\epsp{\bval}{\fval}$.
This is rather powerful in that it can be used to show the EPS property for operations on
subsystems, for controlled-unitaries, and for some rather exotic projectors.  This will be the
subject of the following theorem and corollaries.

\begin{theorem}[Block diagonal]
	\label{thm:block_diag_eps}
	For $r \in \{1,\dotsc,R\}$, let $A^{(r)}$ be an $\epsp{\bval_r}{\fval}$ matrix
	of dimension $M_r \times N_r$.
	Let $A$ be the block diagonal matrix $A=\oplus_r A^{(r)}$ of dimension
	$\sum_r M_r \times \sum_r N_r$.
	Suppose that it is possible in time $\bigomic(\fval)$ to convert between row/column indices of
	$A$ and the corresponding block indices (i.e.\ $m' \to (r,m)$ and $n' \to (s,n)$ and their
	inverse maps, with $A_{m'n'} = \delta_{rs} A^{(r)}_{mn}$).
	Then $A$ is $\epsp{\max_r \{ \bval_r \}}{\fval}$.
\end{theorem}
\begin{proof}
	Since $A^{(r)}$ is $\epsp{\bval_r}{\fval}$ for each $r$, there are
	$K_r$, $\alpha^{(r)}_{mnk}$,
	$\Pdistsub{r}{n,k|m}$, and $\Qdistsub{r}{m,k|n}$ satisfying \cref{def:eps},
	with $m \in \{1,\dotsc,M_r\}$, $n \in \{1,\dotsc,N_r\}$, and $k \in K_r$.
	Since we can convert between row/column indices of
	$A$ and the corresponding block indices in time $\bigomic(\fval)$, go ahead and label the
	indices of $A$ using block indices: $A_{(r,m),(s,n)} = \delta_{rs} A^{(r)}_{mn}$.
	Define $K = \cup_r K_r$ and
	\begin{align}
		\alpha_{(r,m),(s,n),k} = \begin{cases}
			\alpha^{(r)}_{mnk} &\mbox{if $r=s$ and $k \in K_r$}
			\\ 0 &\mbox{otherwise.}
		\end{cases}
	\end{align}
	This satisfies condition~\eqref{cond:eps_alphasum} of \cref{def:eps} since
	\begin{align}
		\sum_{k \in K} \alpha_{(r,m),(s,n),k}
		&= \delta_{rs} \sum_{k \in K_r} \alpha^{(r)}_{mnk}
		\\ &= \delta_{rs} A^{(r)}_{mn}
		\\ &= A_{(r,m),(s,n)}.
	\end{align}
	Define the probability distributions
	\begin{align}
		\Pdist{(s,n),k|(r,m)} &= \delta_{rs} \Pdistsub{r}{n,k|m},
		\\
		\Qdist{(r,m),k|(s,n)} &= \delta_{rs} \Qdistsub{s}{m,k|n}.
	\end{align}
	That $\alpha_{(r,m),(s,n),k}$, $\Pdist{(s,n),k|(r,m)}$, and $\Qdist{(r,m),k|(s,n)}$
	satisfy conditions~\eqref{cond:eps_samplr} and~\eqref{cond:eps_samprl}
	of \cref{def:eps} directly follows from the fact that
	$\alpha^{(r)}_{mnk}$, $\Pdistsub{r}{n,k|m}$, and $\Qdistsub{s}{m,k|n}$
	satisfy conditions~\eqref{cond:eps_samplr} and~\eqref{cond:eps_samprl}
	for all $r$.
	Condition~\eqref{cond:eps_cost} is satisfied as well, since
	\begin{align}
		\max_{(r,m),(s,n),k} \left\{
			\frac{\abs{ \alpha_{(r,m),(s,n),k} }}{
				\Pdist{(s,n),k|(r,m)}^{1/p}
				\Qdist{(r,m),k|(s,n)}^{1/q}
			}
		\right\}
		&=
		\max_r
		\max_{mnk} \left\{
			\frac{\abs{\alpha^{(r)}_{mnk}}}{
				\Pdistsub{r}{n,k|m}^{1/p}
				\Qdistsub{r}{m,k|n}^{1/q}
			}
		\right\}
		\\ &\le \max_r \{ \bval_r \}.
	\end{align}
\end{proof}

\begin{corollary}
	\label{thm:equal_blocks_eps}
	For $r \in \{1,\dotsc,R\}$, let $A^{(r)}$ be matrices on a space of dimension $N$.
	Suppose that each $A^{(r)}$ is $\epsp{\bval}{\fval}$ with $\fval=\Omega(\log^2(N))$.
	Then $A=\sum_{r=1}^R \ket{r}\bra{r} \ot A^{(r)}$, where the $\ket{r}$ are computational
	basis states, is $\epsp{\bval}{\fval}$.
\end{corollary}
\begin{proof}
	This is essentially a restatement of \cref{thm:block_diag_eps} for the case where all the
	$A^{(r)}$ are the same size.
	We require $\fval=\Omega(\log^2(N))$ because converting row or column indices of $A$ to
	indices of the blocks (as required for application of \cref{thm:block_diag_eps}) requires
	the operation of computing the quotient and remainder of division by $N$.
	The $\fval=\Omega(\log^2(N))$ requirement can be dropped if one is dealing with query
	complexity rather than computational complexity.
\end{proof}

\begin{corollary}
	\label{thm:low_rank_fourier_eps}
	Let $U$ denote a unitary matrix on $n$ qubits whose rows are CT states (e.g.\ the Fourier
	transform).
	Let $g : \{0, \dotsc, 2^{n-1}\} \to \{0, \dotsc, 2^{n-1}\}$ be a $\poly(n)$
	time computable function.
	Then the projector $\sum_{x=0}^{2^{n-1}} \ket{x}\bra{x} \ot U^\dag \ket{g(x)}\bra{g(x)} U$
	is $\epstwo{1}{\poly(n)}$.
	This projector corresponds to measuring half of the system in the computational basis to get
	measurement result $x$, measuring the other half of the system in the basis determined by
	$U$ to get $y$, and returning true if $y=g(x)$.
	The measurement depicted in \cref{fig:big_circuit} is of this form.
\end{corollary}
\begin{proof}
	Apply \cref{thm:equal_blocks_eps} with $A^{(x)} = U^\dag \ket{g(x)}\bra{g(x)} U$.
	$U^\dag \ket{g(x)}$ is a CT state, so by \cref{thm:dyads} $A^{(x)}$ is
	$\ehttwo{1}{\poly(n)}$ and therefore also $\epstwo{1}{\poly(n)}$.
\end{proof}

\begin{corollary}
	\label{thm:subsystem_eps}
	Let $I_{M_1}$ and $I_{M_2}$ denote the identity operator on spaces of dimension $M_1$ and
	$M_2$.
	Let $A$ be an $\epsp{\bval}{\fval}$ matrix of dimension $N_1 \times N_2$ with
	$\fval = \Omega(\log^2(M_1 M_2 N_1 N_2))$.
	Then $I_{M_1} \ot A \ot I_{M_2}$ is $\epsp{\bval}{\fval}$.
	This somewhat trivial result is important in that it allows the matrix to act on subsystems
	of the full state.
\end{corollary}
\begin{proof}
	Apply \cref{thm:block_diag_eps} with all of the $A^{(r)}$ blocks being
	equal.  We require $\fval = \Omega(\log^2(M_1 M_2 N_1 N_2))$ in order to allow converting row or
	column indices of $I_{M_1} \ot A \ot I_{M_2}$ to indices of $A$ in time $\bigomic(\fval)$.
\end{proof}

We now turn to the Grover reflection operation.  We will show this operator to be
$\epstwo{3}{\log(N)}$.
Since a unitary operator incurs a time expense of $\bval^4$ as per~\eqref{eq:circuitsim_cost},
each round of Grover's algorithm multiplies the simulation time by $3^4=81$.
This time is constant
in the number of qubits, but is exponential in the number of rounds.  Our technique is
therefore perfectly capable of simulating a small number of Grover reflections placed anywhere
in a circuit, but would perform very poorly, $\exp(\Theta(\sqrt{N}))$ time, if
applied to the $\Theta(\sqrt{N})$ rounds required by Grover's algorithm.

\begin{theorem}
	\label{thm:grover_eps}
	Let $\ket{+} = N^{-1/2} \sum_{i=0}^{N-1} \ket{i}$.
	The Grover reflection $I-2\ket{+}\bra{+}$ is $\epstwo{3}{\log(N)}$.
\end{theorem}
\begin{proof}
	Let $\delta_{mn}$ be the Kronecker delta.
	The identity operator can be seen to be $\epsp{1}{\log(N)}$, for any $p$ but in particular
	$p=2$, by simple inspection of \cref{def:eps} with $K=\{0\}$ and
	$\alpha_{mnk} = \Pdist{n,k|m} = \Qdist{m,k|n} = \delta_{mn}$.
	Note that we must take $\fval=\log(N)$ rather than $\fval=1$ since it takes
	$\Omega(\log(N))$ time to even write the indices $m$ and $n$, which are $\log(N)$ bits long.

	By \Cref{thm:dyads}, the projector $\ket{+}\bra{+}$ is $\ehttwo{1}{\log(N)}$,
	and therefore also $\epstwo{1}{\log(N)}$.
	By \cref{thm:basic_math} the operator $(-2)\ket{+}\bra{+}$ is $\epstwo{2}{\log(N)}$
	and by \cref{thm:eps_sumprodexp}\eqref{part:eps_sum_simple} the operator
	$I-2\ket{+}\bra{+}$ is $\epstwo{3}{\log(N)}$.
	One cannot do much better than $\bval=3$ since
	$\twonorm{I-2(\ket{+}\bra{+})} \to 3$ as $N \to \infty$.
\end{proof}

Next we show that the Haar wavelet transform on $n$ qubits, denoted $\haarn$, is
$\epstwo{\sqrt{n+1}}{n}$.
This is the lowest possible value of $\bval$, since $\twonorm{\abshaarn} = \sqrt{n+1}$.

\begin{definition}
	\label{def:haar}
	The \textit{Haar wavelet transform} on $n$ qubits is defined to be
	\begin{equation}
		\label{eq:haar}
		\haarn =
			\left( \ket{0} \bra{+} \right)^{\ot n} +
			\sum_{m=0}^{n-1}
			\left( \ket{0} \bra{+} \right)^{\ot m}
			\ot \ket{1} \bra{-} \ot I^{\ot n-m-1}.
	\end{equation}
	Note that there are other conventions that differ from this by a permutation in the
	computational basis.  Such permutations do not affect whether the Haar transform is
	$\epstwo{\sqrt{n+1}}{n}$.
\end{definition}

As an example, the Haar transform on three qubits is implemented by the circuit depicted in
\cref{fig:haarcircuit} and in the computational basis takes the form
\begin{equation}
	G_3 = \left[
		\begin{array}{cccccccc}
			\frac{1}{\sqrt{8}} & \frac{1}{\sqrt{8}} & \frac{1}{\sqrt{8}} & \frac{1}{\sqrt{8}} &
			\frac{1}{\sqrt{8}} & \frac{1}{\sqrt{8}} & \frac{1}{\sqrt{8}} & \frac{1}{\sqrt{8}}
			\\
			\frac{1}{\sqrt{8}} & \frac{-1}{\sqrt{8}} & \frac{1}{\sqrt{8}} & \frac{-1}{\sqrt{8}} &
			\frac{1}{\sqrt{8}} & \frac{-1}{\sqrt{8}} & \frac{1}{\sqrt{8}} & \frac{-1}{\sqrt{8}}
			\\
			\frac{1}{\sqrt{4}} & 0 & \frac{-1}{\sqrt{4}} & 0 &
			\frac{1}{\sqrt{4}} & 0 & \frac{-1}{\sqrt{4}} & 0
			\\
			0 & \frac{1}{\sqrt{4}} & 0 & \frac{-1}{\sqrt{4}} &
			0 & \frac{1}{\sqrt{4}} & 0 & \frac{-1}{\sqrt{4}}
			\\
			\frac{1}{\sqrt{2}} & 0 & 0 & 0 & \frac{-1}{\sqrt{2}} & 0 & 0 & 0 \\
			0 & \frac{1}{\sqrt{2}} & 0 & 0 & 0 & \frac{-1}{\sqrt{2}} & 0 & 0 \\
			0 & 0 & \frac{1}{\sqrt{2}} & 0 & 0 & 0 & \frac{-1}{\sqrt{2}} & 0 \\
			0 & 0 & 0 & \frac{1}{\sqrt{2}} & 0 & 0 & 0 & \frac{-1}{\sqrt{2}}
		\end{array}
	\right].
\end{equation}

\begin{figure}
	\begin{center}
		\includegraphics{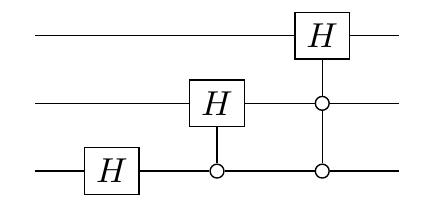}
	\end{center}
	\caption{
		This circuit implements the Haar transform of \cref{def:haar}, on three
		qubits~\cite{arxiv:quant-ph/9702028}.
		The gates in this circuit are controlled-Hadamard gates, and the open circles denote
		that the Hadamard gates are active when all of the controls are in the $\ket{0}$ state.
	}
	\label{fig:haarcircuit}
\end{figure}

\begin{theorem}
	\label{thm:haar_eps}
	The Haar transform on $n$ qubits is $\epstwo{\sqrt{n+1}}{n}$.
\end{theorem}
\begin{proof}
	Since we are dealing with spaces of dimension $2^n$, made of qubits, it will be convenient to
	index the space using bit strings $\xvec, \yvec \in \{0, 1\}^n$.
	We will denote the corresponding basis vectors using the notation
	$\ket{\xvec} = \ket{x_0} \ot \dotsb \ot \ket{x_{n-1}}$.
	To avoid notational confusion regarding subscripts, define $A=\haarn$.
	Then $A_{\xvec \yvec}$ refers to the matrix element $\braopket{\xvec}{\haarn}{\yvec}$.

	Take $K = \{0\}$ (i.e.\ don't make use of the index $k$), and set
	$\alpha_{\xvec \yvec k} = A_{\xvec \yvec}$.  This satisfies
	condition~\eqref{cond:eps_alphasum} of \cref{def:eps} trivially.
	Take the probability distributions $\Pdist{\yvec|\xvec}$ and $\Qdist{\xvec|\yvec}$
	to be uniform over the nonzero elements of the given row or column of
	$A_{\xvec \yvec}$.  Despite the apparent simplicity of this choice, analysis will be
	tedious due to the somewhat complicated definition of $A$.
	These probability distributions can be expressed as follows.
	\begin{align}
		\label{eq:haar_P}
		\Pdist{\yvec|\xvec}
		&= \frac{1}{2^n} [\xvec = 0] + \sum_{m=0}^{n-1}
			\frac{1}{2^{m+1}}
			\left( \prod_{i=0}^{m-1} [ x_i=0 ] \right)
			[x_{m}=1]
			\left( \prod_{i=m+1}^{n-1} [ y_i=x_i ] \right)
	\\
		\label{eq:haar_Q}
		\Qdist{\xvec|\yvec}
		&= \frac{1}{n+1} \left\{
			[\xvec = 0] + \sum_{m=0}^{n-1}
			\left( \prod_{i=0}^{m-1} [ x_i=0 ] \right)
			[x_{m}=1]
			\left( \prod_{i=m+1}^{n-1} [ x_i=y_i ] \right)
		\right\}
	\end{align}
	These can be sampled from in time $\bigomic(n)$.
	Consider first $\Pdist{\yvec|\xvec}$.  Given an $\xvec$,
	only a single one of the $n+1$ terms of~\eqref{eq:haar_P} doesn't vanish, and this term can
	be identified in time $\bigomic(n)$, by searching for the smallest (if any) $m$ for which
	$x_m=1$.  The nonvanishing
	term defines the value of $y_i$ for some of
	the $i$, and gives a uniform distribution for each of the remaining $y_i$.
	For $\Qdist{\xvec|\yvec}$, each of the $n+1$ terms of~\eqref{eq:haar_Q}
	is nonvanishing for a single value of $\xvec$, and each occurs with equal probability.
	Therefore, sampling from $\Qdist{\xvec|\yvec}$ is accomplished by
	drawing from a uniform distribution over $n+1$ possibilities.

	To satisfy conditions~\eqref{cond:eps_samplr} and~\eqref{cond:eps_samprl}
	of \cref{def:eps} we must also show that
	$A_{\xvec\yvec} / \Pdist{\yvec|\xvec}$ and
	$A_{\xvec\yvec} / \Qdist{\xvec|\yvec}$
	can be computed in time $\bigomic(n)$.
	We begin by writing an expression for $A_{\xvec\yvec}$.
	In the equations below, square brackets denote the Iverson bracket, which takes a value
	of 1 if the enclosed expression is true and 0 otherwise.
	\begin{align}
		A_{\xvec\yvec} &= \bra{\xvec} \left(
			\left( \ket{0} \bra{+} \right)^{\ot n} +
			\sum_{m=0}^{n-1}
			\left( \ket{0} \bra{+} \right)^{\ot m}
			\ot \ket{1} \bra{-} \ot I^{\ot n-m-1}
		\right) \ket{\yvec}
		\\
		\label{eq:haar_element}
		&= \frac{1}{\sqrt{2^n}} [\xvec = 0] + \sum_{m=0}^{n-1}
			\frac{(-1)^{y_m}}{\sqrt{2^{m+1}}}
			\left( \prod_{i=0}^{m-1} [ x_i=0 ] \right)
			[x_m=1]
			\left( \prod_{i=m+1}^{n-1} [ x_i=y_i ] \right),
	\end{align}
	Since only a single term for each of~\eqref{eq:haar_P}, \eqref{eq:haar_Q},
	and~\eqref{eq:haar_element} is nonvanishing for each given $\xvec, \yvec$ pair,
	we can divide these equations term-by-term to get
	\begin{align}
		\label{eq:haar_A_over_P}
		\frac{A_{\xvec\yvec}}{\Pdist{\yvec|\xvec}} &=
			\sqrt{2^n} [\xvec = 0] + \sum_{m=0}^{n-1}
			(-1)^{y_m}
			\sqrt{2^{m+1}}
			\left( \prod_{i=0}^{m-1} [ x_i=0 ] \right)
			[x_m=1]
			\left( \prod_{i=m+1}^{n-1} [ x_i=y_i ] \right),
		\\
		\label{eq:haar_A_over_Q}
		\frac{A_{\xvec\yvec}}{\Qdist{\xvec|\yvec}} &=
			(n+1) \left\{
				\frac{1}{\sqrt{2^n}} [\xvec = 0] + \sum_{m=0}^{n-1}
				\frac{(-1)^{y_m}}{\sqrt{2^{m+1}}}
				\left( \prod_{i=0}^{m-1} [ x_i=0 ] \right)
				[x_m=1]
				\left( \prod_{i=m+1}^{n-1} [ x_i=y_i ] \right)
			\right\}.
	\end{align}
	At most a single term of these expressions is nonvanishing for each given $\xvec, \yvec$
	pair, and this term can be identified in time $\bigomic(n)$ by searching for the smallest
	(if any) $m$ for which $x_m=1$.  The value of nonvanishing terms is of the form
	$\pm \sqrt{2^s}$ or $\pm (n+1)/\sqrt{2^s}$ for some $s$, and this can be computed in
	$\bigomic(1)$ time.

	That condition~\eqref{cond:eps_cost} of \cref{def:eps} is satisfied is checked
	directly,
	\begin{align}
		\max_{\xvec \yvec} \left\{
			\frac{\abs{A_{\xvec \yvec}}}
			{\Pdist{\yvec|\xvec}^{1/2} \Qdist{\xvec|\yvec}^{1/2}}
		\right\}
		&=
		\max_{\xvec \yvec} \left\{
			\left(
			\frac{\abs{A_{\xvec \yvec}}}{\Pdist{\yvec|\xvec}}
			\frac{\abs{A_{\xvec \yvec}}}{\Qdist{\xvec|\yvec}}
			\right)^{1/2}
		\right\}
		\\ \label{eq:haar_cost_termwise} &=
		\max_{\xvec \yvec} \left\{
			(n+1)^{1/2}
		\right\}
		\\ &= \sqrt{n+1},
	\end{align}
	where~\eqref{eq:haar_cost_termwise} follows from the fact that only a single term from each
	of~\eqref{eq:haar_A_over_P} and~\eqref{eq:haar_A_over_Q} is nonvanishing, so they can be
	multiplied term-by-term.
\end{proof}

\bibliographystyle{apsrev}
\bibliography{everything}

\end{document}